\documentclass{article}%
\usepackage{amssymb}
\usepackage{amsfonts}
\usepackage{amsmath}
\usepackage[numbers,sort&compress]{natbib}
\usepackage{graphicx}%
\providecommand{\U}[1]{\protect\rule{.1in}{.1in}}
\newtheorem{theorem}{Theorem}

\newtheorem{corollary}[theorem]{Corollary}

\newtheorem{example}{Example}

\newtheorem{lemma}{Lemma}

\newenvironment{proof}[1][Proof]{\noindent\textbf{#1.} }{\ \rule{0.5em}{0.5em}}
\numberwithin{equation}{section}
\begin{document}

\title{Systematic construction of non-autonomous Hamiltonian equations of
Painlev\'{e}-type. II. Isomonodromic Lax representation }
\author{Maciej B\l aszak\\Faculty of Physics, Department of Mathematical Physics and Computer Modelling,\\A. Mickiewicz University, Uniwersytetu Pozna\'nskiego 2, 61-614 Pozna\'{n}, Poland\\\texttt{blaszakm@amu.edu.pl}
\and Ziemowit Doma\'nski\\Institute of Mathematics, Pozna{\'{n}} University of Technology\\Piotrowo 3A, 60-965 Pozna{\'{n}}, Poland\\\texttt{ziemowit.domanski@put.poznan.pl}
\and Krzysztof Marciniak\\Department of Science and Technology \\Campus Norrk\"{o}ping, Link\"{o}ping University\\601-74 Norrk\"{o}ping, Sweden\\\texttt{krzma@itn.liu.se}}
\maketitle

\begin{abstract}
This is the second article in a suite of articles investigating relations
between St\"{a}ckel-type systems and Painlev\'{e}-type systems. In this
article we construct isomonodromic Lax representations for Painlev\'{e}-type
systems found in the previous paper \cite{part1} by Frobenius integrable
deformations of St\"{a}ckel-type systems. We first construct isomonodromic Lax
representations for Painlev\'{e}-type systems in the so called magnetic
representation and then, using a multitime-dependent canonical transformation,
we also construct isomonodromic Lax representations for Painlev\'{e}-type
systems in the non-magnetic representation. Thus, we prove that the Frobenius
integrable systems constructed in Part I are indeed of Painlev\'{e}-type. We
also present isomonodromic Lax representations for all one-, two- and
three-dimensional Painlev\'{e}-type systems originating in our scheme. Based
on these results we propose complete hierarchies of $P_{I}-P_{IV}$ that follow
from our construction.

\end{abstract}



\section{Introduction}

This is the second article in the suit of articles investigating a systematic
way of constructing Painlev\'{e}-type systems from an appropriate
St\"{a}ckel-type systems. In the previous paper (Part I, \cite{part1}) we have
constructed multi-parameter families of Frobenius integrable non-autonomous
Hamiltonian systems with arbitrary number of degrees of freedom. Each of these
families was written in two different representations (two different
coordinate systems), that we called an ordinary one and a magnetic one,
respectively, connected by a multi-time dependent canonical transformation
\cite{Iwasaki}.

In this paper (Part II) we construct the isomonodromic Lax equations for both
representations of these systems, thus proving that Frobenius integrable
systems constructed in Part I are indeed Painlev\'{e}-type systems. Based on
these results we propose complete hierarchies of the celebrated Painlev\'{e}
equations $P_{I}$, $P_{II}$, $P_{III}$ and $P_{IV}\,$. By a complete hierarchy
we mean that \emph{for arbitrary }$n$\emph{ we construct }$n$\emph{
Painlev\'{e}-type flows satisfying the Frobenius integrability condition}, in
contrast to the literature where one encounters only Painlev\'{e} hierarchies
with one flow for each $n$, see for example \cite{S2,J,K,Sakka,Harnad}.
Completeness here means thus analogy with autonomous integrable systems and
Liouville theorem - any Liouville integrable system with $n$ degrees of
freedom has $n$ (functionaly independent) commuting Hamiltonian flows.

This paper has the following content. In Section \ref{S2} we present an
isospectral Lax formulation for Liouville integrable systems of
St\"{a}ckel-type. This Lax formulation is parametrized by $n+1$ arbitrary
functions and as such is very general. It encompasses both the ordinary and
the magnetic representations of the corresponding St\"{a}ckel-type systems. In
Section \ref{S3} we remind the procedure of deforming of St\"{a}ckel-type
systems into the magnetic representation of Frobenius integrable
non-autonomous systems (magnetic representation of Painlev\'{e}-type systems)
and then, in Section \ref{S3a} we construct their isomonodromic Lax
representation. In Section \ref{S4} we apply our theory to obtain
isomonodromic Lax representations for all one-, two- and three-dimensional
Painlev\'{e}-type systems in the magnetic representation. In Section \ref{S5}
we construct, by using the multitime-dependent canonical transformation found
in Part I, the corresponding isomonodromic Lax representations of the obtained
Painlev\'{e}-type systems in the ordinary (non-magnetic) representation. In
Section \ref{S6} we present isomonodromic Lax representations for all one-,
two- and three-dimensional Painlev\'{e}-type systems in the non-magnetic
representation, which leads to our novel proposal of complete $P_{I}-P_{IV}$
hierarchies. Some technical proofs are moved to Appendix.

\section{Isospectral Lax representation for Liouville integrable systems of
St\"{a}ckel-type\label{S2}}

Consider the St\"{a}ckel system (separable system) generated by the following
hyperelliptic spectral curve (separation curve):
\begin{equation}
\sum_{r=1}^{n}h_{r}x^{n-r}=\frac{1}{2}f(x)y^{2}-\varphi(x)y-\sigma
(x)\equiv\Phi(x,y) \label{1}%
\end{equation}
on an $(x,y)$-plane, where $\sigma$, $\varphi$ and $f$ are (arbitrary so far)
smooth functions of one variable.\ By taking $n$ copies of (\ref{1}) at points
$(x,y)=(\lambda_{i},\mu_{i})$, $i=1,\dotsc,n$, we obtain a system of $n$
linear equations (separation relations) for $h_{r}:$
\begin{equation}
\sum_{r=1}^{n}h_{r}\lambda_{i}^{n-r}=\frac{1}{2}f(\lambda_{i})\mu_{i}%
^{2}-\varphi(\lambda_{i})\mu_{i}-\sigma(\lambda_{i})\equiv\Phi(\lambda_{i}%
,\mu_{i})\text{, \ \ }i=1,\ldots,n\text{.} \label{systemik}%
\end{equation}
Solving this system (by inverting of the Vandermonde matrix $\lambda_{i}%
^{n-r}$) yields $n$ functions (Hamiltonians)
\begin{equation}
h_{r}=E_{r}+M_{r}^{\varphi}+V_{r}^{\sigma},\quad r=1,\dotsc,n, \label{1aa}%
\end{equation}
depending on $2n$ variables $(\lambda,\mu)=(\lambda_{1},\ldots,\lambda_{n}%
,\mu_{1},\ldots,\mu_{n})$. We will from now on assume that these variables
parametrize a $2n$-dimensional smooth manifold $\mathcal{M}=T^{\ast}Q$ in such
a way that $\lambda_{i}$ are coordinates on an $n$-dimensional configurational
manifold $Q$ and $\mu_{i}$ are fiber coordinates (momenta) on $T^{\ast}Q$.
Explicitly, we obtain
\[
E_{r}=\frac{1}{2}\mu^{T}K_{r}G\mu\text{, \ \ }r=1,\dotsc,n,
\]
where $\mu=(\mu_{1},\ldots,\mu_{n})^{T}$ and where the $n\times n\,\ \lambda
$-dependent matrix $G$ can be interpreted as a contravariant metric tensor on
$Q$ (thus turning $Q$ into a Riemannian manifold). The metric $G$ is flat if
$f$ is a polynomial of order less then $n+1$ and of constant curvature if $f$
is a polynomial of order $n+1$. The matrices $K_{r}$ ($K_{1}=\operatorname{Id}%
$) can be shown to be $(1,1)$-Killing tensors for the metric $G$ (for any
given $f$) \cite{blasz2005,blasz2007,Book}. The functions $E_{r}$ on
$\mathcal{M}=T^{\ast}Q$ are called \emph{geodesic St\"{a}ckel Hamiltonians}.

Further, $V_{r}^{\sigma}$ are functions on $Q$ that we call separable
potentials. In case that $\sigma$ is a Laurent sum $\sigma(x)=\sum_{\alpha
}\varepsilon_{\alpha}x^{\alpha}$, $V_{r}^{\sigma}$ will be the corresponding
Laurent sum $V_{r}^{\sigma}(\lambda)=\sum_{\alpha}\varepsilon_{\alpha}%
V_{r}^{(\alpha)}$of basic separable potentials $V_{r}^{(\alpha)},$ that can be
constructed by the formula \cite{blasz2011}
\begin{equation}
V^{(\alpha)}=R^{\alpha}V^{(0)},\ \ \ \ V^{(\alpha)}=(V_{1}^{(\alpha
)},...,V_{n}^{(\alpha)})^{T}, \label{5d}%
\end{equation}
where
\begin{equation}
R=\left(
\begin{array}
[c]{cccc}%
-\rho_{1} & 1 & 0 & 0\\
\vdots & 0 & \ddots & 0\\
\vdots & 0 & 0 & 1\\
-\rho_{n} & 0 & 0 & 0
\end{array}
\right)  \label{6a}%
\end{equation}
with $V^{(0)}=(0,0,...,-1)^{T}\,\ $and with $\rho_{k}(\lambda)=(-1)^{k}%
s_{k}(\lambda)$, $k=1,\ldots,n$, where $s_{k}\left(  \lambda\right)  $ is the
elementary symmetric polynomial in $\lambda_{i}$ of degree $k$ (so that
$s_{1}=\lambda_{1}+\ldots+\lambda_{n}$ and so on). Further, $M_{r}^{\varphi}$
are some, in general complicated, functions on $\mathcal{M}$. In case that
$\varphi$ is a Laurent sum $\varphi(x)=\sum_{\gamma}\varepsilon_{\gamma
}x^{\gamma}$, $M_{r}^{\varphi}$ will be the corresponding Laurent sum
$M_{r}^{\varphi}(\lambda,\mu)=\sum_{\gamma}\varepsilon_{\gamma}M_{r}%
^{(\gamma)}$ of basic separable magnetic\emph{ }potentials $M_{r}^{(\gamma)}$.
They have the explicit form
\begin{equation}
M_{r}^{(\gamma)}=\sum_{i=1}^{n}\frac{\partial\rho_{r}}{\partial\lambda_{i}%
}\frac{\lambda_{i}^{\gamma}\mu_{i}}{\Delta_{i}},\ \Delta_{i}=\prod
\limits_{j\neq i}(\lambda_{i}-\lambda_{j}) \label{5dd}%
\end{equation}
(see Part I \cite{part1} for more details) and are called magnetic since they
depend linearly on momenta $\mu_{i}$.

Assume now that $(\lambda,\mu)=(\lambda_{1},\ldots,\lambda_{n},\mu_{1}%
,\ldots,\mu_{n})$ are Darboux (canonical) coordinates for a time-independent
Poisson tensor $\pi$ on $\mathcal{M}$ (so that $\{\mu_{i},\lambda_{j}\}_{\pi
}=\delta_{ij}$,$\ \{\lambda_{i},\lambda_{j}\}_{\pi}=\{\mu_{i},\mu_{j}\}_{\pi
}=0$, $i,j=1,\dots,n$). Then, the Hamiltonians $h_{r}$ generate $n$ separable
autonomous evolution equations (Hamiltonian \emph{flows})
\begin{equation}
\frac{d\xi}{dt_{r}}=X_{r}(\xi)\equiv\pi dh_{r}(\xi),\quad r=1,\ldots,n,
\label{1b}%
\end{equation}
where $\xi\in\mathcal{M}$ and where $X_{r}$ are the related autonomous
Hamiltonian vector fields $X_{r}=\pi dh_{r}$ (autonomous means in this context
that $X_{r}$ do not depend explicitly on time variables $t_{s}$). From the
very construction it follows that
\[
\left\{  h_{r},h_{s}\right\}  =0\text{ and thus }\left[  X_{r},X_{s}\right]
=0,\quad r,s=1,\ldots,n
\]
(and hence the set of $n$ Hamiltonian systems (\ref{1b}) is a Liouville
integrable system, we will refer to this set as \emph{St\"{a}ckel system}) and
moreover that the canonical coordinates $(\lambda,\mu)$ are separation
coordinates for all the flows (\ref{1b}). Since (\ref{1b}) is autonomous, it
is also Frobenius integrable, i.e. the set of $n$ equations in (\ref{1b})\ has
a common, unique solution $\xi(t_{1},\ldots,t_{n},\xi_{0})$ through each point
$\xi_{0}\in\mathcal{M},$ depending in general on all the evolution parameters
$t_{r}$.

In this section we derive isospectral Lax representations for each Hamiltonian
flow in (\ref{1b}). In literature the reader can find Lax representation for
flows related to various subcases of separation curves from the family
(\ref{1}) (see for example \cite{Mum,En,Ts,Pol} and references therein) but
our general construction is new. Let us consider the following family of
\emph{non-equivalent} (in the sense of lack of similarity transformation
between them) Lax matrices $L\in\mathfrak{sl}(2,\mathbb{R})$, parametrized by
an arbitrary non-zero smooth function $g$:
\begin{equation}
L(x,\xi)=%
\begin{pmatrix}
v(x)-\frac{g(x)\varphi(x)}{f(x)} & u(x)\\
w(x) & -v(x)+\frac{g(x)\varphi(x)}{f(x)}%
\end{pmatrix}
. \label{L1}%
\end{equation}
where
\begin{equation}
u(x)=\prod_{k=1}^{n}(x-\lambda_{k})=\sum_{k=0}^{n}\rho_{k}x^{n-k},\quad
\rho_{0}\equiv1, \label{L2}%
\end{equation}
$v$ is a polynomial of order $n-1$ such that $v(\lambda_{i})=g(\lambda_{i}%
)\mu_{i}$, so that it takes the form
\begin{equation}
v(x)=\sum_{i=1}^{n}g(\lambda_{i})\mu_{i}\prod_{k\neq i}\frac{x-\lambda_{k}%
}{\lambda_{i}-\lambda_{k}}=-\sum_{k=1}^{n}\left[  \sum_{i=1}^{n}\frac
{\partial\rho_{k}}{\partial\lambda_{i}}\frac{g(\lambda_{i})\mu_{i}}{\Delta
_{i}}\right]  x^{n-k}, \label{L3}%
\end{equation}
while $w$ is determined by the right hand side of the spectral curve (\ref{1})
through
\begin{equation}
w(x)=-2\frac{g^{2}(x)}{f(x)}\left[  \frac{\Phi(x,v(x)/g(x))}{u(x)}\right]
_{+}\text{.} \label{L4}%
\end{equation}
The operation $\left[  \frac{\cdot}{\cdot}\right]  _{+}\,\ $is defined as
follows: for an arbitrary smooth function $b$ and an arbitrary polynomial $a$,
$\left[  \frac{b(x)}{a(x)}\right]  _{+}$ is a smooth function defined uniquely
through
\begin{equation}
\frac{b(x)}{a(x)}=\left[  \frac{b(x)}{a(x)}\right]  _{+}+\frac{r(x)}{a(x)}
\label{L4a}%
\end{equation}
where $r$ is a polynomial of degree $\deg r<\deg a$. In case that $b$ is a
polynomial (Laurent polynomial) then $\left[  \frac{b(x)}{a(x)}\right]  _{+}%
$is a polynomial (Laurent polynomial) part of the division of the polynomial
(Laurent polynomial) $b$ by the polynomial $a$ (see \cite{Blaszak2019} for the
details of this construction).

Note that in our notation $u=u(x,\xi)$, $v=v(x,\xi)$ and $w=w(x,\xi)$ are
functions depending not only on the spectral parameter $x$ but also on the
point $\xi$ on $\mathcal{M}$, but in the sequel we will omit this dependence
on $\xi$ in $u$, $v$ and $w$ and we will write $u(x)$, $v(x)$ and $w(x)$ in
order to shorten the notation.

The Lax matrix $L$ has the important property: it can be used to reconstruct
the separation curve, as the next Lemma shows.

\begin{lemma}
\label{1l}The element $w$ in $L$ can be written as
\begin{equation}
w(x)=-\frac{v^{2}(x)}{u(x)}+\frac{2g(x)\varphi(x)v(x)}{f(x)u(x)}+\frac
{2g^{2}(x)\sigma(x)}{f(x)u(x)}+\frac{2g^{2}(x)}{f(x)u(x)}\sum_{k=1}^{n}%
h_{k}x^{n-k}. \label{L5}%
\end{equation}

\end{lemma}

\begin{proof}
The following identity with respect to $x,\lambda,\mu$
\[
\sum_{k=1}^{n}h_{k}x^{n-k}\equiv\Phi(x,v(x)/g(x))\bmod u(x)=\Phi
(x,v(x)/g(x))-u(x)\left[  \frac{\Phi(x,v(x)/g(x))}{u(x)}\right]  _{+}.
\]
was proved in \cite[Lemma~4.1]{Blaszak2019}. Taking $\Phi(x,y)$ in (\ref{1}),
that is $\Phi(x,y)=\frac{1}{2}f(x)y^{2}-\varphi(x)y-\sigma(x)$, and using
(\ref{L4}), this identity can be rewritten in the form
\[
\sum_{k=1}^{n}h_{k}x^{n-k}=\frac{1}{2}f(x)\frac{v^{2}(x)}{g^{2}(x)}%
-\varphi(x)\frac{v(x)}{g(x)}-\sigma(x)+\frac{f(x)}{2g^{2}(x)}u(x)w(x).
\]
Solving the above identity with respect to $w$ yields (\ref{L5}).
\end{proof}

\begin{theorem}
\label{1t}For an arbitrary smooth function $g(x)$, the separation curve
(\ref{1}) is reconstructed from the Lax matrix (\ref{L1}) in the sense that
\begin{equation}
\det\left[  L(x,\xi)-g(x)\left(  y-\frac{\varphi(x)}{f(x)}\right)
\operatorname{Id}\right]  =2\frac{g^{2}(x)}{f(x)}\left(  \Phi(x,y)-\sum
_{r=1}^{n}h_{r}x^{n-1}\right)  . \label{eq:2}%
\end{equation}

\end{theorem}

\begin{proof}
We prove this theorem by direct calculation:
\begin{align*}
\det\left[  L(x,\xi)-g(x)\left(  y-\frac{\varphi(x)}{f(x)}\right)
\operatorname{Id}\right]   &  =\det\left(
\begin{matrix}
v(x)-g(x)y & u(x)\\
w(x) & -v(x)-g(x)y+\frac{2g(x)\varphi(x)}{f(x)}%
\end{matrix}
\right) \\
&  =-v^{2}(x)+g^{2}(x)y^{2}+\frac{2g(x)\varphi(x)v(x)}{f(x)}-\frac
{2g^{2}(x)\varphi(x)y}{f(x)}-u(x)w(x)\\
&  \overset{(\ref{1})}{=}g^{2}(x)y^{2}-\frac{2g^{2}(x)\varphi(x)y}{f(x)}%
-\frac{2g^{2}(x)\sigma(x)}{f(x)}-\frac{2g^{2}(x)}{f(x)}\sum_{k=1}^{n}%
h_{k}x^{n-k}\\
&  =\frac{2g^{2}(x)}{f(x)}\left(  \frac{1}{2}f(x)y^{2}-\varphi(x)y-\sigma
(x)-\sum_{k=1}^{n}h_{k}x^{n-k}\right)  .
\end{align*}

\end{proof}

We present now the main theorem of this section.

\begin{theorem}
\label{2t}Each Hamiltonian flow $\frac{d\xi}{dt_{r}}=X_{r}$ in the St\"{a}ckel
system (\ref{1b}) has the isospectral Lax representation
\begin{equation}
\frac{d}{dt_{r}}L(x,\xi)=[U_{r}(x,\xi),L(x,\xi)] \label{eq:3}%
\end{equation}
where
\begin{equation}
\begin{gathered} U_{r}(x,\xi)=g_{r}(x)L(x,\xi)+\frac{1}{2u(x)}P_{r}(x), \\ P_{r}(x)=\begin{pmatrix} \{u(x),h_{r}\} & 0\\ -2\{v(x),h_{r}\} & -\{u(x),h_{r}\} \end{pmatrix} \end{gathered} \label{eq:4}%
\end{equation}
and with $L$ given by (\ref{L1})-(\ref{L4}), $g(x)$ in $L$ being an arbitrary
smooth function of one argument, and where $g_{r}(x)$ is an arbitrary smooth
function of $x$ and of $(\lambda,\mu)$.
\end{theorem}

Thus, each flow $\frac{d\xi}{dt_{r}}=X_{r}$ in the St\"{a}ckel system
(\ref{1b}) has the isospectral Lax representation parametrized by two
arbitrary functions: $g$ (common for all flows) and $g_{r}$ (that can be
chosen uniquely for each flow).

\begin{proof}
To prove (\ref{eq:3}) note that
\begin{align*}
&  [U_{k}(x),L(x)]=\frac{1}{2u(x,\xi)}%
\begin{pmatrix}
\{u(x),h_{k}\} & 0\\
-2\{v(x),h_{k}\} & -\{u(x),H_{k}\}
\end{pmatrix}%
\begin{pmatrix}
v(x)-\tfrac{g(x)\varphi(x)}{f(x)} & u(x)\\
w(x) & -v(x)+\tfrac{g(x)\varphi(x)}{f(x)}%
\end{pmatrix}
\\
& \\
&  \qquad{}-\frac{1}{2u(x)}%
\begin{pmatrix}
v(x)-\tfrac{g(x)\varphi(x)}{f(x)} & u(x)\\
w(x) & -v(x)+\tfrac{g(x)\varphi(x)}{f(x)}%
\end{pmatrix}%
\begin{pmatrix}
\{u(x),h_{k}\} & 0\\
-2\{v(x),h_{k}\} & -\{u(x),h_{k}\}
\end{pmatrix}
\\
& \\
&  \quad=\frac{1}{2u(x)}%
\begin{pmatrix}
\left(  v(x)-\tfrac{g(x)\varphi(x)}{f(x)}\right)  \{u(x),h_{k}\} &
u(x)\{u(x),h_{k}\}\\%
\begin{array}
{c}%
\\
-2\left(  v(x)-\tfrac{g(x)\varphi(x)}{f(x)}\right)  \{v(x),h_{k}\}\\
-w(x)\{u(x),h_{k}\}
\end{array}
&
\begin{array}
{c}%
\\
\left(  v(x)-\tfrac{g(x)\varphi(x)}{f(x)}\right)  \{u(x),h_{k}\}\\
-2u(x)\{v(x),h_{k}\}
\end{array}
\end{pmatrix}
\\
& \\
&  \qquad{}-\frac{1}{2u(x)}%
\begin{pmatrix}
\left(  v(x)-\tfrac{g(x)\varphi(x)}{f(x)}\right)  \{u(x),h_{k}%
\}-2u(x)\{v(x),h_{k}\} & -u(x)\{u(x),h_{k}\}\\
w(x)\{u(x),h_{k}\}+2\left(  v(x)-\tfrac{g(x)\varphi(x)}{f(x)}\right)
\{v(x),h_{k}\} & \left(  v(x)-\tfrac{g(x)\varphi(x)}{f(x)}\right)
\{u(x),h_{k}\}
\end{pmatrix}
\\
& \\
&  \quad=%
\begin{pmatrix}
\{v(x),h_{k}\} & \{u(x),h_{k}\}\\
-\tfrac{w(x)}{u(x)}\{u(x),h_{k}\}-2\left(  \tfrac{v(x)}{u(x)}-\tfrac
{g(x)\varphi(x)}{f(x)u(x)}\right)  \{v(x),h_{k}\} & -\{v(x),h_{k}\}
\end{pmatrix}
.
\end{align*}
Since
\[
\{u(x)w(x),h_{k}\}=u(x)\{w(x),h_{k}\}+w(x)\{u(x),h_{k}\}
\]
and on the other hand
\begin{align*}
\{u(x)w(x),h_{k}\}  &  =-\{v^{2}(x),h_{k}\}+\frac{2g(x)\varphi(x)}%
{f(x)}\{v(x),h_{k}\}+\frac{2g^{2}(x)}{f(x)}\sum_{l=1}^{n}\{h_{l}%
,h_{k}\}x^{n-l}\\
&  =-2v(x)\{v(x),h_{k}\}+\frac{2g(x)\varphi(x)}{f(x)}\{v(x),h_{k}\}
\end{align*}
we get that
\[
\{w(x),h_{k}\}=-\frac{w(x)}{u(x)}\{u(x),h_{k}\}-2\left(  \frac{v(x)}%
{u(x)}-\frac{g(x)\varphi(x)}{f(x)u(x)}\right)  \{v(x),h_{k}\}.
\]
Hence
\[
\lbrack U_{k}(x),L(x)]=%
\begin{pmatrix}
\{v(x),h_{k}\} & \{u(x),h_{k}\}\\
\{w(x),h_{k}\} & -\{v(x),h_{k}\}
\end{pmatrix}
=\frac{d}{dt_{k}}L(x).
\]

\end{proof}

In this article we will use two important specifications of the functions
$g_{r}(x)$. These specifications are used when discussing Painlev\'{e}-type
systems in next sections, as for Painlev\'{e}-type systems we have to specify
$g$ and $g_{r}$ in a very concrete way.

\begin{lemma}
\label{2l}If
\begin{equation}
g_{r}(x)=\frac{1}{2u(x)}\frac{f(x)}{g(x)}\left[  \frac{u(x)}{x^{n+1-r}%
}\right]  _{+} \label{1.1}%
\end{equation}
then the auxiliary matrices $U_{r}(x)$ are of the form
\begin{equation}
U_{r}(x)=\left[  \frac{B_{r}(x)}{u(x)}\right]  _{+},\ \ \ \ \ \ \ \ B_{r}%
(x)=\frac{1}{2}\frac{f(x)}{g(x)}\left[  \frac{u(x)}{x^{n+1-r}}\right]
_{+}L(x) \label{1.2}%
\end{equation}
and if \
\begin{equation}
g_{r}(x)=-\frac{1}{2u(x)}\frac{f(x)}{g(x)}\left[  \frac{u(x)}{x^{n+1-r}%
}\right]  _{-}, \label{1.4}%
\end{equation}
where
\begin{equation}
\left[  \frac{u(x)}{x^{n+1-r}}\right]  _{-}=\frac{u(x)}{x^{n+1-r}}-\left[
\frac{u(x)}{x^{n+1-r}}\right]  _{+}, \label{9}%
\end{equation}
then the auxiliary matrices $U_{r}(x)$ are
\begin{equation}
U_{r}(x)=\left[  \frac{B_{r}(x)}{u(x)}\right]  _{+},\ \ \ \ \ \ \ B_{r}%
(x)=-\frac{1}{2}\frac{f(x)}{g(x)}\left[  \frac{u(x)}{x^{n+1-r}}\right]
_{-}L(x). \label{1.5}%
\end{equation}

\end{lemma}

\begin{proof}
Indeed, from \cite[Lemma~4.4]{Blaszak2019} we get
\begin{align*}
\{u(x),h_{r}\}  &  =\left(  -\frac{f(x)}{g(x)}v(x)+\varphi(x)\right)  \left[
\frac{u(x)}{x^{n-r+1}}\right]  _{+}\bmod u(x),\\
\{v(x),h_{r}\}  &  =\frac{f(x)}{2g(x)}w(x)\left[  \frac{u(x)}{x^{n-r+1}%
}\right]  _{+}\bmod u(x).
\end{align*}
Using these formulas we calculate that for $g_{r}(x)$ given by (\ref{1.1}) we
get
\begin{align*}
(U_{r})_{21}(x)  &  =\frac{1}{2u(x)}\left(  \frac{f(x)}{g(x)}w(x)\left[
\frac{u(x)}{x^{n-r+1}}\right]  _{+}-2\{v(x),h_{r}\}\right) \\
&  =\frac{1}{2u(x)}\left(  \frac{f(x)}{g(x)}w(x)\left[  \frac{u(x)}{x^{n-r+1}%
}\right]  _{+}-\frac{f(x)}{g(x)}w(x)\left[  \frac{u(x)}{x^{n-r+1}}\right]
_{+}\bmod u(x)\right) \\
&  =\left[  \frac{\frac{f(x)}{2g(x)}w(x)\left[  \frac{u(x)}{x^{n-r+1}}\right]
_{+}}{u(x)}\right]  _{+}%
\end{align*}
and the remaining components of $U_{r}(x)$ can be received in a similar
fashion. On the other hand for $g_{r}(x)$ given by (\ref{1.4}) we obtain
\begin{align*}
(U_{r})_{21}(x)  &  =\frac{1}{2u(x)}\left(  -\frac{f(x)}{g(x)}w(x)\left[
\frac{u(x)}{x^{n-r+1}}\right]  _{-}-2\{v(x),h_{r}\}\right) \\
&  =\frac{1}{2u(x)}\left(  -\frac{f(x)}{g(x)}w(x)\left[  \frac{u(x)}%
{x^{n-r+1}}\right]  _{-}+\frac{f(x)}{g(x)}w(x)\left[  \frac{u(x)}{x^{n-r+1}%
}\right]  _{-}\bmod u(x)\right) \\
&  =\left[  \frac{-\frac{f(x)}{2g(x)}w(x)\left[  \frac{u(x)}{x^{n-r+1}%
}\right]  _{-}}{u(x)}\right]  _{+},
\end{align*}
where we used the fact that since $\frac{f(x)}{g(x)}w(x)\frac{u(x)}{x^{n-r+1}%
}\bmod u(x)=0$ we can write
\begin{align*}
\frac{f(x)}{g(x)}w(x)\left[  \frac{u(x)}{x^{n-r+1}}\right]  _{+}\bmod u(x)  &
=\frac{f(x)}{g(x)}w(x)\left(  \left[  \frac{u(x)}{x^{n-r+1}}\right]
_{+}-\frac{u(x)}{x^{n-r+1}}\right)  \bmod u(x)\\
&  =-\frac{f(x)}{g(x)}w(x)\left[  \frac{u(x)}{x^{n-r+1}}\right]
_{-}\bmod u(x).
\end{align*}
The remaining components of $U_{r}(x)$ can be received in a similar fashion.
\end{proof}

\section{Frobenius integrability of non-autonomous equations of
Painlev\'{e}-type in magnetic representation\label{S3}}

In this section we briefly sketch the construction of Frobenius integrable
non-autonomous Hamiltonian systems through appropriate deformations of
St\"{a}ckel systems. For details of this procedure, we refer the reader to
Part I \cite{part1}. Consider a multitime-dependent spectral curve:%
\begin{equation}
\sum_{r=1}^{n}h_{r}x^{n-r}=\frac{1}{2}x^{m}y^{2}-\sum_{\gamma=0}%
^{n+1}d_{\gamma}(t)x^{\gamma}y-e(t)x^{n}\equiv\Phi(x,y,t),\ \ \ \ m\in\left\{
0,\dotsc,n+1\right\}  \label{1a}%
\end{equation}
with $t=(t_{1},\dotsc,t_{n})$, $t_{i}\in\mathbf{R}$. It means that we specify
$f,$ $\sigma$ and $\varphi$ in (\ref{1}) as follows:
\begin{equation}
f(x)=x^{m}\text{, \ \ }\sigma=\sigma(x,t)=e(t)x^{n},\text{ \ \ }%
\varphi=\varphi(x,t)=\sum_{\gamma=0}^{n+1}d_{\gamma}(t)x^{\gamma}.
\label{ephi}%
\end{equation}
where we allow for a direct dependence of $\sigma$ and $\varphi$ on all
$t_{r}$. The corresponding Hamiltonians $h_{r}$ become then non-autonomous
(directly multitime-dependent):%
\begin{equation}
h_{r}=E_{r}+\sum_{\gamma=0}^{n+1}d_{\gamma}(t_{1},\dotsc,t_{n})M_{r}%
^{(\gamma)}+e(t_{1},\dotsc,t_{n})V_{r}^{(n)},\quad r=1,\dotsc,n, \label{2}%
\end{equation}
where $M_{r}^{(\gamma)}$ are the magnetic potentials (\ref{5dd}), $V_{r}%
^{(n)}=-\rho_{r}$ is the first nontrivial ordinary potential in (\ref{5d}) and
where $d_{\gamma}$ and $e$ are yet unspecified functions of in general all
evolution parameters $t_{r}$. Let us now perturb the Hamiltonians $h_{r}$ in
(\ref{2})
\begin{equation}
h_{1}^{B}=h_{1},\quad h_{r}^{B}=h_{r}+W_{r},\quad r=2,\dotsc,n. \label{pert}%
\end{equation}
by quasi-St\"{a}ckel terms $W_{r}=W_{r}(\lambda,\mu)$, linear in momenta
$\mu_{i}$, and given by%
\begin{equation}
W_{r}=-\sum\limits_{i=1}^{n}\left(  \sum\limits_{k=1}^{r-1}\,k\,\rho
_{r-k-1}\frac{\lambda_{i}^{m+k-1}}{\Delta_{i}}\right)  \mu_{i}=\sum
\limits_{i=1}^{n}J_{r}^{i}\mu_{i},\quad r\in I_{1}^{m} \label{j1}%
\end{equation}
and by
\begin{equation}
W_{r}=-\sum\limits_{i=1}^{n}\left(  \sum\limits_{k=1}^{n-r+1}\,k\,\rho
_{r+k-1}\frac{\lambda_{i}^{m-k-1}}{\Delta_{i}}\right)  \mu_{i}=\sum
\limits_{i=1}^{n}J_{r}^{i}\mu_{i},\quad r\in I_{2}^{m}, \label{j2}%
\end{equation}
where for each $m\in\{0,\ldots,n+1\}$ the index sets $I_{1}^{m}$ and
$I_{2}^{m}$ are defined as follows:
\begin{equation}
I_{1}^{m}=\{2,\ldots,n-m+1\},\quad I_{2}^{m}=\{n-m+2,\ldots,n\},\quad
m=0,\ldots,n+1 \label{I}%
\end{equation}
with the following degenerations for $m=0$ and for $m=n+1$:%
\[
I_{1}^{0}=I_{1}^{1}=\{2,\dots,n\}\text{, \ }I_{2}^{0}=I_{2}^{1}=\emptyset
\text{, , }I_{1}^{n+1}=I_{1}^{n}=\emptyset\text{, }I_{2}^{n+1}=I_{2}%
^{n}=\{2,\dots,n\}.
\]
The vector fields $J_{r}=J_{r}^{i}\frac{\partial}{\partial\lambda_{i}}$ on $Q$
are Killing vector fields for $g=G^{-1}$. Thus, the terms $W_{r}=%
%TCIMACRO{\tsum \nolimits_{i=1}^{n}}%
%BeginExpansion
{\textstyle\sum\nolimits_{i=1}^{n}}
%EndExpansion
J_{r}^{i}\mu_{i}$ (these terms were first introduced in (\cite{Marciniak
2017})) constitute linear in momenta $\mu_{i}$ constants of motion for the
geodesic St\"{a}ckel Hamiltonian $E_{1}$.

The functions $\mathcal{E}_{r}=E_{r}+W_{r}$ (called the \emph{geodesic
quasi-St\"{a}ckel Hamiltonians} \cite{part1},\cite{Marciniak 2017}) constitute
a nilpotent Lie algebra $\mathfrak{g}=\mathrm{span}\{\mathcal{E}_{r},\ $
$r=1,\ldots,n\}$ with the following commutation relations:
\[
\{\mathcal{E}_{1},\mathcal{E}_{r}\}=0,\quad r=2,\dots,n,
\]
and
\begin{equation}
\{\mathcal{E}_{r},\mathcal{E}_{s}\}=%
\begin{cases}
0 & \text{for $r\in I_{1}^{m}$ and $s\in I_{2}^{m}$},\\
(s-r)\mathcal{E}_{r+s-(n-m+2)} & \text{for $r,s\in I_{1}^{m}$},\\
-(s-r)\mathcal{E}_{r+s-(n-m+2)} & \text{for $r,s\in I_{2}^{m}$},
\end{cases}
\label{str}%
\end{equation}
where we denote $\mathcal{E}_{r}=0$ as soon as $r\leq0$ or $r>n$. The algebra
$\mathfrak{g}$ has an Abelian subalgebra
\begin{equation}
\mathfrak{a}=\mathrm{span}\left\{  \mathcal{E}_{1},\dotsc,\mathcal{E}%
_{\kappa_{1}},\mathcal{E}_{n-\kappa_{2}+1},\dotsc,\mathcal{E}_{n}\right\}
\label{gma}%
\end{equation}
where%
\[
\kappa_{1}=\left[  \frac{n+3-m}{2}\right]  ,\qquad\kappa_{2}=\left[  \frac
{m}{2}\right]  .
\]
The $n$ Killing vector fields $J_{r}$ were carefully chosen from the whole
$n(n+1)/2$-dimensional Lie algebra of all Killing vectors fields for
$g=G^{-1}$ precisely in order to guarantee that $\mathcal{E}_{r}=E_{r}%
+\sum\nolimits_{i=1}^{n}J_{r}^{i}\mu_{i}$ would become a nilpotent Lie-algebra.

Finally, we construct $n$ new Hamiltonians $H_{r}^{B}$ such that for
$r\in\{1\}\cup I_{1}^{m}$
\begin{align}
H_{r}^{B}  &  =h_{r}^{B},\quad\text{for $r=1,\dotsc,\kappa_{1}$},\nonumber\\
H_{r}^{B}  &  =\sum_{j=1}^{r}\zeta_{r,j}(t_{1},\dotsc,t_{r-1})h_{j}^{B}%
,\quad\zeta_{r,r}=1,\quad\text{for $r=\kappa_{1}+1,\dotsc,n-m+1$} \label{7b}%
\end{align}
and for $r\in I_{2}^{m}$
\begin{align}
H_{r}^{B}  &  =\sum_{j=0}^{n-r}\zeta_{r,r+j}(t_{r+1},\dotsc,t_{n})h_{r+j}%
^{B},\quad\zeta_{r,r}=1,\quad\text{for $r=n-m+2,\dotsc,n-\kappa_{2}$%
},\nonumber\\
H_{r}^{B}  &  =h_{r}^{B},\quad\text{for $r=n-\kappa_{2}+1,\dotsc,n$},
\label{7c}%
\end{align}
where $\zeta_{r,j}$ are some functions of some of the evolution parameters
$t_{s}$. Let us now demand that the Hamiltonians $H_{r}^{B}$ constitute a
Frobenius integrable system, i.e. that they satisfy the Frobenius
integrability conditions for non-autonomous Hamiltonians
\begin{equation}
\{H_{r}^{B},H_{s}^{B}\}+\frac{\partial H_{r}^{B}}{\partial t_{s}}%
-\frac{\partial H_{s}^{B}}{\partial t_{r}}=0,\quad r,s=1,\dots,n \label{4}%
\end{equation}
(we will thus call the Hamiltonians $H_{r}^{B}$ the \emph{Frobenius integrable
deformations} of the quasi-St\"{a}ckel Hamiltonians $h_{r}^{B}$). The $n+2$
functions $d_{\gamma}$ and the function $e$ that enter $H_{r}^{B}$ through
(\ref{2}) can be determined directly from the Frobenius condition (\ref{4}) or
equivalently - as it was shown in Part I - from the following set of linear
first order PDE's: \newline1. For $r\in\{1,\dotsc,\kappa_{1}\}\subset\{1\}\cup
I_{1}^{m}$
\begin{align}
\frac{\partial d_{\gamma}}{\partial t_{r}}  &  =0,\quad\gamma\neq
m,\dotsc,m+r-1,\label{4o}\\
\frac{\partial d_{\gamma}}{\partial t_{r}}  &  =(\gamma-m+1)d_{n-m+2+\gamma
-r},\quad\gamma=m,\dotsc,m+r-1. \label{4a}%
\end{align}
2. For $r\in\{\kappa_{1}+1,\dotsc,n-m+1\}\subset I_{1}^{m}$
\begin{align}
\frac{\partial d_{\gamma}}{\partial t_{r}}  &  =0,\quad\gamma\neq
m,\dotsc,m+r-1,\label{4n}\\
\frac{\partial d_{\gamma}}{\partial t_{r}}  &  =(\gamma-m+1)\sum
_{j=\gamma-m+1}^{r}\zeta_{r,j}(t_{1},\dotsc,t_{r-1})d_{n-m+2+\gamma-j}%
,\quad\gamma=m,\dotsc,m+r-1. \label{4b}%
\end{align}
3. For $r\in\{n-m+2,\dotsc,n-\kappa_{2}\}\subset I_{2}^{m}$
\begin{align}
\frac{\partial d_{\gamma}}{\partial t_{r}}  &  =0,\quad\gamma\neq
r-(n-m+2),\ldots,m-2,\label{ziuta}\\
\frac{\partial d_{\gamma}}{\partial t_{r}}  &  =-(\gamma-m+1)\sum
_{j=0}^{n-m+2+\gamma-r}\zeta_{r,r+j}(t_{r+1},\dotsc,t_{n})d_{n-m+2+\gamma
-r-j},\quad\gamma=r-(n-m+2),\ldots,m-2. \label{4c}%
\end{align}
4. For $r\in\{n+1-\kappa_{2},\dotsc,n\}\subset I_{2}^{m}$
\begin{align}
\frac{\partial d_{\gamma}}{\partial t_{r}}  &  =0,\quad\gamma\neq
r-(n-m+2),\ldots,m-2,\label{4r}\\
\frac{\partial d_{\gamma}}{\partial t_{r}}  &  =-(\gamma-m+1)d_{n-m+2+\gamma
-r},\quad\gamma=r-(n-m+2),\ldots,m-2. \label{4d}%
\end{align}
5. For $r=1,\dotsc,n$%
\begin{equation}
\frac{\partial e}{\partial t_{r}}=0,\quad m=0,\dotsc,n,\qquad\frac{\partial
e}{\partial t_{r}}=e\delta_{1,r},\quad m=n+1, \label{4e}%
\end{equation}
where the functions $\zeta_{r,j}(t_{1},\dotsc,t_{r-1})$ and $\zeta
_{r,r+j}(t_{r+1},\dotsc,t_{n})$ can be calculated from the first-order PDE's
resulting from the compatibility conditions for the above system of PDE's
\begin{equation}
\frac{\partial^{2}d_{\gamma}}{\partial t_{r}\partial t_{s}}=\frac{\partial
^{2}d_{\gamma}}{\partial t_{s}\partial t_{r}},\quad r,s=1,\dotsc,n
\label{7.13}%
\end{equation}
provided that all integration constants in (\ref{7.13}) are chosen to be zero.
The details of the above construction the reader can find in Part I.

The main object of our interest in this paper is the Frobenius integrable,
non-autonomous Hamiltonian system
\begin{equation}
\frac{d\xi}{dt_{r}}=Y_{r}^{B}(\xi,t)=\pi dH_{r}^{B}(\xi,t),\quad r=1,\dotsc,n
\label{5}%
\end{equation}
where $H_{r}^{B}$ are given by (\ref{7b})-(\ref{7c}) with the functions
$d_{\gamma}(t),\,\ e(t),$ $\zeta_{r,j}(t_{1},\dotsc,t_{r-1})$ and
$\zeta_{r,r+j}(t_{r+1},\dotsc,t_{n})$ satisfying the set of PDE's
(\ref{4a})-(\ref{4e}) and (\ref{7.13})). The main goal of this paper is to
prove that the system (\ref{5}) is of Painlev\'{e}-type and in order to do
this we have to construct an isomonodromic Lax representation for each of the
flows contained in the system (\ref{5}). This is done in the following section.

\section{Isomonodromic Lax representation for non-autonomous equations of
Painlev\'{e}-type in the magnetic representation\label{S3a}}

In this section we present the main theorem of this paper, namely that each of
the non-autonomous Hamiltonian flows in (\ref{5}) has an isomonodromic Lax
representation \cite{Bo,Its,Iwasaki} and thus it is appropriate to call the
system (\ref{5}) a \emph{Painlev\'{e}-type system}. The Lax matrices
considered here will have the form of the explicitly time-dependent matrix
$L(x,\xi,t)$ obtained from (\ref{L1})--(\ref{L4}) by choosing $g(x)=f(x)$ and
by assuming (\ref{ephi}).

\begin{theorem}
\label{main}Each non-autonomous Hamiltonian flow $\frac{d\xi}{dt_{r}}%
=Y_{r}^{B}(\xi,t)$ in (\ref{5}) has the isomonodromic Lax representation
\begin{equation}
\frac{d}{dt_{r}}L(x,\xi,t)=[\overline{U}_{r}(x,\xi,t),L(x,\xi,t)]+2x^{m}%
\frac{\partial}{\partial x}\overline{U}_{r}(x,\xi,t) \label{6}%
\end{equation}
where now
\begin{equation}
\frac{d}{dt_{r}}=\frac{\partial}{\partial t_{r}}+\{\cdot,H_{r}^{B}\}
\label{czasB}%
\end{equation}
is the evolutionary derivative along the Hamiltonian vector field $Y_{r}^{B}$,
the Lax matrix $L(x,\xi,t)$ is given by
\begin{equation}
L(x,\xi,t)=%
\begin{pmatrix}%
\begin{array}
{c}%
v(x)-\sum_{\gamma=0}^{n+1}d_{\gamma}(t)x^{\gamma}\\
\end{array}
&
\begin{array}
{c}%
u(x)\\
\end{array}
\\
w(x,t) & -v(x)+\sum_{\gamma=0}^{n+1}d_{\gamma}(t)x^{k}%
\end{pmatrix}
\label{Lspec}%
\end{equation}
with $u(x)$ given by (\ref{L2}) as before, with
\[
v(x)=\sum_{i=1}^{n}\lambda_{i}^{m}\mu_{i}\prod_{k\neq i}\frac{x-\lambda_{k}%
}{\lambda_{i}-\lambda_{k}}=-\sum_{k=1}^{n}\left[  \sum_{i=1}^{n}\frac
{\partial\rho_{k}}{\partial\lambda_{i}}\frac{\lambda_{i}^{m}\mu_{i}}%
{\Delta_{i}}\right]  x^{n-k},
\]
while $w$ is given by
\[
w(x,t)=-2x^{m}\left[  \frac{\Phi(x,v(x)x^{-m},t)}{u(x)}\right]  _{+}\text{,}%
\]
with $\Phi(x,y,t)$ given by (\ref{1a}) and by (\ref{ephi}). Further, for
$r\in\{1\}\cup I_{1}^{m}$
\begin{align}
\overline{U}_{r}(x,\xi,t)  &  =U_{r}(x,\xi,t)\quad\text{for $r=1,\dotsc
,\kappa_{1}$},\nonumber\\
\overline{U}_{r}(x,\xi,t)  &  =\sum_{j=1}^{r}\zeta_{r,j}(t_{1},\dotsc
,t_{r-1})U_{j}(x,\xi,t)\quad\text{for $r=\kappa_{1}+1,\dotsc,n-m+1$}
\label{7a}%
\end{align}
and for $r\in I_{2}^{m}$
\begin{align}
\overline{U}_{r}(x,\xi,t)  &  =\sum_{j=0}^{n-r}\zeta_{r,r+j}(t_{r+1}%
,\dotsc,t_{n})U_{r+j}(x,\xi,t)\quad\text{for $r=n-m+2,\dotsc,n-\kappa_{2}$%
},\nonumber\\
\overline{U}_{r}(x,\xi,t)  &  =U_{r}(x,\xi,t)\quad\text{for $r=n-\kappa
_{2}+1,\dotsc,n,$} \label{7aa}%
\end{align}
with the matrix $U_{r}(x;t)$ given by (\ref{1.2})--(\ref{1.5}):%
\begin{align}
U_{r}(x,\xi,t)  &  =\left[  \frac{B_{r}(x)}{u(x)}\right]  _{+}%
,\ \ \ \ \ \ \ B_{r}(x)=\frac{1}{2}\left[  \frac{u(x)}{x^{n+1-r}}\right]
_{+}L(x,\xi,t),\ \ \ \ \ \quad r\in\{1\}\cup I_{1}^{m},\label{8}\\
U_{r}(x,\xi,t)  &  =\left[  \frac{B_{r}(x)}{u(x)}\right]  _{+}%
,\ \ \ \ \ \ \ B_{r}(x)=-\frac{1}{2}\left[  \frac{u(x)}{x^{n+1-r}}\right]
_{-}L(x,\xi,t),\ \ \ \quad r\in I_{2}^{m}. \label{11}%
\end{align}
and with the functions $\zeta_{r,j}$ and $\zeta_{r,r+j}$ satisfying the set of
PDE's (\ref{7.13}).
\end{theorem}

The proof of the theorem is technical and requires additional lemmas, so we
shifted it to Appendix. The proof, as well as all the examples, are presented
in Vi\`{e}te coordinates $(q,p)$ instead of the separation coordinates
$(\lambda,\mu)$. Vi\`{e}te coordinates $(q,p)=(q_{1}\ldots,q_{n},p_{1}%
,\ldots,p_{n})$ are related with separation coordinates through the point transformation%

\begin{equation}
q_{i}=\rho_{i}(\lambda),\quad p_{i}=\ -\sum_{k=1}^{n}\frac{\lambda_{k}%
^{n-i}\mu_{k}}{\Delta_{k}},\quad i=1,\dotsc,n. \label{Viete}%
\end{equation}
In this specification $u(x)$ and $v(x)$ in (\ref{Lspec}) are expressed as
follows:%
\[
u(x)=\sum_{k=0}^{n}q_{n-k}x^{k}\text{, \ \ }v(x)=\sum_{k=0}^{n-1}M_{n-k}%
^{(m)}x^{k}.
\]
where the magnetic potentials $M_{r}^{(m)}$ attain the form
\begin{align}
M_{r}^{(m)}  &  =-\sum_{j=1}^{n}\left[  \sum_{s=0}^{r-1}q_{s}V_{j}%
^{(r+m-s-1)}\right]  p_{j}\nonumber\\
&  =\left\{
\begin{array}
{c}%
\sum\limits_{j=0}^{r-1}q_{j}p_{n+1-m-r+j}\text{ }\ \ \text{for }%
r=1,\ldots,n-m\\
\\
-\sum\limits_{j=r}^{n}q_{j}p_{n+1-m-r+j}\text{ }\ \ \text{for }r=n-m+1,\ldots
,n
\end{array}
\right.  \text{ for }m=0,\ldots,n+1. \label{18a}%
\end{align}
The first equality in (\ref{18a}) is shown in Part I while the second equality
in (\ref{18a}) is proved in \cite{Marciniak 2017}. Note also that $V_{r}%
^{(k)}$ are defined by the non-tensorial formula (\ref{5d})-(\ref{6a}) that is
valid in any coordinate system. These formulas yield easily the form of
$V_{r}^{(k)}$ in Vi\`{e}te coordinates $q$. Finally, the quasi-St\"{a}ckel
terms $W_{r}$ in Vi\`{e}te coordinates are given by%

\begin{equation}%
\begin{split}
W_{r}  &  =\sum\limits_{k=n-m-r+2}^{n-m}(n+1-m-k)q_{m+r-n-2+k}\,p_{k},\quad
r\in I_{1}^{m}\\
W_{r}  &  =-\sum\limits_{k=n-m+2}^{2n-m+2-r}(n+1-m-k)q_{m+r-n-2+k}%
\,p_{k},\quad r\in I_{2}^{m}.
\end{split}
\label{kil5}%
\end{equation}
For details, see Part I.

\begin{example}
\label{1e}Let us explicitly show the isomonodromic Lax representation for the
Painlev\'{e}-type system (\ref{5}) given by $n=3,$ $m=1,$ $b_{0}=b_{1}%
=b_{2}=b_{4}=\overline{b}=0$ (see Example 4 in Part I). We have $H_{r}%
^{B}=h_{r}^{B},$ $\overline{U}_{r}=U_{r},\ r=1,2,$ $H_{3}^{B}=h_{3}^{B}%
+t_{2}h_{1}^{B}$, $\overline{U}_{3}=U_{3}+t_{2}U_{1}$\ with $h_{r}^{B}$ of the
form
\[
h_{r}^{B}=\ \mathcal{E}_{r}+b_{3}M_{r}^{(3)}+2b_{3}t_{3}M_{r}^{(2)}%
+b_{3}(t_{3}^{2}+t_{2})M_{r}^{(1)},
\]
where the geodesic quasi-St\"{a}ckel Hamiltonians $\mathcal{E}_{r}=E_{r}%
+W_{r}$\ are given explicitly by (\ref{geo6a}) while the magnetic potentials
\ $M_{r}^{(\gamma)}$\ are given by (\ref{18a}) (they are also given explicitly
in (\ref{2b})). The Lax matrix $L$ takes the form%
\[
L=\left(
\begin{array}
[c]{cc}%
-b_{3}x^{3}-(2b_{3}t_{3}+p_{2})x^{2}-[b_{3}(t_{3}^{2}+t_{2})+p_{1}+q_{1}%
p_{2}]x+q_{3}p_{3} & x^{3}+q_{1}x^{2}+q_{2}x+q_{3}\\%
\begin{array}
{c}%
\\
-2b_{3}p_{2}x^{2}-2b_{3}(p_{1}+2t_{3}p_{2})x-q_{3}p_{3}^{2}%
\end{array}
&
\begin{array}
{c}%
\\
-L_{11}%
\end{array}
\end{array}
\right)
\]
while the auxiliary matrices are as follows
\[
U_{1}=\left(
\begin{array}
[c]{cc}%
-\frac{1}{2}b_{3} & \frac{1}{2}\\%
\begin{array}
{c}%
\\
0
\end{array}
&
\begin{array}
{c}%
\\
\frac{1}{2}b_{3}%
\end{array}
\end{array}
\right)  ,
\]%
\[
\]%
\[
U_{2}=\left(
\begin{array}
[c]{cc}%
-\frac{1}{2}b_{3}x-b_{3}t_{3}-\frac{1}{2}p_{2} & \frac{1}{2}x+\frac{1}{2}%
q_{1}\\%
\begin{array}
{c}%
\\
-b_{3}p_{2}%
\end{array}
&
\begin{array}
{c}%
\\
-(U_{2})_{11}%
\end{array}
\end{array}
\right)  ,
\]%
\[
\]%
\[
U_{3}=\left(
\begin{array}
[c]{cc}%
\begin{array}
{c}%
-\frac{1}{2}b_{3}x^{2}-(b_{3}t_{3}+\frac{1}{2}p_{2})x\\
-\frac{1}{2}[b_{3}(t_{3}^{2}+t_{2})+p_{1}+q_{1}p_{2}]
\end{array}
& \frac{1}{2}x^{2}+\frac{1}{2}q_{1}x+\frac{1}{2}q_{2}\\%
\begin{array}
{c}%
\\
-b_{3}p_{2}x-b_{3}(p_{1}+2t_{3}p_{2})-\frac{1}{2}p_{1}^{2}%
\end{array}
&
\begin{array}
{c}%
\\
-(U_{3})_{11}%
\end{array}
\end{array}
\right)
\]
A direct calculation confirms that the matrices $L,\overline{U}_{1}%
,\overline{U}_{2},\overline{U}_{3}$ do satisfy the isomonodromic Lax
representation (\ref{6}).
\end{example}

\section{Isomonodromic Lax representations for one-, two- and
three-dimensional Painlev\'{e}-type systems in the magnetic
representation\label{S4}}

In Section $7$ and Section $9$ of Part I we presented the complete list of all
one-, two- and three-dimensional non-autonomous Frobenius integrable systems,
originating from our deformation procedure, in the magnetic representation.
Here we present again these one-, two- and three-dimensional systems, this
time together with their isomonodromic Lax representations (\ref{6}) in
Vi\`{e}te coordinates. In each case $(n,m)$ we obtain a $(n+3)$-parameter
family of systems, parametrized by real constants $b_{0},\ldots,b_{n+1}%
,\overline{b}$. We also remind the reader that $e(t)=\overline{b}$ for
$m=0,\ldots,n$ and $e(t)=\overline{b}e^{t_{1}}$ for $m=n+1$, due to (\ref{4e}).

\subsection{One-dimensional systems}

Let us first consider the case $n=1$. In this case $H^{B}=h^{B}$,
$\overline{U}=U$ for all $m=0,1,2$ and we obtain for each $m$ a $4$-parameter
family of related Painlev\'{e}-type systems.

For $m=0$ we get
\[
h^{B}=\frac{1}{2}p^{2}+(b_{2}q^{2}-b_{1}q+b_{2}t+b_{0})p+\overline{b}\,q,
\]%
\[
L=\left(
\begin{array}
[c]{cc}%
-b_{2}x^{2}-b_{1}x-b_{2}t-b_{0}-p & x+q\\%
\begin{array}
{c}%
\\
-2(b_{2}p\,x-b_{2}qp+b_{1}p-\overline{b})
\end{array}
&
\begin{array}
{c}%
\\
-L_{11}%
\end{array}
\end{array}
\right)  ,\ \ \ U=\left(
\begin{array}
[c]{cc}%
-\frac{1}{2}(b_{2}x-b_{2}q+b_{1}) & \frac{1}{2}\\%
\begin{array}
{c}%
\\
-b_{2}p
\end{array}
&
\begin{array}
{c}%
\\
-U_{11}%
\end{array}
\end{array}
\right)  .
\]

For $m=1$ we obtain%
\[
h^{B}=-\frac{1}{2}qp^{2}+\left[  b_{2}q^{2}-(b_{2}t+b_{1})q+b_{0}\right]
p+\overline{b}\,q,
\]%
\[
L=\left(
\begin{array}
[c]{cc}%
-b_{2}x^{2}-(b_{2}t+b_{1})x+qp-b_{0} & x+q\\%
\begin{array}
{c}%
\\
2(b_{2}qp+\overline{b})x-qp^{2}+2b_{0}p
\end{array}
&
\begin{array}
{c}%
\\
-L_{11}%
\end{array}
\end{array}
\right)  ,\ \ \ U=\left(
\begin{array}
[c]{cc}%
-\frac{1}{2}(b_{2}x-b_{2}q+b_{2}t+b_{1}) & \frac{1}{2}\\%
\begin{array}
{c}%
\\
b_{2}qp+\overline{b}%
\end{array}
&
\begin{array}
{c}%
\\
-U_{11}%
\end{array}
\end{array}
\right)  ,
\]

while for $m=2$ we obtain%
\[
h^{B}=\frac{1}{2}q^{2}p^{2}+(b_{2}{\mathrm{e}^{t}}q^{2}-b_{1}q+b_{0}%
)p+\overline{b}\,{\mathrm{e}^{t}}q,
\]%
\begin{align*}
L  &  =\left(
\begin{array}
[c]{cc}%
-b_{2}{\mathrm{e}^{t}}x^{2}-b_{1}x-q^{2}p-b_{0} & x+q\\%
\begin{array}
{c}%
\\
2\overline{b}{\mathrm{e}^{t}}x^{2}+(q^{2}p^{2}-2b_{1}qp+2b_{0}p)x-q^{3}%
p^{2}-2b_{0}qp
\end{array}
&
\begin{array}
{c}%
\\
-L_{11}%
\end{array}
\end{array}
\right)  ,\ \ \\
&  \ \\
U  &  =\left(
\begin{array}
[c]{cc}%
-\frac{1}{2}(b_{2}{\mathrm{e}^{t}}x-b_{2}{\mathrm{e}^{t}}q+b_{1}) & \frac
{1}{2}\\%
\begin{array}
{c}%
\\
\overline{b}{\mathrm{e}^{t}}x+\frac{1}{2}q^{2}p^{2}-b_{1}qp+b_{0}%
p-\overline{b}{\mathrm{e}^{t}}q
\end{array}
&
\begin{array}
{c}%
\\
-U_{11}%
\end{array}
\end{array}
\right)  .
\end{align*}
In particular they contain the isomonodromic Lax representation for
Painlev\'{e}-II, Painlev\'{e}-IV and Painlev\'{e}-III respectively, in the
magnetic representation (see Part I).

\subsection{Two-dimensional systems}

For $n=2$ again $H_{r}^{B}=h_{r}^{B}$ and $\overline{U}_{r}=U_{r}$ for all
$m=0,...,3$ and we obtain for each $m$ a $5$-parameter family of related
Painlev\'{e}-type systems. We use here the notation $p_{0}=-q_{1}p_{1}%
-q_{2}p_{2}$ (which follows formally from (\ref{Viete}) with $i=0$).

For $m=0$ we get
\begin{align*}
h_{1}^{B}  &  =p_{{1}}p_{{2}}+\tfrac{1}{2}q_{{1}}\,{p_{{2}}^{2}}+[b_{3}%
(q_{1}^{2}-q_{2}+2t_{2})-b_{2}q_{1}+b_{1}]p_{1}+[b_{3}(q_{1}q_{2}+t_{1}%
)-b_{2}(q_{2}-t_{2})+b_{0}]p_{2}+\overline{b}q_{1},\\
&  \ \\
h_{2}^{B}  &  =\tfrac{1}{2}\,{p_{{1}}^{2}+}q_{{1}}p_{{1}}p_{{2}}+\tfrac{1}%
{2}\left(  {q_{{1}}^{2}}-q_{{2}}\right)  {p_{{2}}^{2}}+[b_{3}(q_{1}q_{2}%
+t_{1})-b_{2}(q_{2}-t_{2})+b_{0}]p_{1}\\
&  \ \ \ \ +[b_{3}(q_{2}^{2}+t_{1}q_{1}-2t_{2}q_{2})+b_{2}t_{2}q_{1}%
-b_{1}q_{2}+b_{0}q_{1}+1]p_{2}+\overline{b}q_{2},
\end{align*}%
\[
\]
with%
\begin{equation}
L=\left(
\begin{array}
[c]{cc}%
L_{11} & x^{2}+q_{1}x+q_{2}\\%
\begin{array}
{c}%
\\
L_{21}%
\end{array}
&
\begin{array}
{c}%
\\
-L_{11}%
\end{array}
\end{array}
\right)  , \label{Lax}%
\end{equation}%
\[
\]
where%
\begin{align*}
L_{11}  &  =-b_{3}x^{3}-b_{2}x^{2}-(p_{2}+2b_{3}t_{2}+b_{1})x-(p_{1}%
+q_{1}p_{2}+b_{3}t_{1}+b_{2}t_{2}+b_{0}),\\
& \\
L_{21}  &  =-2b_{3}p_{2}x^{2}-2(b_{3}p_{1}+b_{2}p_{2})x-p_{2}^{2}%
-2[b_{3}(p_{0}+t_{2}p_{2})+b_{2}p_{1}+b_{1}p_{2}+\overline{b}],
\end{align*}%
\[
\]%
\[
U_{1}=\left(
\begin{array}
[c]{cc}%
-\frac{1}{2}b_{3}x+\frac{1}{2}(b_{3}q_{1}-b_{2}) & \frac{1}{2}\\%
\begin{array}
{c}%
\\
-b_{3}p_{2}%
\end{array}
&
\begin{array}
{c}%
\\
-(U_{1})_{11}%
\end{array}
\end{array}
\right)  ,
\]%
\[
\]%
\[
U_{2}=\left(
\begin{array}
[c]{cc}%
-\frac{1}{2}b_{3}x^{2}-\frac{1}{2}b_{2}x+b_{3}(\frac{1}{2}q_{2}-t_{2}%
)-\frac{1}{2}b_{1}-\frac{1}{2}p_{2} & \frac{1}{2}x+\frac{1}{2}q_{1}\\%
\begin{array}
{c}%
\\
-b_{3}p_{2}x-b_{3}p_{1}-b_{2}p_{2}%
\end{array}
&
\begin{array}
{c}%
\\
-(U_{2})_{11}%
\end{array}
\end{array}
\right)  .
\]

For $m=1$ we obtain%
\begin{align*}
h_{1}^{B} &  =\tfrac{1}{2}\,{p}_{1}^{2}-\tfrac{1}{2}\,q_{{2}}{p}_{2}%
^{2}+[b_{3}(q_{1}^{2}-q_{2}-2t_{2}q_{1}+t_{1}+t_{2}^{2})-b_{2}(q_{1}%
-t_{2})+b_{1}]p_{1}\\
&  \ \ \ +[b_{3}(q_{1}-2t_{2})q_{2}-b_{2}q_{2}+b_{0}]p_{2}+\overline{b}%
q_{1},\\
& \\
h_{2}^{B} &  =-q_{{2}}p_{{1}}p_{{2}}-\tfrac{1}{2}\,q_{{1}}q_{{2}}{p}_{2}%
^{2}+[b_{3}(q_{1}-2t_{2})q_{2}-b_{2}q_{2}+b_{0}+1]p_{1}\\
&  \ \ \ +[b_{3}(q_{2}^{2}-q_{2}t_{1}-q_{2}t_{2}^{2})-b_{2}t_{2}q_{2}%
-b_{1}q_{2}+b_{0}q_{1}]p_{2}+\overline{b}q_{2},\ \
\end{align*}
with \thinspace$L$ given by (\ref{Lax}) with%
\begin{align*}
L_{11} &  =-b_{3}x^{3}-(2b_{3}t_{2}+b_{2})x^{2}-[b_{3}(t_{1}+t_{2}^{2}%
)+b_{2}t_{2}+b_{1}+p_{1}]x+q_{2}p_{2}-b_{0})\\
& \\
L_{21} &  =-2b_{3}p_{1}x^{2}-2[b_{3}(p_{0}+2t_{2}p_{1})+b_{2}p_{1}%
-\overline{b}]x-q_{2}p_{2}^{2}+2b_{0}p_{2}%
\end{align*}
and with%
\[
U_{1}=\left(
\begin{array}
[c]{cc}%
-\frac{1}{2}b_{3}x+b_{3}(\frac{1}{2}q_{1}-t_{2})-\frac{1}{2}b_{2} & \frac
{1}{2}\\%
\begin{array}
{c}%
\\
-b_{3}p_{1}%
\end{array}
&
\begin{array}
{c}%
\\
-(U_{1})_{11}%
\end{array}
\end{array}
\right)  ,
\]%
\[
\]%
\[
U_{2}=\left(
\begin{array}
[c]{cc}%
-\tfrac{1}{2}b_{3}x^{2}-(b_{3}t_{2}+\tfrac{1}{2}b_{2})x-\tfrac{1}{2}%
[b_{3}(t_{1}+t_{2}^{2}-q_{2})+b_{2}t_{2}+b_{1}+p_{1}] & \frac{1}{2}x+\frac
{1}{2}q_{1}\\%
\begin{array}
{c}%
\\
-b_{3}p_{1}x-b_{3}(p_{0}+2t_{2}p_{1})-b_{2}p_{1}+\overline{b}%
\end{array}
&
\begin{array}
{c}%
\\
-(U_{2})_{11}%
\end{array}
\end{array}
\right)  .
\]

For $m=2$ we get%
\begin{align*}
h_{1}^{B}  &  =-\tfrac{1}{2}q_{{1}}\,{p}_{1}^{2}-q_{{2}}p_{{2}}p_{{1}}%
+[b_{3}(q_{1}^{2}-q_{2}-t_{1}q_{1})-b_{2}q_{1}+b_{1}]p_{1}+[b_{3}(q_{1}%
-t_{1})q_{2}-b_{2}q_{2}+b_{0}{\mathrm{e}^{t_{2}}}]p_{2}+\overline{b}q_{1},\\
& \\
h_{2}^{B}  &  =-\tfrac{1}{2}\,q_{{2}}{p}_{1}^{2}+\frac{1}{2}{q}_{2}^{2}%
\,{p}_{2}^{2}+[b_{3}(q_{1}q_{2}-t_{1}q_{2})-b_{2}q_{2}+b_{0}{\mathrm{e}%
^{t_{2}}}]p_{1}+[b_{3}q_{2}^{2}-b_{1}q_{2}+b_{0}{\mathrm{e}^{t_{2}}}%
q_{1}+q_{2}]p_{2}+\overline{b}q_{2},
\end{align*}
\thinspace while $L$ is given by (\ref{Lax}), where now
\begin{align*}
L_{11}  &  =-b_{3}x^{3}-(b_{3}t_{1}+b_{2})x^{2}-(p_{0}+b_{1})x+q_{2}%
p_{1}-b_{0}{\mathrm{e}^{t_{2}}}),\\
& \\
L_{21}  &  =-2(b_{3}p_{0}-\overline{b})x^{2}+2(b_{1}p_{1}+b_{0}{\mathrm{e}%
^{t_{2}}p}_{2}-\frac{1}{2}q_{1}p_{1}^{2}-q_{2}p_{1}p_{2})x+2b_{0}%
{\mathrm{e}^{t_{2}}}p_{1}-q_{2}p_{1}^{2}%
\end{align*}
and moreover%
\[
U_{1}=\left(
\begin{array}
[c]{cc}%
-\frac{1}{2}b_{3}x+\frac{1}{2}[b_{3}(q_{1}-t_{1})-b_{2}] & \frac{1}{2}\\%
\begin{array}
{c}%
\\
-b_{3}p_{0}+\overline{b}%
\end{array}
&
\begin{array}
{c}%
\\
-(U_{1})_{11}%
\end{array}
\end{array}
\right)  ,
\]%
\[
\]%
\[
U_{2}=\left(
\begin{array}
[c]{cc}%
\frac{1}{2}b_{3}q_{2}+\frac{1}{2}(b_{0}{\mathrm{e}^{t_{2}}}-q_{2}p_{1})x^{-1}
& -\frac{1}{2}q_{2}x^{-1}\\%
\begin{array}
{c}%
\\
(\frac{1}{2}q_{2}p_{1}^{2}-b_{0}{\mathrm{e}^{t_{2}}}p_{1})x^{-1}%
\end{array}
&
\begin{array}
{c}%
\\
-(U_{2})_{11}%
\end{array}
\end{array}
\right)  .
\]

For $m=3$ we get%
\begin{align*}
h_{1}^{B}  &  =\tfrac{1}{2}p_{0}^{2}-\tfrac{1}{2}{q}_{2}\,{p}_{1}^{2}%
+[b_{3}{\mathrm{e}^{t_{1}}}(q_{1}^{2}-q_{2})-b_{2}q_{1}+b_{1}+b_{0}t_{2}%
]p_{1}+(b_{3}{\mathrm{e}^{t_{1}}q}_{1}q_{2}-b_{2}q_{2}+b_{0})p_{2}%
+\overline{b}{\mathrm{e}^{t_{1}}}q_{1},\\
& \\
h_{2}^{B}  &  =\tfrac{1}{2}\,q_{{1}}q_{{2}}{p}_{1}^{2}+q_{2}^{2}p_{1}%
p_{2}+[b_{3}{\mathrm{e}^{t_{1}}}q_{1}q_{2}-b_{2}q_{2}+b_{0}+q_{2}%
]p_{1}\ +[b_{3}{\mathrm{e}^{t_{1}}}q_{2}^{2}-b_{1}q_{2}+b_{0}(q_{1}-t_{2}%
q_{2})]p_{2}+\overline{b}{\mathrm{e}^{t_{1}}}q_{2},
\end{align*}
while $L$ again is given by (\ref{Lax}), where%
\begin{align*}
L_{11}  &  =-b_{3}{\mathrm{e}^{t_{1}}}x^{3}-b_{2}x^{2}-(b_{1}+b_{0}t_{2}%
-q_{1}p_{0}-q_{2}p_{1})x+q_{2}p_{0}-b_{0},\\
& \\
L_{21}  &  =2\overline{b}{\mathrm{e}^{t_{1}}}x^{3}+[2b_{2}p_{0}+2b_{1}%
p_{1}+2b_{0}(t_{2}p_{1}+p_{2})+p_{0}^{2}-q_{2}p_{1}^{2}]x^{2}\\
&  \ \ \ \ +2[b_{1}p_{0}+b_{0}(t_{2}p_{0}+p_{1})-\tfrac{1}{2}q_{1}p_{0}%
^{2}-q_{2}p_{1}p_{0}]x+2b_{0}p_{0}-q_{2}p_{0}^{2}%
\end{align*}
and moreover%
\[
U_{1}=\left(
\begin{array}
[c]{cc}%
-\frac{1}{2}b_{3}{\mathrm{e}^{t_{1}}}x+\frac{1}{2}(b_{3}{\mathrm{e}^{t_{1}}%
q}_{1}-b_{2}) & \frac{1}{2}\\%
\begin{array}
{c}%
\\
\overline{b}{\mathrm{e}^{t_{1}}}x+b_{2}p_{0}+b_{1}p_{1}+b_{0}(t_{2}p_{1}%
+p_{2})+\frac{1}{2}p_{0}^{2}-\frac{1}{2}q_{2}p_{1}^{2}-\overline{b}%
{\mathrm{e}^{t_{1}}}q_{1}%
\end{array}
&
\begin{array}
{c}%
\\
-(U_{1})_{11}%
\end{array}
\end{array}
\right)  ,
\]%
\[
\]%
\[
U_{2}=\left(
\begin{array}
[c]{cc}%
\frac{1}{2}b_{3}{\mathrm{e}^{t_{1}}q}_{2}-\frac{1}{2}(q_{2}p_{0}-b_{0})x^{-1}
& -\frac{1}{2}q_{2}x^{-1}\\%
\begin{array}
{c}%
\\
-\overline{b}{\mathrm{e}^{t_{1}}}q_{2}+(-b_{0}p_{0}+\frac{1}{2}q_{2}p_{0}%
^{2})x^{-1}%
\end{array}
&
\begin{array}
{c}%
\\
-(U_{2})_{11}%
\end{array}
\end{array}
\right)  .
\]

\subsection{Three-dimensional systems}

For each $m=0,...,4$ we obtain a $6$-parameter family of Painlev\'{e}-type
systems satisfying (\ref{4}). We use the notation $p_{0}=-q_{1}p_{1}%
-q_{2}p_{2}-q_{3}p_{3}$ (again, see (\ref{Viete}) with $i=0$). For all
families the magnetic potentials $M_{r}^{(\gamma)}$ in (\ref{18a}) are given
explicitly by%

\begin{equation}%
\begin{array}
[c]{lll}%
M_{1}^{(0)}=p_{3}, & M_{2}^{(0)}=p_{2}+q_{1}p_{3}, & M_{3}^{(0)}=p_{1}%
+q_{1}p_{2}+q_{2}p_{3},\\
M_{1}^{(1)}=p_{2}, & M_{2}^{(1)}=p_{1}+q_{1}p_{2}, & M_{3}^{(1)}=-q_{3}%
p_{3},\\
M_{1}^{(2)}=p_{1}, & M_{2}^{(2)}=-q_{2}p_{2}-q_{3}p_{3}, & M_{3}^{(2)}%
=-q_{3}p_{2},\\
M_{1}^{(3)}=p_{0}, & M_{2}^{(3)}=-q_{2}p_{1}-q_{3}p_{2}, & M_{3}^{(3)}%
=-q_{3}p_{1},\\
M_{1}^{(4)}=-q_{1}p_{0}-q_{2}p_{1}-q_{3}p_{2}, & M_{2}^{(4)}=-q_{2}p_{0}%
-q_{3}p_{1}, & M_{3}^{(4)}=-q_{3}p_{0}.
\end{array}
\label{2b}%
\end{equation}
while $W_{r}$ are given by (\ref{kil5}).

For $m=0$ we have $\mathfrak{a}$ $=\mathfrak{g}$ so $H_{r}^{B}=h_{r}^{B}$ and
$\overline{U}_{r}=U_{r}$ for all $r$ and our procedure yields
\begin{align*}
h_{r}^{B}= &  \ \mathcal{E}_{r}+b_{4}M_{r}^{(4)}+b_{3}M_{r}^{(3)}%
+(b_{2}+3b_{4}t_{3})M_{r}^{(2)}+(b_{1}+2b_{3}t_{3}+2b_{4}t_{2})M_{r}^{(1)}\\
&  +[b_{0}+b_{2}t_{3}+b_{3}t_{2}+b_{4}(\tfrac{3}{2}t_{3}^{2}+t_{1}%
)]M_{r}^{(0)}+\overline{b}q_{r},
\end{align*}
where
\begin{align}
\mathcal{E}_{1} &  =q_{{1}}p_{{2}}p_{{3}}+\ p_{{1}}p_{{3}}+\tfrac{1}{2}%
\,{p}_{2}^{2}+\tfrac{1}{2}\,q_{{2}}{p}_{3}^{2},\nonumber\\
\mathcal{E}_{2} &  =\ p_{{1}}p_{{2}}+q_{{1}}p_{{1}}p_{{3}}+q_{{1}}{p}_{2}%
^{2}+{q}_{1}^{2}p_{{2}}p_{{3}}+\tfrac{1}{2}\left(  q_{{1}}q_{{2}}-\,q_{{3}%
}\right)  {p}_{3}^{2}+p_{3},\label{geo5a}\\
\mathcal{E}_{3} &  =\tfrac{1}{2}\,{p}_{1}^{2}+q_{{1}}p_{{1}}p_{{2}}+\ q_{{2}%
}p_{{1}}p_{{3}}+\tfrac{1}{2}{q}_{1}^{2}{p}_{2}^{2}+\tfrac{1}{2}\left(
\,{q}_{2}^{2}-\,q_{{3}}q_{{1}}\right)  {p}_{3}^{2}\nonumber\\
&  \ \ \ ~+\left(  q_{{1}}q_{{2}}-q_{{3}}\right)  p_{{2}}p_{{3}}+2p_{2}%
+q_{1}p_{3}\nonumber
\end{align}
(so that $W_{1}=0$, $W_{2}=p_{3}$, $W_{3}=2p_{2}+q_{1}p_{3}$). The
isomonodromic Lax representation is given by
\begin{equation}
L=\left(
\begin{array}
[c]{cc}%
L_{11} & x^{3}+q_{1}x^{2}+q_{2}x+q_{3}\\%
\begin{array}
{c}%
\\
L_{21}%
\end{array}
&
\begin{array}
{c}%
\\
-L_{11}%
\end{array}
\end{array}
\right)  ,\label{48}%
\end{equation}
where
\begin{align*}
L_{11} &  =-b_{4}x^{4}-b_{3}x^{3}-(3b_{4}t_{3}+b_{2}+p_{3})x^{2}-(2b_{4}%
t_{2}+2b_{3}t_{3}+b_{1}+p_{2}+q_{1}p_{3})x\\
&  \ \ \ -[b_{4}(\tfrac{3}{2}t_{3}^{2}+t_{1})+b_{3}t_{2}+b_{2}t_{3}%
+b_{0}+q_{1}p_{2}+q_{2}p_{3}+p_{1}],
\end{align*}%
\begin{align*}
L_{21} &  =-2b_{4}p_{3}x^{3}-2(b_{4}p_{2}+b_{3}p_{3})x^{2}-2[b_{4}%
(p_{1}+3t_{3}p_{3})+b_{3}p_{2}+b_{2}p_{3}+\tfrac{1}{2}p_{3}^{2}]x\\
&  \ \ \ -2[b_{4}(p_{0}+3t_{3}p_{2}+2t_{2}p_{3})+b_{3}(p_{1}+2t_{3}%
p_{3})+b_{2}p_{2}+b_{1}p_{3}+\tfrac{1}{2}q_{1}p_{3}^{2}+p_{2}p_{3}%
+2-\overline{b}]
\end{align*}
and
\[
U_{1}=\left(
\begin{array}
[c]{cc}%
-\frac{1}{2}b_{4}x+\frac{1}{2}(b_{4}q_{1}-b_{3}) & \frac{1}{2}\\%
\begin{array}
{c}%
\\
-b_{4}p_{3}%
\end{array}
&
\begin{array}
{c}%
\\
-(U_{1})_{11}%
\end{array}
\end{array}
\right)  ,
\]%
\[
\]%
\[
U_{2}=\left(
\begin{array}
[c]{cc}%
-\frac{1}{2}b_{4}x^{2}-\frac{1}{2}b_{3}x-\frac{1}{2}[b_{4}(3t_{3}-q_{2}%
)+b_{2}+p_{3}] & \frac{1}{2}x+\frac{1}{2}q_{1}\\%
\begin{array}
{c}%
\\
-b_{4}p_{3}x-b_{4}p_{2}-b_{3}p_{3}%
\end{array}
&
\begin{array}
{c}%
\\
-(U_{2})_{11}%
\end{array}
\end{array}
\right)  ,
\]%
\[
\]%
\[
U_{3}=\left(
\begin{array}
[c]{cc}%
\begin{array}
{c}%
-\frac{1}{2}b_{4}x^{3}-\frac{1}{2}b_{3}x^{2}-\frac{1}{2}(3b_{4}t_{3}%
+b_{2}+p_{3})x\\
-\frac{1}{2}[b_{4}(2t_{2}-q_{3})+2b_{3}t_{3}+b_{1}+p_{2}+q_{1}p_{3}]
\end{array}
& \frac{1}{2}x^{2}+\frac{1}{2}q_{1}x+\frac{1}{2}q_{2}\\%
\begin{array}
{c}%
\\
-b_{4}p_{3}x^{2}-(b_{4}p_{2}+b_{3}p_{3})x\\
-[b_{4}(p_{1}+3t_{3}p_{3})+b_{3}p_{2}+b_{2}p_{3}+\frac{1}{2}p_{3}^{2}]
\end{array}
&
\begin{array}
{c}%
\\
-(U_{3})_{11}%
\end{array}
\end{array}
\right)  .
\]

For $m=1$, we have $H_{r}^{B}=h_{r}^{B},$ $\overline{U}_{r}=U_{r},\ r=1,2,$
$H_{3}^{B}=h_{3}^{B}+t_{2}h_{1}^{B}$, $\overline{U}_{3}=U_{3}+t_{2}U_{1}$\ and
our procedure yields%
\begin{align}
h_{r}^{B}= &  \ \mathcal{E}_{r}+b_{4}M_{r}^{(4)}+(b_{3}+3b_{4}t_{3}%
)M_{r}^{(3)}+[b_{2}+2b_{3}t_{3}+b_{4}(3t_{3}^{2}+2t_{2})]M_{r}^{(2)}%
\nonumber\\
&  \ +[b_{1}+b_{2}t_{3}+b_{3}(t_{3}^{2}+t_{2})+b_{4}(t_{3}^{3}+3t_{2}%
t_{3}+t_{1})]M_{r}^{(1)}+b_{0}M_{r}^{(0)},\label{v1}%
\end{align}
where%
\begin{align}
\mathcal{E}_{1} &  =\ p_{{1}}p_{{2}}+\tfrac{1}{2}\,q_{{1}}{p}_{2}^{2}%
-\tfrac{1}{2}\,q_{3}{p_{{3}}^{2}},\nonumber\\
\mathcal{E}_{2} &  =\tfrac{1}{2}{p}_{1}^{2}+q_{{1}}p_{{1}}p_{{2}}+\tfrac{1}%
{2}\left(  \,{q}_{1}^{2}-\,q_{{2}}\right)  {p_{{2}}^{2}}-q_{{3}}p_{{2}}p_{{3}%
}-\tfrac{1}{2}\,q_{{1}}q_{{3}}{p}_{3}^{2}+p_{2},\label{geo6a}\\
\mathcal{E}_{3} &  =-q_{{3}}p_{{1}}p_{{3}}-\tfrac{1}{2}\,q_{{3}}{p}_{2}%
^{2}-q_{{1}}q_{{3}}p_{{2}}p_{{3}}-\tfrac{1}{2}\,q_{{2}}q_{{3}}{p}_{3}%
^{2}+q_{1}p_{2}+2p_{1}.\nonumber
\end{align}
The Lax matrix $L$ is given by (\ref{48}) with
\begin{align*}
L_{11} &  =-b_{4}x^{4}-(3b_{4}t_{3}+b_{3})x^{3}-[b_{4}(3t_{3}^{2}%
+2t_{2})+2b_{3}t_{3}+b_{2}+p_{2}]x^{2}\\
&  \ \ \ -[b_{4}(t_{3}^{3}+3t_{2}t_{3}+t_{1})+b_{3}(t_{3}^{2}+t_{2}%
)+b_{2}t_{3}+b_{1}+p_{1}+q_{1}p_{2}]x+q_{3}p_{3}-b_{0},
\end{align*}%
\begin{align*}
L_{21} &  =-2b_{4}p_{2}x^{3}-2[b_{4}(p_{1}+3t_{3}p_{2})+b_{3}p_{2}]x^{2}\\
&  \ \ \ -2[b_{4}(p_{0}+3t_{3}p_{1}+2t_{2}p_{2}+3t_{3}^{2}p_{2})+b_{3}%
(p_{1}+2t_{3}p_{2})+b_{2}p_{2}+\tfrac{1}{2}p_{2}^{2}-\overline{b}]x-q_{3}%
p_{3}^{2}+2b_{0}p_{3}%
\end{align*}
while
\[
U_{1}=\left(
\begin{array}
[c]{cc}%
-\frac{1}{2}b_{4}x+\frac{1}{2}[b_{4}(q_{1}-3t_{3})-b_{3}] & \frac{1}{2}\\%
\begin{array}
{c}%
\\
-b_{4}p_{2}%
\end{array}
&
\begin{array}
{c}%
\\
-(U_{1})_{11}%
\end{array}
\end{array}
\right)  ,
\]%
\[
\]%
\[
U_{2}=\left(
\begin{array}
[c]{cc}%
-\frac{1}{2}b_{4}x^{2}-\frac{1}{2}(3b_{4}t_{3}+b_{3})x-\frac{1}{2}%
[b_{4}(3t_{3}^{2}+2t_{2}-q_{2})+2b_{3}t_{3}+b_{2}+p_{2}] & \frac{1}{2}%
x+\frac{1}{2}q_{1}\\%
\begin{array}
{c}%
\\
-b_{4}p_{2}x-b_{4}(p_{1}+3t_{3}p_{2})-b_{3}p_{2}%
\end{array}
&
\begin{array}
{c}%
\\
-(U_{2})_{11}%
\end{array}
\end{array}
\right)  ,
\]%
\[
\]%
\[
U_{3}=\left(
\begin{array}
[c]{cc}%
\begin{array}
{c}%
-\frac{1}{2}b_{4}x^{3}-\frac{1}{2}(3b_{4}t_{3}+b_{3})x^{2}-\frac{1}{2}%
[b_{4}(3t_{3}^{2}+2t_{2})+2b_{3}t_{3}+b_{2}+p_{2}]x\\
-\frac{1}{2}[b_{4}(t_{3}^{3}+3t_{2}t_{3}+t_{1}-q_{3})+b_{3}(t_{3}^{2}%
+t_{2})+b_{2}t_{3}+b_{1}+p_{1}+q_{1}p_{2}]
\end{array}
& \frac{1}{2}x^{2}+\frac{1}{2}q_{1}x+\frac{1}{2}q_{2}\\%
\begin{array}
{c}%
\\
-b_{4}p_{2}x^{2}-[b_{4}(p_{1}+3t_{3}p_{2})+b_{3}p_{2}]x-b_{4}[(3t_{3}%
p_{1}+2t_{2}p_{2}+3t_{3}^{2}p_{2})+p_{0}]\\
-b_{3}(p_{1}+2t_{3}p_{2})-b_{2}p_{2}-\frac{1}{2}p_{1}^{2}+\overline{b}%
\end{array}
&
\begin{array}
{c}%
\\
-(U_{3})_{11}%
\end{array}
\end{array}
\right)  .
\]

For $m=2$ we have $\mathfrak{a}$ $=\mathfrak{g}$ so $H_{r}^{B}=h_{r}^{B}$ and
$\overline{U}_{r}=U_{r}$ for all $r$ and so%
\[
h_{r}^{B}=\ \mathcal{E}_{r}+b_{4}M_{r}^{(4)}+(b_{3}+2b_{4}t_{2})M_{r}%
^{(3)}+[b_{2}+b_{3}t_{2}+b_{4}(t_{2}^{2}+t_{1})]M_{r}^{(2)}+b_{1}M_{r}%
^{(1)}+b_{0}{\mathrm{e}^{t_{{3}}}}M_{r}^{(0)},
\]
where
\begin{align}
\mathcal{E}_{1}=  &  \ \tfrac{1}{2}\,{p}_{1}^{2}-\tfrac{1}{2}q_{{2}}\,{p}%
_{2}^{2}-q_{{3}}p_{{2}}p_{{3}},\nonumber\\
\mathcal{E}_{2}=  &  \ -q_{{2}}p_{{1}}p_{{2}}-q_{{3}}p_{{1}}p_{{3}}-q_{{1}%
}q_{{3}}p_{{2}}p_{{3}}-\tfrac{1}{2}\left(  q_{{1}}q_{{2}}+\,q_{{3}}\right)
{p}_{2}^{2}+p_{1},\label{geo7}\\
\mathcal{E}_{3}=  &  \ -q_{{3}}p_{{1}}p_{{2}}-\tfrac{1}{2}\,q_{{1}}q_{{3}}%
{p}_{2}^{2}+\tfrac{1}{2}\,{q}_{3}^{2}{p}_{3}^{2}+q_{3}p_{3}.\nonumber
\end{align}
The Lax matrix $L$ is given by (\ref{48}) with
\[
L_{11}=-b_{4}x^{4}-(2b_{4}t_{2}+b_{3})x^{3}-[b_{4}(t_{2}^{2}+t_{1})+b_{3}%
t_{2}+b_{2}+p_{1}]x^{2}+(q_{2}p_{2}+q_{3}p_{3}-b_{1})x+q_{3}p_{2}%
-b_{0}{\mathrm{e}^{t_{{3}}}},
\]%
\[
L_{21}=-2b_{4}p_{1}x^{3}-2[b_{4}(p_{0}+2t_{2}p_{1})+b_{3}p_{1}-\overline
{b}]x^{2}+2(b_{1}p_{2}+b_{0}{\mathrm{e}^{t_{{3}}}}p_{3}-\tfrac{1}{2}q_{2}%
p_{2}^{2}-q_{3}p_{2}p_{3})x-q_{3}p_{2}^{2}+2b_{0}{\mathrm{e}^{t_{{3}}}}p_{2}%
\]
while
\[
U_{1}=\left(
\begin{array}
[c]{cc}%
-\frac{1}{2}b_{4}x+\frac{1}{2}[b_{4}(q_{1}-2t_{2})-b_{3}] & \frac{1}{2}\\%
\begin{array}
{c}%
\\
-b_{4}p_{1}%
\end{array}
&
\begin{array}
{c}%
\\
-(U_{1})_{11}%
\end{array}
\end{array}
\right)  ,
\]%
\[
\]%
\[
U_{2}=\left(
\begin{array}
[c]{cc}%
-\frac{1}{2}b_{4}x^{2}-\frac{1}{2}(2b_{4}t_{2}+b_{3})x-\frac{1}{2}[b_{4}%
(t_{2}^{2}+t_{1}-q_{2})+b_{3}t_{2}+b_{2}+p_{1}] & \frac{1}{2}x+\frac{1}%
{2}q_{1}\\%
\begin{array}
{c}%
\\
-b_{4}p_{1}x-b_{4}(p_{0}+2t_{2}p_{1})-b_{3}p_{1}+\overline{b}%
\end{array}
&
\begin{array}
{c}%
\\
-(U_{2})_{11}%
\end{array}
\end{array}
\right)  ,
\]%
\[
\]%
\[
U_{3}=\left(
\begin{array}
[c]{cc}%
\frac{1}{2}b_{4}q_{3}+\frac{1}{2}(b_{0}{\mathrm{e}^{t_{{3}}}}-q_{3}%
p_{2})x^{-1}] & -\frac{1}{2}q_{3}x^{-1}\\%
\begin{array}
{c}%
\\
-\frac{1}{2}(2b_{0}{\mathrm{e}^{t_{{3}}}}-q_{3}p_{2})p_{2}x^{-1}%
\end{array}
&
\begin{array}
{c}%
\\
-(U_{3})_{11}%
\end{array}
\end{array}
\right)  .
\]

For $m=3$, we have $H_{r}^{B}=h_{r}^{B},$ $\overline{U}_{r}=U_{r},\ r=1,3,$
$H_{2}^{B}=h_{2}^{B}+t_{3}h_{3}^{B}$, $\overline{U}_{2}=U_{2}+t_{3}U_{3}$\ and
our procedure yields
\[
h_{r}^{B}=\ \mathcal{E}_{r}+b_{4}M_{r}^{(4)}+(b_{3}+b_{4}t_{1})M_{r}%
^{(3)}+b_{2}M_{r}^{(2)}+(b_{0}t_{3}{\mathrm{e}^{2t_{{2}}}}+b_{1}%
{\mathrm{e}^{t_{{2}}}})M_{r}^{(1)}+b_{0}{\mathrm{e}^{2t_{{2}}}}M_{r}^{(0)},
\]
where
\begin{align}
\mathcal{E}_{1} &  =-\tfrac{1}{2}\,q_{{1}}{p}_{1}^{2}-q_{{2}}p_{{1}}p_{{2}%
}-q_{{3}}p_{{1}}p_{{3}}-\tfrac{1}{2}q_{3}\,{p}_{2}^{2},\label{geo8}\\
\mathcal{E}_{2} &  =-\tfrac{1}{2}q_{{2}}{p}_{1}^{2}-q_{{3}}p_{{1}}p_{{2}%
}+\tfrac{1}{2}\left(  \,{q}_{2}^{2}-\,q_{{1}}q_{{3}}\right)  {p}_{2}%
^{2}+q_{{2}}q_{{3}}p_{{2}}p_{{3}}+\ \tfrac{1}{2}\,q_{3}^{2}{p}_{3}^{2}%
+q_{3}p_{2},\nonumber\\
\mathcal{E}_{3} &  =-\tfrac{1}{2}\,q_{{3}}{p}_{1}^{2}+\tfrac{1}{2}q_{{2}%
}q_{{3}}\,{p}_{2}^{2}+{q}_{3}^{2}p_{{2}}p_{{3}}+q_{2}p_{2}+2q_{3}%
p_{3}.\nonumber
\end{align}
The Lax matrix $L$ is given by (\ref{48}) with
\[
L_{11}=-b_{4}x^{4}-(b_{4}t_{1}+b_{3})x^{3}-(p_{0}+b_{2})x^{2}-(b_{1}%
{\mathrm{e}^{t_{{2}}}}+b_{0}t_{3}{\mathrm{e}^{2t_{{2}}}}-q_{2}p_{1}-q_{3}%
p_{2})x-b_{0}{\mathrm{e}^{2t_{{2}}}}+q_{3}p_{1},
\]%
\begin{align*}
L_{21} &  =2(-b_{4}p_{0}+\overline{b})x^{3}+2[b_{2}p_{1}+b_{1}{\mathrm{e}%
^{t_{{2}}}}p_{2}+b_{0}{\mathrm{e}^{2t_{{2}}}}(t_{3}p_{2}+p_{3})+p_{1}%
p_{0}+\tfrac{1}{2}q_{1}p_{1}^{2}-\tfrac{1}{2}q_{3}p_{2}^{2}]x^{2}\\
&  \ \ \ +2[b_{1}{\mathrm{e}^{t_{{2}}}}p_{1}+b_{0}{\mathrm{e}^{2t_{{2}}}%
}(t_{3}p_{1}+p_{2})-\tfrac{1}{2}q_{2}p_{1}^{2}-q_{3}p_{1}p_{2}])x-q_{3}%
p_{1}^{2}+2b_{0}{\mathrm{e}^{2t_{{2}}}}p_{1}%
\end{align*}
while
\[
U_{1}=\left(
\begin{array}
[c]{cc}%
-\frac{1}{2}b_{4}x+\frac{1}{2}[b_{4}(q_{1}-t_{1})-b_{3}] & \frac{1}{2}\\%
\begin{array}
{c}%
\\
-b_{4}p_{0}+\overline{b}%
\end{array}
&
\begin{array}
{c}%
\\
-(U_{1})_{11}%
\end{array}
\end{array}
\right)  ,
\]%
\[
\]%
\[
U_{2}=\left(
\begin{array}
[c]{cc}%
\begin{array}
{c}%
\frac{1}{2}b_{4}q_{2}+\frac{1}{2}(b_{1}{\mathrm{e}^{t_{{2}}}}+b_{0}%
t_{3}{\mathrm{e}^{2t_{{2}}}}-q_{2}p_{1}-q_{3}p_{2})x^{-1}\\
+\frac{1}{2}(b_{0}{\mathrm{e}^{2t_{{2}}}}-q_{3}p_{1})x^{-2}%
\end{array}
& -\frac{1}{2}q_{2}x^{-1}-\frac{1}{2}q_{3}x^{-2}\\%
\begin{array}
{c}%
\\
-[b_{1}{\mathrm{e}^{t_{{2}}}}p_{1}+b_{0}{\mathrm{e}^{2t_{{2}}}}(t_{3}%
p_{1}+p_{2})-\frac{1}{2}q_{2}p_{1}^{2}-q_{3}p_{1}p_{2}]x^{-1}\\
-(b_{0}{\mathrm{e}^{2t_{{2}}}}p_{1}-\frac{1}{2}q_{3}p_{1}^{2})x^{-2}%
\end{array}
&
\begin{array}
{c}%
\\
-(U_{2})_{11}%
\end{array}
\end{array}
\right)  ,
\]%
\[
\]%
\[
U_{3}=\left(
\begin{array}
[c]{cc}%
\frac{1}{2}b_{4}q_{3}+\frac{1}{2}(b_{0}{\mathrm{e}^{2t_{{2}}}}-q_{3}%
p_{1})x^{-1}] & -\frac{1}{2}q_{3}x^{-1}\\%
\begin{array}
{c}%
\\
-\frac{1}{2}(2b_{0}{\mathrm{e}^{2t_{{2}}}}-q_{3}p_{1})p_{1}x^{-1}%
\end{array}
&
\begin{array}
{c}%
\\
-(U_{3})_{11}%
\end{array}
\end{array}
\right)  .
\]

Finally, for $m=4$ we have again $\mathfrak{a}$ $=\mathfrak{g}$ so $H_{r}%
^{B}=h_{r}^{B}$ and $\overline{U}_{r}=U_{r}$ for all $r$ and our procedure
yields%
\begin{equation}
h_{r}^{B}=\ \mathcal{E}_{r}+b_{4}{\mathrm{e}^{t_{{1}}}}M_{r}^{(4)}+b_{3}%
M_{r}^{(3)}+[b_{2}+b_{1}t_{2}+b_{0}(t_{2}^{2}+t_{3})]M_{r}^{(2)}+(b_{1}%
+2b_{0}t_{2})M_{r}^{(1)}+b_{0}M_{r}^{(0)},\nonumber
\end{equation}
where
\begin{align}
\mathcal{E}_{1}  &  =\tfrac{1}{2}p_{0}^{2}-\tfrac{1}{2}q_{2}p_{1}^{2}%
-q_{3}p_{1}p_{2},\nonumber\\
\mathcal{E}_{2}  &  =-q_{2}p_{0}p_{1}-q_{3}p_{0}p_{2}\ -\tfrac{1}{2}\left(
q_{{1}}\,q_{{2}}+\,q_{{3}}\right)  {p}_{1}^{2}-q_{{1}}q_{{3}}\,p_{{1}}p_{{2}%
}+q_{2}p_{1}+2q_{3}p_{2},\label{geo9}\\
\mathcal{E}_{3}  &  =-q_{3}\,p_{0}p_{{1}}-\tfrac{1}{2}q_{1}q_{3}p_{1}%
^{2}+\tfrac{1}{2}q_{3}^{2}p_{2}^{2}+q_{3}p_{1}.\nonumber
\end{align}
The Lax matrix $L$ is given by (\ref{48}) with
\begin{align*}
L_{11}  &  =-b_{4}{\mathrm{e}^{t_{{1}}}}x^{4}-b_{3}x^{3}-[b_{2}+b_{1}%
t_{2}+b_{0}(t_{2}^{2}+t_{3})-q_{1}p_{0}-q_{2}p_{1}-q_{3}p_{2}]x^{2}\\
&  \ \ \ \ -(b_{1}+2b_{0}t_{2}-q_{2}p_{0}-q_{3}p_{1})x+q_{3}p_{0}-b_{0},
\end{align*}%
\begin{align*}
L_{21}  &  =2\overline{b}{\mathrm{e}^{t_{{1}}}}x^{4}+2[b_{3}p_{0}+b_{2}%
p_{1}+b_{1}(t_{2}p_{1}+p_{2})+b_{0}((t_{2}^{2}+t_{3})p_{1}+2t_{2}p_{2}%
+p_{3})+\tfrac{1}{2}p_{0}^{2}-\tfrac{1}{2}q_{2}p_{1}^{2}-q_{3}p_{1}p_{2}%
]x^{3}\\
&  \ \ \ +2[b_{2}p_{0}+b_{1}(t_{2}p_{0}+p_{1})+b_{0}((t_{2}^{2}+t_{3}%
)p_{0}+2t_{2}p_{1}+p_{2})-\tfrac{1}{2}q_{1}p_{0}^{2}-(q_{2}p_{1}+q_{3}%
p_{2})p_{0}-\tfrac{1}{2}q_{3}p_{1}^{2}]x^{2}\\
&  \ \ \ +2[b_{1}p_{0}+b_{0}(2t_{2}p_{0}+p_{1})-\tfrac{1}{2}q_{2}p_{0}%
^{2}-q_{3}p_{1}p_{0}]x+2b_{0}p_{0}-q_{3}p_{0}^{2}%
\end{align*}
while
\[
U_{1}=\left(
\begin{array}
[c]{cc}%
-\frac{1}{2}b_{4}{\mathrm{e}^{t_{{1}}}}x+\frac{1}{2}(b_{4}{\mathrm{e}^{t_{{1}%
}}}q_{1}-b_{3}) & \frac{1}{2}\\%
\begin{array}
{c}%
\\
\overline{b}{\mathrm{e}^{t_{{1}}}}x+b_{3}p_{0}+b_{2}p_{1}+b_{1}(t_{2}%
p_{1}+p_{2})+b_{0}[(t_{2}^{2}+t_{3})p_{1}+2t_{2}p_{2}+p_{3}]\\
+\frac{1}{2}p_{0}^{2}-\frac{1}{2}q_{2}p_{1}^{2}-q_{3}p_{1}p_{2}-\overline
{b}{\mathrm{e}^{t_{{1}}}}q_{1}%
\end{array}
&
\begin{array}
{c}%
\\
-(U_{1})_{11}%
\end{array}
\end{array}
\right)  ,
\]%
\[
\]%
\[
U_{2}=\left(
\begin{array}
[c]{cc}%
\frac{1}{2}b_{4}{\mathrm{e}^{t_{{1}}}}q_{2}+\frac{1}{2}(b_{1}+2b_{0}%
t_{2}-q_{2}p_{0}-q_{3}p_{1})x^{-1}+\frac{1}{2}(b_{0}-q_{3}p_{0})x^{-2} &
-\frac{1}{2}q_{2}x^{-1}-\frac{1}{2}q_{3}x^{-2}\\%
\begin{array}
{c}%
\\
-\overline{b}{\mathrm{e}^{t_{{1}}}}q_{2}-[b_{1}p_{0}+b_{0}(2t_{2}p_{0}%
+p_{1})-\frac{1}{2}q_{2}p_{0}^{2}-q_{3}p_{1}p_{0}]x^{-1}\\
-(b_{0}p_{0}-\frac{1}{2}q_{3}p_{0}^{2})x^{-2}%
\end{array}
&
\begin{array}
{c}%
\\
-(U_{2})_{11}%
\end{array}
\end{array}
\right)  ,
\]%
\[
\]%
\[
U_{3}=\left(
\begin{array}
[c]{cc}%
\frac{1}{2}b_{4}{\mathrm{e}^{t_{{1}}}}q_{3}+\frac{1}{2}(b_{0}-q_{3}%
p_{0})x^{-1} & -\frac{1}{2}q_{3}x^{-1}\\%
\begin{array}
{c}%
\\
-\overline{b}{\mathrm{e}^{t_{{1}}}}q_{3}-(b_{0}p_{0}-\frac{1}{2}q_{3}p_{0}%
^{2})x^{-1}%
\end{array}
&
\begin{array}
{c}%
\\
-(U_{1})_{11}%
\end{array}
\end{array}
\right)  .
\]

\section{Painlev\'{e}-type systems with ordinary potentials\label{S5}}

In Part I \cite{part1} we constructed Frobenius integrable non-autonomous
Hamiltonian systems with ordinary potentials, generated by the following
spectral curve%

\begin{equation}
\sum_{r=1}^{n}h_{r}^{\prime}x^{n-r}=\frac{1}{2}x^{m}y^{2}-\sum_{\alpha
=-m}^{2n-m+2}c_{\alpha}(t)x^{\alpha}\equiv\Psi(x,y,t),\ \ \ \ \ \ \ \ m\in
\{0,...,n+1\},\quad\label{sA}%
\end{equation}
where $c_{\alpha}(t)=c_{\alpha}(t_{1},\ldots,t_{n}).$ This construction
procedure is analogous to the procedure applied in Section \ref{S3} for the
case with magnetic potentials. First, taking $n$ copies of (\ref{sA}) with
$(x,y)=(\lambda_{i},\mu_{i}^{\prime})$, $i=1,\dotsc,n$, we find
\begin{equation}
h_{r}^{\prime}(\lambda,\mu^{\prime})=\frac{1}{2}\mu^{\prime T}K_{r}%
G\mu^{\prime}+\sum_{\alpha=-m}^{2n-m+2}c_{\alpha}(t)V_{r}^{(\alpha)},\quad
r=1,\dotsc,n. \label{h'}%
\end{equation}
In the next step we perturb the Hamiltonians $h_{r}^{\prime},$ defined through
(\ref{h'}), to the quasi-St\"{a}ckel Hamiltonians $h_{r}^{A}=h_{r}^{\prime
}+W_{r}^{\prime},$ where $W_{r}^{\prime}=$ $W_{r}(\lambda,\mu^{\prime})$ with
$W_{r}$ given by (\ref{j1})-(\ref{j2}). Next, we deform the Hamiltonians
$h_{r}^{A}$ to the Frobenius integrable Hamiltonians $H_{r}^{A}$ through the
deformations (\ref{7b})-(\ref{7c}). In Part I \cite{part1} we also proved that
the Hamiltonians $H_{r}^{A}$ with ordinary potentials and $H_{r}^{B}$ with
magnetic potentials are related by the multitime-dependent canonical
transformation (rational symplectic transformation \cite{Iwasaki})%

\begin{equation}
\lambda_{i}^{\prime}=\frac{\partial F(\lambda,\mu^{\prime},t)}{\partial\mu
_{i}^{\prime}}=\lambda_{i},\quad\mu_{i}=\frac{\partial F(\lambda,\mu^{\prime
},t)}{\partial\lambda_{i}}=\mu_{i}^{\prime}+\sum_{\gamma=0}^{n+1}d_{\gamma
}(t)\lambda_{i}^{\gamma-m},\quad i=1,\dotsc,n, \label{7.3}%
\end{equation}
generated by
\begin{equation}
F(\lambda,\mu^{\prime},t)=\sum_{i=1}^{n}\left(  \lambda_{i}\mu_{i}^{\prime
}+\sum_{\gamma=0,\gamma\neq m-1}^{n+1}\frac{1}{\gamma-m+1}d_{\gamma}%
(t)\lambda_{i}^{\gamma-m+1}+d_{m-1}(t)\ln\lambda_{i}\right)  , \label{7.4}%
\end{equation}
provided that the functions $c_{\alpha}$ and $d_{\gamma},e$ are related by the
polynomial (w.r.t $x$) identity%

\begin{equation}
\sum_{\alpha=-m}^{2n-m+2}c_{\alpha}(t)x^{\alpha}=\frac{1}{2}x^{m}\left(
\sum_{\gamma=0}^{n+1}d_{\gamma}(t)x^{\gamma-m}\right)  ^{2}+\left(
e(t)-d_{n+1}(t)\right)  x^{n}. \label{7.10}%
\end{equation}
It means, that in the above notation and up to terms independent of the
coordinates on $\mathcal{M}$:%
\begin{equation}
H_{r}^{A}(\lambda,\mu^{\prime},t)=H_{r}^{B}(\lambda,\mu^{\prime}%
,t)+\frac{\partial F(\lambda,\mu^{\prime},t)}{\partial t_{r}},\text{
\ \ }r=1,\ldots,n. \label{7.5partI}%
\end{equation}
For details of this construction, see Part I \cite{part1}.

As a result of the above considerations, we obtain the following corollary

\begin{corollary}
Each non-autonomous Hamiltonian flow on $\mathcal{M}$%
\begin{equation}
\frac{d\xi}{dt_{r}}=Y_{r}^{A}(\xi,t)=\pi dH_{r}^{A}(\xi,t) \label{5A}%
\end{equation}
(cf. (\ref{5})) has the isomonodromic Lax representation
\[
\frac{d}{dt_{r}}L(x,\xi,t)=[\overline{U}_{r}(x,\xi,t),L(x,\xi,t)]+2x^{m}%
\frac{\partial}{\partial x}\overline{U}_{r}(x,\xi,t)
\]
with the evolutionary derivative given by%
\begin{equation}
\frac{d}{dt_{r}}=\frac{\partial}{\partial t_{r}}+\{\cdot,H_{r}^{A}\}
\label{czasA}%
\end{equation}
(cf. (\ref{czasB})), with the Lax matrix $L(x)$ given by%
\begin{equation}
L(x,\xi,t)=%
\begin{pmatrix}
v(x) & u(x)\\
w(x,t) & -v(x)
\end{pmatrix}
, \label{LA}%
\end{equation}
$u(x)$ given by (\ref{L2}), $v(x)$ given by (\ref{L3}) with $g(\lambda
_{i})=f(\lambda_{i})=\lambda_{i}^{m}$,
\begin{equation}
w(x,t)=-2x^{m}\left[  \frac{\Psi(x,v(x)x^{-m},t)}{u(x)}\right]  _{+}.
\label{wA}%
\end{equation}
with $\mu$ replaced by $\mu^{\prime}$ and with $\Psi(x,y,t)$ given in
(\ref{sA}). Finally, the matrices $\overline{U}_{r}$ are given by
(\ref{7a})-(\ref{11}) with the same functions $\zeta_{r,j}$ and $\zeta
_{r,r+j}$ as for the corresponding (i.e. with the same $n$ and $m$) magnetic flow.
\end{corollary}

The system (\ref{5A}) can be considered as \emph{the non-magnetic
representation} of the corresponding (i.e. with the same $n$ and $m$) system
(\ref{5}). In the sequel we will omit $^{\prime}$ at $\mu$ when writing our
systems in the non-magnetic representation (\ref{5A}).

In order to find an explicit form of the part of the Lax element
$L_{21}=w(x,t)$ in (\ref{LA}) that is generated by the ordinary potentials in
(\ref{wA}) (i.e. by the term $%
%TCIMACRO{\tsum \nolimits_{\alpha}}%
%BeginExpansion
{\textstyle\sum\nolimits_{\alpha}}
%EndExpansion
c_{\alpha}x^{\alpha}$ in $\Psi$; note that the operation $\left[  \frac{\cdot
}{\cdot}\right]  _{+}$ defined in (\ref{4a}) is linear) we need the following lemma.

\begin{lemma}
\label{4L}

\begin{itemize}
\item[(i)] For $s\in\mathbb{N}$%
\begin{equation}
\left[  \frac{x^{n+s}}{u(x)}\right]  _{+}=-\sum_{r=0}^{s}V_{1}^{(n+r-1)}%
\lambda^{s-r}=-\sum_{r=0}^{s}V_{1}^{(n+s-r-1)}x^{r}. \label{23b}%
\end{equation}

\item[(ii)] For $s\in\mathbb{N}$%
\begin{equation}
\left[  \frac{x^{-s}}{u(x)}\right]  _{+}=\sum_{r=1}^{s}V_{1}^{(-r)}%
\lambda^{-s+r-1}=\sum_{r=1}^{s}V_{1}^{(-s+r-1)}x^{-r}. \label{24b}%
\end{equation}

\end{itemize}
\end{lemma}

\begin{proof}
The basic ordinary potentials satisfy the following recursion relations:%
\begin{equation}
V_{k}^{(r+1)}=V_{k+1}^{(r)}-\rho_{k}V_{1}^{(r)},\ \ \ \ \ \ r\in\mathbb{N},
\label{9a}%
\end{equation}%
\begin{equation}
V_{k}^{(-r-1)}=V_{r-1}^{(-r)}-\frac{\rho_{r-1}}{\rho_{n}}V_{n}^{(-r)}%
,\ \ \ \ \ \ r\in\mathbb{N}, \label{9b}%
\end{equation}
that follow directly from (\ref{5d})-(\ref{6a}). The proof of this Lemma is by
induction with respect to $s$.

\begin{itemize}
\item[(i)] Due to (\ref{5d}), (\ref{6a}) we have $V_{1}^{(n-1)}=-1$ and
$V_{k}^{(n)}=\rho_{k}$. So, according to (\ref{L2})
\[
\frac{x^{n}}{u(x)}=1-\frac{\sum_{k=1}^{n}\rho_{k}x^{n-k}}{u(x)}=-V_{1}%
^{(n-1)}-\frac{\sum_{k=1}^{n}V_{k}^{(n)}x^{n-k}}{u(x)}.
\]
Assume now that for a fixed $s\in\mathbb{N}$%
\begin{equation}
\frac{x^{n+s}}{u(x)}=-\sum_{r=0}^{s}V_{1}^{(n+r-1)}x^{s-r}-\frac{\sum
_{k=1}^{n}V_{k}^{(n+s)}x^{n-k}}{u(x)}. \label{wzorek1}%
\end{equation}
Then
\begin{align*}
x\frac{x^{n+s}}{u(x)}  &  =-x\sum_{r=0}^{s}V_{1}^{(n+r-1)}x^{s-r}-\frac
{x\sum_{k=1}^{n}V_{k}^{(n+s)}x^{n-k}}{u(x)}\\
&  =-\sum_{r=0}^{s}V_{1}^{(n+r-1)}x^{s-r+1}-V_{1}^{(n+s)}-\frac{\sum_{k=1}%
^{n}\left(  V_{k+1}^{(n+s)}-\rho_{k}V_{1}^{(n+s)}\right)  x^{n-k}}{u(x)}\\
&  \overset{(\ref{9a})}{=}-\sum_{r=0}^{s+1}V_{1}^{(n+r-1)}x^{s-r+1}-\frac
{\sum_{k=1}^{n}V_{k}^{(n+s+1)}x^{n-k}}{u(x)}.
\end{align*}
Thus, by induction, (\ref{wzorek1}) is true for any $s\in\mathbb{N\,}$and
using the definition of $\left[  \frac{\cdot}{\cdot}\right]  _{+}$ in
(\ref{L4a}) we obtain (\ref{23b}).

\item[(ii)] By (\ref{5d}), (\ref{6a}) we have, $V_{1}^{(-1)}=\frac{1}{\rho
_{n}}$ and $V_{k}^{(-1)}=\frac{\rho_{k-1}}{\rho_{n}}$, . So, due to (\ref{L2})%
\[
\frac{x^{-1}}{u(x)}=\frac{1}{\rho_{n}}x^{-1}-\frac{\sum_{k=1}^{n}\frac
{\rho_{k-1}}{\rho_{n}}x^{n-k}}{u(x)}=V_{1}^{(-1)}x^{-1}-\frac{\sum_{k=1}%
^{n}V_{k}^{(-1)}x^{n-k}}{u(x)}.
\]
Assume now that for a fixed $s\in\mathbb{N}$%
\begin{equation}
\frac{x^{-s}}{u(x)}=\sum_{r=1}^{s}V_{1}^{(-r)}x^{-s+r-1}-\frac{\sum_{k=1}%
^{n}V_{k}^{(-s)}x^{n-k}}{u(x)}. \label{wzorek2}%
\end{equation}
Then
\begin{align*}
x^{-1}\frac{x^{-s}}{u(x)}  &  =x^{-1}\sum_{r=1}^{s}V_{1}^{(-r)}x^{-s+r-1}%
-\frac{x^{-1}\sum_{k=1}^{n}V_{k}^{(-s)}x^{n-k}}{u(x)}\\
&  =\sum_{r=1}^{s}V_{1}^{(-r)}x^{-(s+1)+r-1}-\frac{\sum_{k=1}^{n-1}%
V_{k}^{(-s)}x^{n-k-1}}{u(x)}-\frac{V_{n}^{(-s)}x^{-1}}{u(x)}\\
&  =\sum_{r=1}^{s}V_{1}^{(-r)}x^{-(s+1)+r-1}+\frac{V_{n}^{(-s)}x^{-1}}%
{\rho_{n}}-\frac{\sum_{k=1}^{n}\left(  V_{k-1}^{(-s)}-\frac{\rho_{k-1}}%
{\rho_{n}}V_{n}^{(-s)}\right)  x^{n-k}}{u(x)}\\
&  \overset{(\ref{9b})}{=}\sum_{r=1}^{s+1}V_{1}^{(-r)}x^{-(s+1)+r-1}%
-\frac{\sum_{k=1}^{n}V_{k}^{(-s-1)}x^{n-k}}{u(x)}.
\end{align*}
Thus, (\ref{wzorek2}) is valid for any $s\in\mathbb{N}$ and it implies, by the
definition of $\left[  \frac{\cdot}{\cdot}\right]  _{+}$ in (\ref{L4a}), that
(\ref{24b}) is true.
\end{itemize}
\end{proof}

\begin{example}
\label{2e}The system from Example \ref{1e} (i.e. given by $n=3,$ $m=1,$
$b_{0}=b_{1}=b_{2}=b_{4}=\overline{b}=0$) has in the non-magnetic
representation (see Example 2 in Part I) the form $H_{r}^{A}=h_{r}^{A},$
$\overline{U}_{r}=U_{r},\ r=1,2,$ $H_{3}^{A}=h_{3}^{A}+t_{2}h_{1}^{A}$,
$\overline{U}_{3}=U_{3}+t_{2}U_{1}$\ with $h_{r}^{A}$ given by (up to terms
independent on coordinates on $\mathcal{M}$)%
\[
h_{r}^{A}=\ \mathcal{E}_{r}+2a_{5}(3t_{3}^{2}+t_{2})V_{r}^{(3)}+4a_{5}%
t_{3}V_{r}^{(4)}+a_{5}V_{r}^{(5)}%
\]
with $a_{5}=\frac{1}{2}b_{3}^{2}$ (this follows from the map (\ref{7.10}))
where the geodesic quasi-St\"{a}ckel Hamiltonians $\mathcal{E}_{r}$\ are the
same and are given by (\ref{geo6a}) while the ordinary potentials
\ $V_{r}^{(\alpha)}$\ are given by (\ref{5d}) and (\ref{6a}). Explicitly, in
Vi\`{e}te coordinates:
\[
V^{(5)}=\left(
\begin{array}
{c}%
{q_{{1}}^{3}}-2\,q_{{1}}q_{{2}}+{q_{{3}}}\\
{q_{{1}}^{2}}q_{{2}}-\,q_{{1}}{q_{{3}}}-\,q_{{2}}^{2}\\
{q_{{1}}^{2}}q_{{3}}-q_{{2}}q_{{3}}%
\end{array}
\right)  ,\text{ \ \ }V^{(4)}=\left(
\begin{array}
{c}%
q_{2}-{q_{{1}}^{2}}\\
{q_{{3}}}-\,q_{1}q_{{2}}\\
-q_{{1}}q_{{3}}%
\end{array}
\right)  ,\text{ \ \ }V^{(3)}=\left(
\begin{array}
[c]{r}%
{q_{{1}}}\\
q_{{2}}\\
{q}_{3}%
\end{array}
\right)
\]
Further, the Lax matrix $L$ takes the form%
\[
L=\left(
\begin{array}
[c]{cc}%
-p_{2}x^{2}-(q_{1}p_{2}+p_{1})x+q_{3}p_{3} & x^{3}+q_{1}x^{2}+q_{2}x+q_{3}\\%
\begin{array}
{c}%
\\
2a_{5}x^{3}-2a_{5}(q_{1}-4t_{3})x^{2}\\
+\left[  2a_{5}\left(  q_{1}^{2}-4t_{3}q_{1}-q_{2}+6t_{3}^{2}+2t_{2}\right)
-p_{2}^{2}\right]  x-q_{3}p_{3}^{2}%
\end{array}
&
\begin{array}
{c}%
\\
-L_{11}%
\end{array}
\end{array}
\right)
\]
(here and in the sequel we will omit $^{\prime}$ at $p^{\prime}$ when writing
down the isomonodromic Lax representation in the non-magnetic case) while
\[
U_{1}=\left(
\begin{array}
[c]{cc}%
0 & \frac{1}{2}\\
a_{5} & 0
\end{array}
\right)  ,\ \ \ \ U_{2}=\left(
\begin{array}
[c]{cc}%
-\frac{1}{2}p_{2} & \frac{1}{2}x+\frac{1}{2}q_{1}\\%
\begin{array}
{c}%
\\
a_{5}\left(  x-q_{1}+4t_{3}\right)
\end{array}
&
\begin{array}
{c}%
\\
\frac{1}{2}p_{2}%
\end{array}
\end{array}
\right)
\]%
\[
\]%
\[
U_{3}=\left(
\begin{array}
[c]{cc}%
\begin{array}
{c}%
-\frac{1}{2}p_{2}x-\frac{1}{2}q_{1}p_{2}-\frac{1}{2}p_{1}%
\end{array}
& \frac{1}{2}x^{2}+\frac{1}{2}q_{1}x+\frac{1}{2}q_{2}\\%
\begin{array}
{c}%
\\
a_{5}x^{2}-a_{5}(q_{1}-4t_{3})x+a_{5}\left(  q_{1}^{2}-4t_{3}q_{1}%
-q_{2}+6t_{3}^{2}+2t_{2}\right)  -\frac{1}{2}p_{2}^{2}%
\end{array}
&
\begin{array}
{c}%
\\
-(U_{3})_{11}%
\end{array}
\end{array}
\right)
\]
A direct calculation confirms that the matrices $L,\overline{U}_{1}%
,\overline{U}_{2},\overline{U}_{3}$ do satisfy the isomonodromic Lax equation
(\ref{6}) with the time derivative given by (\ref{czasA}). The explicit form
of the multitime-dependent transformation (\ref{7.3}) between both systems has
been presented in Example 5 in Part I.
\end{example}

\section{Isomonodromic Lax representations for one-, two- and
three-dimensional Painlev\'{e}-type systems in the ordinary
representation\label{S6}}

In Sections $5$ and $9$ of Part I we presented a complete list of all one-,
two- and three-dimensional non-autonomous Frobenius integrable systems
originating from our deformation procedure in the non-magnetic representation
(i.e. with ordinary potentials). Here we present these one-, two- and
three-dimensional systems together with their isomonodromic Lax
representations (\ref{6}) in Vi\`{e}te coordinates. All Hamiltonians are given
up to terms independent of the coordinates on $\mathcal{M}$. In particular we
propose below four complete (in the sense explained in Introduction)
Painlev\'{e} hierarchies of $P_{I}-P_{IV}$.

The functions $c_{\alpha}(t)$ are expressed by functions $d_{\gamma}(t)$ and
$e(t)$ through (\ref{7.10}) which - through comparison of coefficients at
equal powers of $x$ - induces the map%
\begin{equation}
(b_{0},\ldots,b_{n+1},\overline{b})\rightarrow(a_{-m},\ldots,a_{2n-m+2})
\label{mapka}%
\end{equation}
between the parameters $(b_{0},\ldots,b_{n+1},\overline{b})$ of the magnetic
representation and the dynamical parameters $(a_{-m},\ldots,a_{-1},$%
$a_{n},\ldots,a_{2n-m+2})$ and the non-dynamical parameters $(a_{0}%
,\ldots,a_{n-1})$ of the non-magnetic (ordinary) representation.

\subsection{One-dimensional systems}

Let us first consider the case $n=1$. In this case $H^{A}=h^{A}$ and
$\overline{U}=U$ for each $m=0,...,2$ and we obtain three $4$-parameter
families of the related Painlev\'{e}-type systems.

For $m=0$ we get
\[
h^{A}=\tfrac{1}{2}p^{2}-a_{4}q^{4}+a_{3}q^{3}-(2a_{4}t+a_{2})q^{2}%
+(a_{3}t+a_{1})q,
\]%
\[
\]%
\[
L=\left(
\begin{array}
[c]{cc}%
-p & x+q\\%
\begin{array}
{c}%
\\
2a_{4}x^{3}-2(a_{4}q-a_{3})x^{2}+2[a_{4}(q^{2}+2t)-a_{3}q+a_{2}]x\\
-2a_{4}(q^{3}+2tq)+2a_{3}(q^{2}+t)-2a_{2}q+2a_{1}%
\end{array}
&
\begin{array}
{c}%
\\
p
\end{array}
\end{array}
\right)  ,\ \
\]%
\[
\]%
\[
U=\left(
\begin{array}
[c]{cc}%
\begin{array}
{c}%
0\\
\end{array}
&
\begin{array}
{c}%
\frac{1}{2}\\
\end{array}
\\
a_{4}x^{2}-(2a_{4}q-a_{3})x+a_{4}(3q^{2}+2t)-2a_{3}q+a_{2} & 0
\end{array}
\right)
\]
while the dynamical part of the map (\ref{mapka}) is
\[
a_{4}=\tfrac{1}{2}b_{2}^{2},\quad a_{3}=b_{1}b_{2},\quad a_{2}=\tfrac{1}%
{2}b_{1}^{2}+b_{0}b_{2},\quad a_{1}=b_{0}b_{1}+\overline{b}-b_{2}.
\]

For $m=1$%
\[
h^{A}=-\tfrac{1}{2}qp^{2}+a_{3}q^{3}-(2a_{3}t+a_{2})q^{2}+(a_{3}t^{2}%
+a_{2}t+a_{1})q+a_{-1}q^{-1},
\]%
\[
\]%
\[
L=\left(
\begin{array}
[c]{cc}%
q\,p & x+q\\%
\begin{array}
{c}%
\\
2a_{3}x^{3}+2[a_{3}(2t-q)+a_{2}]x^{2}+2[a_{3}(q-t)^{2}+a_{2}(t-q)+a_{1}%
]x+2a_{-1}q^{-1}-qp^{2}%
\end{array}
&
\begin{array}
{c}%
\\
-q\,p
\end{array}
\end{array}
\right)  ,\
\]

\[
\]%
\[
\ \ U=\left(
\begin{array}
[c]{cc}%
0 & \frac{1}{2}\\%
\begin{array}
{c}%
\\
a_{3}x^{2}+[2a_{3}(t-q)+a_{2}]x+a_{3}(3q^{2}-4tq+t^{2})-a_{2}(2q-t)+a_{1}%
\end{array}
&
\begin{array}
{c}%
\\
0
\end{array}
\end{array}
\right)
\]
with the dynamical part of the map (\ref{mapka})%
\[
a_{3}=\tfrac{1}{2}b_{2}^{2},\quad a_{2}=b_{1}b_{2},\quad a_{1}=\tfrac{1}%
{2}b_{1}^{2}+b_{0}b_{2}+\overline{b}-b_{2},\quad a_{-1}=\tfrac{1}{2}b_{0}%
^{2}.
\]

For $m=2$%
\[
h^{A}=\tfrac{1}{2}q^{2}p^{2}-a_{2}{\mathrm{e}^{2t}}q^{2}+a_{1}\,{\mathrm{e}%
^{t}}q+a_{-1}q^{-1}-a_{-2}q^{-2},
\]%
\[
\]%
\begin{align*}
L  &  =\left(
\begin{array}
[c]{cc}%
-q^{2}p & x+q\\%
\begin{array}
{c}%
\\
2a_{2}{\mathrm{e}^{2t}}x^{3}-2(a_{2}{\mathrm{e}^{2t}}q-a_{1}{\mathrm{e}^{t}%
})x^{2}\\
+2(a_{-1}q^{-1}-a_{-2}q^{-2}+\frac{1}{2}q^{2}p^{2})x-q^{3}p^{2}+2a_{-2}q^{-1}%
\end{array}
&
\begin{array}
{c}%
\\
q^{2}p
\end{array}
\end{array}
\right)  ,\ \ \\
&  \ \\
U  &  =\left(
\begin{array}
[c]{cc}%
0 & \frac{1}{2}\\%
\begin{array}
{c}%
\\
a_{2}{\mathrm{e}^{2t}}x^{2}-(2a_{2}{\mathrm{e}^{2t}}q-a_{1}{\mathrm{e}^{t}%
})x\\
+\frac{1}{2}q^{2}p^{2}+2a_{2}{\mathrm{e}^{2t}}q^{2}-a_{1}{\mathrm{e}^{t}%
}q+a_{-1}q^{-1}-a_{-2}q^{-2}%
\end{array}
&
\begin{array}
{c}%
\\
0
\end{array}
\end{array}
\right)
\end{align*}
with the dynamical part of the map (\ref{mapka})
\[
a_{2}=\tfrac{1}{2}b_{2}^{2},\quad a_{1}=b_{1}b_{2}+\overline{b}-b_{2},\quad
a_{-1}=b_{0}b_{1},\quad a_{-1}=\tfrac{1}{2}b_{0}^{2}.
\]
In particular the above formulas contain the isomonodromic Lax representation
for Painlev\'{e}-I and Painlev\'{e}-II (for $m=0$), Painlev\'{e}-IV (for
$m=1$) and Painlev\'{e}-III (for $m=2$) in the non-magnetic representation.

\subsection{Two-dimensional systems}

In the case of $n=2$, $\mathfrak{g=a}$ for each $m=0,...,3$ and thus
$H_{r}^{A}=h_{r}^{A}$ and $\overline{U}_{r}=U_{r}$ for each $m=0,...,3$. For
each $m$ we obtain a $5$-parameter family of the related Painlev\'{e}-type
systems in ordinary representation (as usual we denote $p_{0}=-q_{1}%
p_{1}-q_{2}p_{2}$).

\bigskip For $m=0$ we get%
\begin{align*}
h_{1}^{A}  &  =p_{1}p_{2}+\tfrac{1}{2}q_{1}p_{2}^{2}+a_{6}(q_{1}^{5}%
-4q_{1}^{3}q_{2}+3q_{1}q_{2}^{2})-a_{5}(q_{1}^{4}-3q_{1}^{2}q_{2}+q_{2}%
^{2})+(4a_{6}t_{2}+a_{4})(q_{1}^{3}-2q_{1}q_{2})\\
&  \ \ \ +(2a_{6}t_{1}+3a_{5}t_{2}+a_{3})(q_{2}-q_{1}^{2})+(4a_{6}t_{2}%
^{2}+2a_{4}t_{2}+a_{5}t_{1}+a_{2})q_{1},\\
& \\
h_{2}^{A}  &  =\tfrac{1}{2}\,{p_{{1}}^{2}+}q_{{1}}p_{{1}}p_{{2}}+\tfrac{1}%
{2}\left(  {q_{{1}}^{2}}-q_{{2}}\right)  {p_{{2}}^{2}+p}_{2}+a_{6}(q_{1}%
^{4}q_{2}-3q_{1}^{2}q_{2}^{2}+q_{2}^{3})+a_{5}(2q_{1}q_{2}^{2}-q_{1}^{3}%
q_{2})+(4a_{6}t_{2}+a_{4})(q_{1}^{2}q_{2}-q_{2}^{2})\\
&  \ \ \ \ -(2a_{6}t_{1}+3a_{5}t_{2}+a_{3})q_{2}q_{1}+(4a_{6}t_{2}^{2}%
+a_{5}t_{1}+2a_{4}t_{2}+a_{2})q_{2}%
\end{align*}%
\[
L=\left(
\begin{array}
[c]{cc}%
-p_{2}x-q_{1}p_{2}-p_{1} & x^{2}+q_{1}x+q_{2}\\%
\begin{array}
{c}%
\\
L_{21}%
\end{array}
&
\begin{array}
{c}%
\\
p_{2}x+q_{1}p_{2}+p_{1}%
\end{array}
\end{array}
\right)  ,
\]
where%
\[
\]%
\begin{align*}
L_{21}  &  =2a_{6}x^{4}-2(a_{6}q_{1}-a_{5})x^{3}+2[a_{6}(q_{1}^{2}%
-q_{2}+4t_{2})-a_{5}q_{1}+a_{4}]x^{2}\\
&  \ \ \ -2[a_{6}(q_{1}^{3}-2q_{1}q_{2}+4t_{2}q_{1}-2t_{1})-a_{5}(q_{1}%
^{2}-q_{2}+3t_{2})+a_{4}q_{1}-a_{3}]x\\
&  \ \ \ +2a_{6}(q_{1}^{4}-3q_{1}^{2}q_{2}+4t_{2}q_{1}^{2}-2t_{1}q_{1}%
+q_{2}^{2}-4t_{2}q_{2}+4t_{2}^{2})-2a_{5}(q_{1}^{3}-2q_{1}q_{2}+3t_{2}%
q_{1}-t_{1})\\
&  \ \ \ +2a_{4}(q_{1}^{2}-q_{2}+2t_{2})-2a_{3}q_{1}+2a_{2}-p_{2}^{2},
\end{align*}%
\[
\]%
\[
U_{1}=\left(
\begin{array}
[c]{cc}%
0 & \frac{1}{2}\\%
\begin{array}
{c}%
\\
a_{6}x^{2}-(2a_{6}q_{1}-a_{5})x+a_{6}(3q_{1}^{2}-2q_{2}+4t_{2})-2a_{5}%
q_{1}+a_{4}%
\end{array}
&
\begin{array}
{c}%
\\
0
\end{array}
\end{array}
\right)  ,
\]%
\[
\]%
\[
U_{2}=\left(
\begin{array}
[c]{cc}%
-\frac{1}{2}p_{2} & \frac{1}{2}x+\frac{1}{2}q_{1}\\%
\begin{array}
{c}%
\\
a_{6}x^{3}-(a_{6}q_{1}-a_{5})x^{2}+[a_{6}(q_{1}^{2}-2q_{2}+4t_{2})-a_{5}%
q_{1}+a_{4}]x\\
-a_{6}(q_{1}^{3}-4q_{1}q_{2}+4t_{2}q_{1}-2t_{1})+a_{5}(q_{1}^{2}-2q_{2}%
+3t_{2})-a_{4}q_{1}+a_{3}%
\end{array}
&
\begin{array}
{c}%
\\
\frac{1}{2}p_{2}%
\end{array}
\end{array}
\right)
\]
while the dynamical part of the map (\ref{mapka}) becomes
\[
a_{6}=\tfrac{1}{2}b_{3}^{2},\ \ \ a_{5}=b_{2}b_{3},\ \ \ a_{4}=b_{1}%
b_{3}+\tfrac{1}{2}b_{2}^{2},\ \ \ a_{3}=b_{1}b_{2}+b_{0}b_{3},\ \ \ a_{2}%
=b_{0}b_{2}+\tfrac{1}{2}b_{1}^{2}-b_{3}+\overline{b}.
\]

For $m=1$ we get
\begin{align*}
h_{1}^{A} &  =\tfrac{1}{2}\,{p}_{1}^{2}-\tfrac{1}{2}\,q_{{2}}{p}_{2}^{2}%
-a_{5}(q_{1}^{4}-3q_{1}^{2}q_{2}+q_{2}^{2})+(4a_{5}t_{2}+a_{4})(q_{1}%
^{3}-2q_{1}q_{2})\\
&  \ \ \ +[2a_{5}(3t_{2}^{2}+t_{1})+3a_{4}t_{2}+a_{3}](q_{2}-q_{1}%
^{2})+[4a_{5}(t_{2}^{3}+t_{1}t_{2})+a_{4}(3t_{2}^{2}+t_{1})+2a_{3}t_{2}%
+a_{2}]q_{1}+a_{-1}q_{2}^{-1},\\
& \\
h_{2}^{A} &  =-q_{{2}}p_{{1}}p_{{2}}-\tfrac{1}{2}\,q_{{1}}q_{{2}}{p}_{2}%
^{2}+p_{1}-a_{5}(q_{1}^{3}q_{2}-2q_{1}q_{2}^{2})+(4a_{5}t_{2}+a_{4})(q_{1}%
^{2}q_{2}-q_{2}^{2})\\
&  \ \ \ -(2a_{5}(3t_{2}^{2}+t_{1})+3a_{4}t_{2}+a_{3})q_{1}q_{2}+(4a_{5}%
(t_{2}^{3}+t_{1}t_{2})+a_{4}(3t_{2}^{2}+t_{1})+2a_{3}t_{2}+a_{2})q_{2}%
+a_{-1}q_{1}q_{2}^{-1}%
\end{align*}%
\[
\]%
\[
L=\left(
\begin{array}
[c]{cc}%
-p_{1}x+q_{2}p_{2} & x^{2}+q_{1}x+q_{2}\\%
\begin{array}
{c}%
\\
L_{21}%
\end{array}
&
\begin{array}
{c}%
\\
p_{1}x-q_{2}p_{2}%
\end{array}
\end{array}
\right)  ,
\]%
\[
\]
with%
\begin{align*}
L_{21} &  =2a_{5}x^{4}+2[a_{5}(4t_{2}-q_{1})+a_{4})x^{3}+2[a_{5}(q_{1}%
^{2}-q_{2}-4t_{2}q_{1}+6t_{2}^{2}+2t_{1})+a_{4}(3t_{2}-q_{1})+a_{3}]x^{2}\\
&  \ \ \ -2[a_{5}(q_{1}^{3}-2q_{1}q_{2}-4t_{2}q_{1}^{2}+(6t_{2}^{2}%
+2t_{1})q_{1}+4t_{2}q_{2}-4(t_{2}^{3}+t_{1}t_{2}))\\
&  \ \ \ -a_{4}(q_{1}^{2}-q_{2}-3t_{2}q_{1}+3t_{2}^{2}+t_{1})+a_{3}%
(q_{1}-2t_{2})-a_{2}]x+2a_{-1}q_{2}^{-1}-q_{2}p_{2}^{2}%
\end{align*}%
\[
\]%
\[
U_{1}=\left(
\begin{array}
[c]{cc}%
0 & \frac{1}{2}\\%
\begin{array}
{c}%
\\
a_{5}x^{2}+[2a_{5}(2t_{2}-q_{1})+a_{4}]x+a_{5}(3q_{1}^{2}-2q_{2}-8t_{2}%
q_{1}+6t_{2}^{2}+2t_{1})-a_{4}(2q_{1}-3t_{2})+a_{3}%
\end{array}
&
\begin{array}
{c}%
\\
0
\end{array}
\end{array}
\right)  ,
\]%
\[
\]%
\[
U_{2}=\left(
\begin{array}
[c]{cc}%
-\frac{1}{2}p_{1} & \frac{1}{2}x+\frac{1}{2}q_{1}\\%
\begin{array}
{c}%
\\
a_{5}x^{3}+[(a_{5}(4t_{2}-q_{1})+a_{4}]x^{2}+[a_{5}(q_{1}^{2}-2q_{2}%
-4t_{2}q_{1}+6t_{2}^{2}+2t_{1})+a_{4}(3t_{2}-q_{1})+a_{3}]x\\
-a_{5}[q_{1}^{3}-4q_{1}q_{2}-4t_{2}q_{1}^{2}+2(3t_{2}^{2}+t_{1})q_{1}%
+8t_{2}q_{2}-4(t_{2}^{3}+t_{1}t_{2})]\\
+a_{4}(q_{1}^{2}-2q_{2}-3t_{2}q_{1}+3t_{2}^{2}+t_{1})-a_{3}(q_{1}%
-2t_{2})+a_{2}%
\end{array}
&
\begin{array}
{c}%
\\
\frac{1}{2}p_{1}%
\end{array}
\end{array}
\right)
\]
and with the dynamical part of the map (\ref{mapka}) given by
\[
a_{5}=\tfrac{1}{2}b_{3}^{2},\quad a_{4}=b_{2}b_{3},\quad a_{3}=\tfrac{1}%
{2}b_{2}^{2}+b_{1}b_{3},\quad a_{2}=b_{0}b_{3}+b_{1}b_{2}+\overline{b}%
-b_{3},\quad a_{-1}=\tfrac{1}{2}b_{0}.
\]

For $m=2$ we get
\begin{align*}
h_{1}^{A} &  =-\tfrac{1}{2}q_{{1}}\,{p}_{1}^{2}-q_{{2}}p_{{1}}p_{{2}}%
+a_{4}(q_{1}^{3}-2q_{1}q_{2})+(2a_{4}t_{1}+a_{3})(q_{2}-q_{1}^{2})+(a_{4}%
t_{1}^{2}+a_{3}t_{1}+a_{2})q_{1}\\
&  \ \ \ +a_{-1}{\mathrm{e}^{t_{{2}}}}q_{2}^{-1}-a_{-2}{\mathrm{e}^{2t_{{2}}}%
}q_{1}q_{2}^{-2}\ \\
& \\
h_{2}^{A} &  =-\tfrac{1}{2}\,q_{{2}}{p}_{1}^{2}+\tfrac{1}{2}{q}_{2}^{2}%
\,{p}_{2}^{2}+q_{2}p_{2}+a_{4}(q_{1}^{2}q_{2}-q_{2}^{2})-(2a_{4}t_{1}%
+a_{3})q_{1}q_{2}+(a_{4}t_{1}^{2}+a_{3}t_{1}+a_{2})q_{2}\\
&  \ \ \ +a_{-1}{\mathrm{e}^{t_{{2}}}}q_{1}q_{2}^{-1}-a_{-2}{\mathrm{e}%
^{2t_{{2}}}}(q_{1}^{2}-q_{2})q_{2}^{-2}%
\end{align*}%
\[
\]%
\[
L=\left(
\begin{array}
[c]{cc}%
-p_{0}x+q_{2}p_{1} & x^{2}+q_{1}x+q_{2}\\%
\begin{array}
{c}%
\\
L_{21}%
\end{array}
&
\begin{array}
{c}%
\\
p_{0}x-q_{2}p_{1}%
\end{array}
\end{array}
\right)
\]%
\[
\]
with%
\begin{align*}
L_{21} &  =2a_{4}x^{4}+2[a_{4}(2t_{1}-q_{1})+a_{3}]x^{3}+2[a_{4}(q_{1}%
^{2}-q_{2}-2t_{1}q_{1}+t_{1}^{2})+a_{3}(t_{1}-q_{1})+a_{2}]x^{2}\\
&  \ \ \ +2(a_{-1}{\mathrm{e}^{t_{{2}}}}q_{2}^{-1}-a_{-2}{\mathrm{e}^{2t_{{2}%
}}}q_{1}q_{2}^{-2}-\tfrac{1}{2}q_{1}p_{1}^{2}-q_{2}p_{1}p_{2})x+2a_{-2}%
{\mathrm{e}^{2t_{{2}}}}q_{2}^{-1}-q_{2}p_{1}^{2},
\end{align*}%
\[
\]%
\[
U_{1}=\left(
\begin{array}
[c]{cc}%
0 & \frac{1}{2}\\%
\begin{array}
{c}%
\\
a_{4}x^{2}+[2a_{4}(t_{1}-q_{1})+a_{3}]x+a_{4}(3q_{1}^{2}-2q_{2}-4t_{1}%
q_{1}+t_{1}^{2})+a_{3}(t_{1}-2q_{1})+a_{2}%
\end{array}
&
\begin{array}
{c}%
\\
0
\end{array}
\end{array}
\right)  ,
\]%
\[
\]%
\[
U_{2}=\left(
\begin{array}
[c]{cc}%
-\frac{1}{2}q_{2}p_{1}x^{-1} & -\frac{1}{2}q_{2}x^{-1}\\%
\begin{array}
{c}%
\\
a_{4}q_{2}x+2a_{4}q_{2}(q_{1}-t_{1})-a_{3}q_{2}+(\frac{1}{2}q_{2}p_{1}%
^{2}-a_{-2}{\mathrm{e}^{2t_{{2}}}}q_{2}^{-1})x^{-1}%
\end{array}
&
\begin{array}
{c}%
\\
\frac{1}{2}q_{2}p_{1}x^{-1}%
\end{array}
\end{array}
\right)
\]
and with the dynamical part of the map (\ref{mapka}) in the form
\[
a_{4}=\tfrac{1}{2}b_{3}^{2},\ \ \ a_{3}=b_{2}b_{3},\ \ \ a_{2}=b_{1}%
b_{3}+\tfrac{1}{2}b_{2}^{2}+\overline{b}-b_{3},\ \ \ a_{-1}=b_{0}%
b_{1},\ \ \ a_{-2}=\tfrac{1}{2}b_{0}^{2}.
\]

For $m=3$ we get
\begin{align*}
h_{1}^{A}  &  =\tfrac{1}{2}p_{0}^{2}-\tfrac{1}{2}{q}_{2}\,{p}_{1}^{2}%
+a_{3}{\mathrm{e}^{2t_{{1}}}}(q_{2}-q_{1}^{2})+a_{2}{\mathrm{e}^{t_{{1}}}%
}q_{1}+(a_{-1}+a_{-2}t_{2}+a_{-3}t_{2}^{2})q_{2}^{-1}\\
&  \ \ \ -(a_{-2}+2a_{-3}t_{2})q_{1}q_{2}^{-2}+a_{-3}(q_{1}^{2}-q_{2}%
)q_{2}^{-3}\\
& \\
h_{2}^{A}  &  =\tfrac{1}{2}\,q_{{1}}q_{{2}}{p}_{1}^{2}+q_{2}^{2}p_{1}%
p_{2}+q_{2}p_{1}-a_{3}{\mathrm{e}^{2t_{{1}}}}q_{1}q_{2}+a_{2}{\mathrm{e}%
^{t_{{1}}}}q_{2}+(a_{-1}+a_{-2}t_{2}+a_{-3}t_{2}^{2})q_{1}q_{2}^{-1}\\
&  \ \ \ -(a_{-2}+2a_{-3}t_{2})(q_{1}^{2}-q_{2})q_{2}^{-2}+a_{-3}q_{1}%
(q_{1}^{2}-2q_{2})q_{2}^{-3},
\end{align*}%
\[
\]%
\[
L=\left(
\begin{array}
[c]{cc}%
(q_{1}p_{0}+q_{2}p_{1})x+q_{2}p_{0} & x^{2}+q_{1}x+q_{2}\\%
\begin{array}
{c}%
\\
L_{21}%
\end{array}
&
\begin{array}
{c}%
\\
-(q_{1}p_{0}+q_{2}p_{1})x-q_{2}p_{0}%
\end{array}
\end{array}
\right)  ,
\]%
\[
\]%
\begin{align*}
L_{21}  &  =2a_{3}{\mathrm{e}^{2t_{{1}}}}x^{4}-2(a_{3}{\mathrm{e}^{2t_{{1}}}%
}q_{1}-a_{2}{\mathrm{e}^{t_{{1}}}})x^{3}\\
&  \ \ \ +2[a_{-1}q_{2}^{-1}+a_{-2}(t_{2}q_{2}-q_{1})q_{2}^{-2}+a_{-3}%
((t_{2}q_{2}-q_{1})^{2}-q_{2})q_{2}^{-3}+\tfrac{1}{2}p_{0}^{2}-\tfrac{1}%
{2}q_{2}p_{1}^{2}]x^{2}\\
&  \ \ \ +2[a_{-2}q_{2}^{-1}+a_{-3}(2t_{2}q_{2}-q_{1})q_{2}^{-2}-\tfrac{1}%
{2}q_{1}p_{0}^{2}-q_{2}p_{1}p_{0}]x+2a_{-3}q_{2}^{-1}-q_{2}p_{0}^{2},
\end{align*}%
\[
\]%
\[
U_{1}=\left(
\begin{array}
[c]{cc}%
0 & \frac{1}{2}\\%
\begin{array}
{c}%
\\
a_{3}{\mathrm{e}^{2t_{{1}}}}x^{2}-(2a_{3}{\mathrm{e}^{2t_{{1}}}}q_{1}%
-a_{2}{\mathrm{e}^{t_{{1}}}})x+a_{3}{\mathrm{e}^{2t_{{1}}}}(2q_{1}^{2}%
-q_{2})-a_{2}{\mathrm{e}^{t_{{1}}}q}_{1}\\
+a_{-1}q_{2}^{-1}-a_{-2}(q_{1}-t_{2}q_{2})q_{2}^{-2}+a_{-3}((t_{2}q_{2}%
-q_{1})^{2}-q_{2})q_{2}^{-3}+\frac{1}{2}p_{0}^{2}-\frac{1}{2}q_{2}p_{1}^{2}%
\end{array}
&
\begin{array}
{c}%
\\
0
\end{array}
\end{array}
\right)  ,
\]%
\[
\]%
\[
U_{2}=\left(
\begin{array}
[c]{cc}%
-\frac{1}{2}q_{2}p_{0}x^{-1} & -\frac{1}{2}q_{2}x^{-1}\\%
\begin{array}
{c}%
\\
-a_{3}{\mathrm{e}^{2t_{{1}}}}q_{2}x+2a_{3}{\mathrm{e}^{2t_{{1}}}}q_{1}%
q_{2}-a_{2}{\mathrm{e}^{t_{{1}}}}q_{2}+(\frac{1}{2}q_{2}p_{0}^{2}-a_{-3}%
q_{2}^{-1})x^{-1}%
\end{array}
&
\begin{array}
{c}%
\\
\frac{1}{2}q_{2}p_{1}x^{-1}%
\end{array}
\end{array}
\right)
\]
and with the dynamical part of the map (\ref{mapka})
\[
a_{3}=\frac{1}{2}b_{3}^{2},\ \ \ a_{2}=b_{2}b_{3}+\overline{b}-b_{3}%
,\ \ \ a_{-1}=b_{0}b_{2}+\frac{1}{2}b_{1}^{2},\ \ \ a_{-2}=b_{0}%
b_{1},\ \ \ a_{-3}=\frac{1}{2}b_{0}^{2}.
\]

\subsection{Three-dimensional systems}

For each $m=0,...,4$ we obtain a $5$-parameter family of Painlev\'{e}-type
systems. We use here the already introduced abbreviation $p_{0}=-q_{1}%
p_{1}-q_{2}p_{2}-q_{3}p_{3}$.

For $m=0$ we have $\mathfrak{a}$ $=\mathfrak{g}$ so $H_{r}^{A}=h_{r}^{A}$ and
$\overline{U}_{r}=U_{r}$ for all $r$ and our procedure yields
\begin{align*}
h_{r}^{A}  &  =\mathcal{E}_{r}+a_{{8}}V_{r}^{(8)}+a_{{7}}V_{r}^{(7)}+\left(
6\,a_{{8}}t_{{3}}+a_{{6}}\right)  V_{r}^{(6)}+\left(  5\,a_{{7}}t_{{3}%
}+4\,a_{{8}}t_{{2}}+a_{{5}}\right)  V_{r}^{(5)}\\
&  \quad{}+[a_{{4}}+2\,a_{{8}}\left(  6\,{t}_{3}^{2}+t_{{1}}\right)
+3\,a_{{7}}t_{{2}}+4\,a_{{6}}t_{{3}}]V_{r}^{(4)}+[a_{{3}}+12\,a_{{8}}t_{{2}%
}t_{{3}}+a_{{7}}(t_{{1}}+\tfrac{15}{2}\,{t}_{3}^{2})+2\,a_{{6}}t_{{2}%
}+3\,a_{{5}}t_{{3}}]V_{r}^{(3)}%
\end{align*}
where the geodesic quasi-St\"{a}ckel Hamiltonians $\mathcal{E}_{r}$ are given
by (\ref{geo5a}) and basic ordinary potentials $V_{r}^{(\alpha)}$ are given by
(\ref{5d}) and (\ref{6a}). Then, the Lax matrix $L$ is
\begin{equation}
L=\left(
\begin{array}
[c]{cc}%
-p_{3}x^{2}-(p_{2}+q_{1}p_{3})x-(q_{1}p_{2}+q_{2}p_{3}+p_{1}) & x^{3}%
+q_{1}x^{2}+q_{2}x+q_{3}\\%
\begin{array}
{c}%
\\
L_{21}%
\end{array}
&
\begin{array}
{c}%
\\
-L_{11}%
\end{array}
\end{array}
\right)  , \label{LL}%
\end{equation}
where%
\begin{align*}
L_{21}  &  =2a_{8}x^{5}+2(-a_{8}q_{1}+a_{7})x^{4}+2[a_{8}(6t_{3}-V_{1}%
^{(4)})-a_{7}q_{1}+a_{6}]x^{3}+2[a_{8}(4t_{2}-6t_{3}q_{1}-V_{1}^{(5)})\\
&  \ \ \ +a_{7}(5t_{3}-V_{1}^{(4)})-a_{6}q_{1}+a_{5}]x^{2}+2[a_{8}%
(+12t_{3}^{2}+2t_{1}-4t_{2}q_{1}-6t_{3}V_{1}^{(4)}-V_{1}^{(6)})\\
&  \ \ \ +a_{7}(3t_{2}-5t_{3}q_{1}-V_{1}^{(5)})+a_{6}(4t_{3}-V_{1}%
^{(4)})-a_{5}q_{1}+a_{4}-\tfrac{1}{2}p_{3}^{2}]x+2a_{8}(12t_{2}t_{3}%
-(12t_{3}^{2}+2t_{1})q_{1}\\
&  \ \ -4t_{2}V_{1}^{(4)}-6t_{3}V_{1}^{(5)}-V_{1}^{(7)})+2a_{7}(\tfrac{15}%
{2}t_{3}^{2}+t_{1}-3t_{2}q_{1}-5t_{3}V_{1}^{(4)}-V_{1}^{(6)})+2a_{6}%
(2t_{2}-4t_{3}q_{1}-V_{1}^{(5)})\\
&  \ \ +2a_{5}(3t_{3}-V_{1}^{(4)})-2a_{4}q_{1}+2a_{3}-q_{1}p_{3}^{2}%
-2p_{2}p_{3},
\end{align*}%
\[
\]
while%
\[
U_{1}=\left(
\begin{array}
[c]{cc}%
\begin{array}
{c}%
0\\
\end{array}
&
\begin{array}
{c}%
\frac{1}{2}\\
\end{array}
\\
a_{8}x^{2}+(-2a_{8}q_{1}+a_{7})x+a_{8}(6t_{3}-3V_{1}^{(4)}+q_{2})-2a_{7}%
q_{1}+a_{6} & 0
\end{array}
\right)  ,
\]%
\[
\]%
\[
U_{2}=\left(
\begin{array}
[c]{cc}%
-\frac{1}{2}p_{3} & \frac{1}{2}x+\frac{1}{2}q_{1}\\%
\begin{array}
{c}%
\\
a_{8}x^{3}+(-a_{8}q_{1}+a_{7})x^{2}+[a_{8}(6t_{3}-q_{2}-V_{1}^{(4)}%
)-a_{7}q_{1}+a_{6}]x\\
+a_{8}(4t_{2}-6t_{3}q_{1}+q_{3}-2V_{2}^{(4)}-V_{1}^{(5)})+a_{7}(5t_{3}%
-q_{2}-V_{1}^{(4)})-a_{6}q_{1}+a_{5}%
\end{array}
&
\begin{array}
{c}%
\\
\frac{1}{2}p_{3}%
\end{array}
\end{array}
\right)  ,
\]%
\[
\]%
\[
U_{3}=\left(
\begin{array}
[c]{cc}%
-\frac{1}{2}p_{3}x-\frac{1}{2}(p_{2}+q_{1}p_{3}) & \frac{1}{2}x^{2}+\frac
{1}{2}q_{1}x+\frac{1}{2}q_{2}\\%
\begin{array}
{c}%
\\
(U_{3})_{21}%
\end{array}
&
\begin{array}
{c}%
\\
\frac{1}{2}p_{3}x+\frac{1}{2}(p_{2}+q_{1}p_{3})
\end{array}
\end{array}
\right)  ,
\]%
\begin{align*}
(U_{3})_{21}  &  =a_{8}x^{4}+(-a_{8}q_{1}+a_{7})x^{3}+[a_{8}(6t_{3}%
-V_{1}^{(4)})-a_{7}q_{1}+a_{6}]x^{2}+[a_{8}(4t_{2}-6t_{3}q_{1}-V_{1}%
^{(5)}-q_{3})\\
&  \ \ \ +a_{7}(5t_{3}-V_{1}^{(4)})-a_{6}q_{1}+a_{5}]x+a_{8}(2t_{1}%
+12t_{3}^{2}-4t_{2}q_{1}-6t_{3}V_{1}^{(4)}-2V_{3}^{(4)}-V_{1}^{(6)})\\
&  \ \ \ +a_{7}(3t_{2}-5t_{3}q_{1}-q_{3}-V_{1}^{(5)})+a_{6}(4t_{3}-V_{1}%
^{(4)})-a_{5}q_{1}+a_{4}-\tfrac{1}{2}p_{3}^{2}.
\end{align*}
The dynamical part of the map (\ref{mapka}) becomes
\[
a_{8}=\tfrac{1}{2}b_{2}^{2},\ \ \ a_{7}=b_{3}b_{4},\ \ \ a_{6}=\tfrac{1}%
{2}b_{3}^{2}+b_{2}b_{4},\ \ \ a_{5}=b_{1}b_{4}+b_{2}b_{3},
\]%
\[
a_{4}=\tfrac{1}{2}b_{2}^{2}+b_{0}b_{4}+b_{1}b_{3},\ \ \ a_{3}=b_{0}b_{3}%
+b_{1}b_{2}-b_{4}+\overline{b}.
\]

For $m=1$, we have $H_{r}^{A}=h_{r}^{A},$ $\overline{U}_{r}=U_{r},\ r=1,2,$
$H_{3}^{A}=h_{3}^{A}+t_{2}h_{1}^{A}$, $\overline{U}_{3}=U_{3}+t_{2}U_{1}$\ and
our procedure yields
\begin{align*}
h_{r}^{A}  &  =\mathcal{E}_{r}+a_{{7}}V_{r}^{(7)}+(a_{{6}}+6a_{7}t_{3}%
)V_{r}^{(6)}+[a_{{5}}+5a_{6}t_{3}+\,a_{{7}}(4t_{{2}}+15t_{3}^{2})]V_{r}%
^{(5)}+[a_{4}+4\,a_{{5}}t_{{3}}+\,a_{{6}}(3t_{{2}}+10t_{3}^{2})\\
&  \quad{}+2a_{7}(t_{1}+9t_{2}t_{3}+10t_{3}^{3})]V_{r}^{(4)}+[a_{{3}%
}+3\,a_{{4}}t_{{3}}+2\,a_{{5}}(t_{{2}}+3t_{3}^{2})+a_{{6}}\left(  t_{{1}%
}+10t_{2}t_{3}+10t_{3}^{3}\right) \\
&  \quad{}+a_{7}(4t_{2}^{2}+6t_{1}t_{3}+30t_{2}t_{3}^{2}+15t_{3}^{4}%
)]V_{r}^{(3)}+a_{-1}V_{r}^{(-1)}%
\end{align*}
where the geodesic quasi-St\"{a}ckel Hamiltonians $\mathcal{E}_{r}$ are given
by (\ref{geo6a}) and basic ordinary potentials $V_{r}^{(\alpha)}$ are given by
(\ref{5d}) and (\ref{6a}). In this case $L$ is given by
\[
L=\left(
\begin{array}
[c]{cc}%
-p_{2}x^{2}-(p_{1}+q_{1}p_{2})x+q_{3}p_{3} & x^{3}+q_{1}x^{2}+q_{2}x+q_{3}\\%
\begin{array}
{c}%
\\
L_{21}%
\end{array}
&
\begin{array}
{c}%
\\
p_{2}x^{2}+(p_{1}+q_{1}p_{2})x-q_{3}p_{3}%
\end{array}
\end{array}
\right)
\]
with%
\begin{align*}
L_{21}  &  =2a_{7}x^{5}+2[a_{7}(6t_{3}-q_{1})+a_{6}]x^{4}+2[a_{7}(15t_{3}%
^{2}+4t_{2}-6t_{3}q_{1}-V_{1}^{(4)})+a_{6}(5t_{3}-q_{1})+a_{5}]x^{3}\\
&  \ \ \ +2[a_{7}(20t_{3}^{3}+18t_{2}t_{3}+2t_{1}-(15t_{3}^{2}+4t_{2}%
)q_{1}-6t_{3}V_{1}^{(4)}-V_{1}^{(5)})+a_{6}(10t_{3}^{2}+3t_{2}-5t_{3}%
q_{1}-V_{1}^{(4)})\\
&  \ \ \ +a_{5}(4t_{3}-q_{1})+a_{4}]x^{2}+2[a_{7}(15t_{3}^{4}+30t_{2}t_{3}%
^{2}+6t_{1}t_{3}+4t_{2}^{2}-(20t_{3}^{3}+18t_{2}t_{3}+2t_{1})q_{1}\\
&  \ \ \ -(15t_{3}^{2}+4t_{2})V_{1}^{(4)}-6t_{3}V_{1}^{(5)}-V_{1}^{(6)}%
)+a_{6}(10t_{3}^{3}+10t_{2}t_{3}+t_{1}-(10t_{3}^{2}+3t_{2})q_{1}-5t_{3}%
V_{1}^{(4)}-V_{1}^{(5)})\\
&  \ \ \ +a_{5}(6t_{3}^{2}+2t_{2}-4t_{3}q_{1}-V_{1}^{(4)})+a_{4}(3t_{3}%
-q_{1})+a_{3}-\tfrac{1}{2}p_{2}^{2}]x+2a_{-1}q_{3}^{-1}-q_{3}p_{3}^{2},
\end{align*}
while%
\[
U_{1}=\left(
\begin{array}
[c]{cc}%
0 & \frac{1}{2}\\%
\begin{array}
{c}%
\\
a_{7}x^{2}+[a_{7}(6t_{3}-2q_{1})+a_{6}]x+a_{7}(15t_{3}^{2}+4t_{2}-12t_{3}%
q_{1}+q_{2}-3V_{1}^{(4)})\\
+a_{6}(5t_{3}-2q_{1})+a_{5}%
\end{array}
&
\begin{array}
{c}%
\\
0
\end{array}
\end{array}
\right)  ,
\]%
\[
\]%
\[
U_{2}=\left(
\begin{array}
[c]{cc}%
-\frac{1}{2}p_{2} & \frac{1}{2}\lambda+\frac{1}{2}q_{1}\\%
\begin{array}
{c}%
\\
(U_{2})_{21}%
\end{array}
&
\begin{array}
{c}%
\\
\frac{1}{2}p_{2}%
\end{array}
\end{array}
\right)  ,
\]%
\begin{align*}
(U_{2})_{21}  &  =a_{7}x^{3}+[a_{7}(6t_{3}-q_{1})+a_{6}]x^{2}+[a_{7}%
(15t_{3}^{2}+4t_{2}-6t_{3}q_{1}-q_{2}-V_{1}^{(4)})+a_{6}(5t_{3}-q_{1}%
)+a_{5}]x\\
&  \ \ \ +a_{7}[20t_{3}^{3}+18t_{2}t_{3}+2t_{1}-(15t_{3}^{2}+4t_{2}%
)q_{1}-6t_{3}(V_{1}^{(4)}+q_{2})+q_{3}-2V_{2}^{(4)}-V_{1}^{(5)}]\\
&  \ \ \ +a_{6}(10t_{3}^{2}+3t_{2}-5t_{3}q_{1}-q_{2}-V_{1}^{(4)})+a_{5}%
(4t_{3}-q_{1})+a_{4},
\end{align*}%
\[
\]%
\[
U_{3}=\left(
\begin{array}
[c]{cc}%
-\frac{1}{2}p_{2}x-\frac{1}{2}(p_{1}+q_{1}p_{2}) & \frac{1}{2}\lambda
^{2}+\frac{1}{2}q_{1}\lambda+\frac{1}{2}q_{2}\\%
\begin{array}
{c}%
\\
(U_{3})_{21}%
\end{array}
&
\begin{array}
{c}%
\\
\frac{1}{2}p_{2}x+\frac{1}{2}(p_{1}+q_{1}p_{2})
\end{array}
\end{array}
\right)  ,
\]%
\begin{align*}
(U_{3})_{21}  &  =a_{7}x^{4}+[a_{7}(6t_{3}-q_{1})+a_{6}]x^{3}+[a_{7}%
(15t_{3}^{2}+4t_{2}-6t_{3}q_{1}-V_{1}^{(4)})+a_{6}(5t_{3}-q_{1})+a_{5}]x^{2}\\
&  \ \ \ +[a_{7}(20t_{3}^{3}+18t_{2}t_{3}+2t_{1}-(15t_{3}^{2}+4t_{2}%
)q_{1}-6t_{3}V_{1}^{(4)}-q_{3}-V_{1}^{(5)})\\
&  \ \ \ +a_{6}(10t_{3}^{2}+3t_{2}-5t_{3}q_{1}-V_{1}^{(4)})+a_{5}(4t_{3}%
-q_{1})+a_{4}]x+a_{7}[15t_{3}^{4}+30t_{2}t_{3}^{2}+6t_{1}t_{3}\\
&  \ \ \ +4t_{2}^{2}-(20t_{3}^{3}+18t_{2}t_{3}+2t_{1})q_{1}-(15t_{3}%
^{2}+4t_{2})V_{1}^{(4)}-6t_{3}(V_{1}^{(5)}+q_{3})-2V_{3}^{(4)}-V_{1}^{(6)}]\\
&  \ \ \ +a_{6}[10t_{3}^{3}+10t_{2}t_{3}+t_{1}-(10t_{3}^{2}+3t_{2}%
)q_{1}-5t_{3}V_{1}^{(4)}-q_{3}-V_{1}^{(5)}]\\
&  \ \ \ +a_{5}(6t_{3}^{2}+2t_{2}-4t_{3}q_{1}-V_{1}^{(4)})+a_{4}(3t_{3}%
-q_{1})+a_{3}-\tfrac{1}{2}p_{2}^{2}.
\end{align*}
The dynamical part of the map (\ref{mapka}) is
\[
a_{7}=\tfrac{1}{2}b_{4}^{2},\quad a_{6}=b_{3}b_{4},\quad a_{5}=\tfrac{1}%
{2}b_{3}^{2}+b_{2}b_{4},\quad a_{4}=b_{1}b_{4}+b_{2}b_{3},
\]%
\[
a_{3}=\tfrac{1}{2}b_{2}^{2}+b_{0}b_{4}+b_{1}b_{3}+\overline{b}-b_{4},\quad
a_{-1}=\tfrac{1}{2}b_{0}^{2}.
\]

For $m=2$ we have $\mathfrak{a}$ $=\mathfrak{g}$ so $H_{r}^{A}=h_{r}^{A}$ and
$\overline{U}_{r}=U_{r}$ for all $r$ and so
\begin{align*}
h_{r}^{A}  &  =\mathcal{E}_{r}+a_{{6}}V_{r}^{(6)}+(a_{{5}}+4a_{6}t_{2}%
)V_{r}^{(5)}+[a_{{4}}+3a_{5}t_{2}+\,2a_{{6}}(t_{{1}}+3t_{2}^{2})]V_{r}^{(4)}\\
&  \quad{}+[a_{3}+2a_{{4}}t_{{2}}+\,a_{{5}}(t_{{1}}+3t_{2}^{2})+4a_{6}%
(t_{1}t_{2}+t_{2}^{3})]V_{r}^{(3)}+a_{-1}{\mathrm{e}^{t_{{3}}}V}_{r}%
^{(-1)}+a_{-2}{\mathrm{e}^{2t_{{3}}}V}_{r}^{(-2)},
\end{align*}
where the geodesic quasi-St\"{a}ckel Hamiltonians $\mathcal{E}_{r}$ are given
by (\ref{geo7}) and basic ordinary potentials $V_{r}^{(\alpha)}$ are given by
(\ref{5d}) and (\ref{6a}). In this case $L$ is given by
\[
L=\left(
\begin{array}
[c]{cc}%
-p_{1}x^{2}+(q_{2}p_{2}+q_{3}p_{3})x+q_{3}p_{2} & x^{3}+q_{1}x^{2}%
+q_{2}x+q_{3}\\%
\begin{array}
{c}%
\\
L_{21}%
\end{array}
&
\begin{array}
{c}%
\\
-L_{11}%
\end{array}
\end{array}
\right)  ,
\]
where%
\begin{align*}
L_{21}  &  =2a_{6}x^{5}+2[a_{6}(4t_{2}-q_{1})+a_{5}]x^{4}+2[a_{6}(6t_{2}%
^{2}+2t_{1}-4t_{2}q_{1}-V_{1}^{(4)})+a_{5}(3t_{2}-q_{1})+a_{4}]x^{3}\\
&  \ \ \ \ 2[a_{6}(4t_{2}^{3}+4t_{1}t_{2}-(6t_{2}^{2}+2t_{1})q_{1}-4t_{2}%
V_{1}^{(4)}-V_{1}^{(5)})+a_{5}(3t_{2}^{2}+t_{1}-3t_{2}q_{1}-V_{1}^{(4)})\\
&  \ \ \ +a_{4}(2t_{2}-q_{1})+a_{3})]x^{2}+2[a_{-1}{\mathrm{e}^{t_{{3}}}}%
V_{1}^{(-1)}+a_{-2}{\mathrm{e}^{2t_{{3}}}}V_{1}^{(-2)}-\frac{1}{2}q_{2}%
p_{2}^{2}-q_{3}p_{2}p_{3}]x\\
&  \ \ \ +2a_{-2}{\mathrm{e}^{2t_{{3}}}}V_{1}^{(-1)}-q_{3}p_{2}^{2},
\end{align*}
while%
\[
U_{1}=\left(
\begin{array}
[c]{cc}%
0 & \frac{1}{2}\\%
\begin{array}
{c}%
\\%
\begin{array}
{c}%
a_{6}x^{2}+[a_{6}(4t_{2}-2q_{1})+a_{5}]x+a_{6}(6t_{2}^{2}+2t_{1}-8t_{2}%
q_{1}+q_{2}-3V_{1}^{(4)})\\
+a_{5}(3t_{2}-2q_{1})+a_{4}%
\end{array}
\end{array}
& 0
\end{array}
\right)  ,
\]%
\[
\]%
\[
U_{2}=\left(
\begin{array}
[c]{cc}%
-\frac{1}{2}p_{1} & \frac{1}{2}\lambda+\frac{1}{2}q_{1}\\%
\begin{array}
{c}%
\\
(U_{2})_{21}%
\end{array}
&
\begin{array}
{c}%
\\
\frac{1}{2}p_{1}%
\end{array}
\end{array}
\right)  ,
\]%
\begin{align*}
(U_{2})_{21}  &  =a_{6}x^{3}+[a_{6}(4t_{2}-q_{1})+a_{5}]x^{2}+[a_{6}%
(6t_{2}^{2}+2t_{1}-4t_{2}q_{1}-q_{2}-V_{1}^{(4)})+a_{5}(3t_{2}-q_{1}%
)+a_{4}]x\\
&  \ \ \ +a_{6}[4t_{2}^{3}+4t_{1}t_{2}-(6t_{2}^{2}+2t_{1})q_{1}-4t_{2}%
(q_{2}+V_{1}^{(4)})+q_{3}-2V_{2}^{(4)}-V_{1}^{(5)}]\\
&  \ \ \ +a_{5}(3t_{2}^{2}+t_{1}-3t_{2}q_{1}-q_{2}-V_{1}^{(4)})+a_{4}%
(2t_{1}-q_{1})+a_{3},
\end{align*}%
\[
\]%
\[
U_{3}=\left(
\begin{array}
[c]{cc}%
-\frac{1}{2}q_{3}p_{2}x^{-1} & -\frac{1}{2}q_{3}x^{-1}\\%
\begin{array}
{c}%
\\
-a_{6}q_{3}x+2a_{6}q_{3}(q_{1}-2t_{2})-a_{5}q_{3}-[a_{-2}{\mathrm{e}^{2t_{{3}%
}}}V_{1}^{(-1)}-\frac{1}{2}q_{3}p_{2}^{2}]x^{-1}%
\end{array}
&
\begin{array}
{c}%
\\
\frac{1}{2}q_{3}p_{2}x^{-1}%
\end{array}
\end{array}
\right)  .
\]
The dynamical part of the map (\ref{mapka}) is
\[
a_{6}=\tfrac{1}{2}b_{4}^{2},\ \ \ a_{5}=b_{3}b_{4},\ \ \ a_{4}=\tfrac{1}%
{2}b_{3}^{2}+b_{2}b_{4},\ \ \ a_{3}=b_{1}b_{4}+b_{2}b_{3}-b_{4}+\overline
{b},\ \ \ a_{-1}=b_{0}b_{1},\ \ \ a_{-2}=\tfrac{1}{2}b_{0}^{2}.
\]

For $m=3$, we have $H_{r}^{A}=h_{r}^{A},$ $\overline{U}_{r}=U_{r},\ r=1,3,$
$H_{2}^{A}=h_{2}^{A}+t_{3}h_{3}^{A}$, $\overline{U}_{2}=U_{2}+t_{3}U_{3}$\ and
our procedure yields
\begin{align*}
h_{r}^{A} &  =\mathcal{E}_{r}+a_{{5}}V_{r}^{(5)}+(a_{{4}}+2a_{5}t_{1}%
)V_{r}^{(4)}+(a_{{3}}+a_{4}t_{1}+\,a_{{5}}t_{{1}}^{2})V_{r}^{(3)}\\
&  \quad{}+[a_{-1}{\mathrm{e}^{2t_{{2}}}+a}_{-2}({\mathrm{e}^{2t_{{2}}}+t}%
_{3}{\mathrm{e}^{3t_{{2}}})+a}_{-3}t_{3}^{2}{\mathrm{e}^{4t_{{2}}}]V}%
_{r}^{(-1)}+(a_{-2}{\mathrm{e}^{3t_{{2}}}+2a}_{-3}t_{3}{\mathrm{e}^{4t_{{2}}%
})V}_{r}^{(-2)}+a_{-3}{\mathrm{e}^{4t_{{2}}}V}_{r}^{(-3)},
\end{align*}
where the geodesic quasi-St\"{a}ckel Hamiltonians $\mathcal{E}_{r}$ are given
by (\ref{geo8}) and basic ordinary potentials $V_{r}^{(\alpha)}$ are given by
(\ref{5d}) and (\ref{6a}). In this case $L$ is given by
\[
L=\left(
\begin{array}
[c]{cc}%
-p_{0}x^{2}+(q_{2}p_{1}+q_{3}p_{2})x+q_{3}p_{1} & x^{3}+q_{1}x^{2}%
+q_{2}x+q_{3}\\%
\begin{array}
{c}%
\\
L_{21}%
\end{array}
&
\begin{array}
{c}%
\\
-L_{11}%
\end{array}
\end{array}
\right)  ,
\]
with%
\begin{align*}
L_{21} &  =2a_{5}x^{5}+2[a_{5}(2t_{1}-q_{1})+a_{4})]x^{4}+2[a_{5}(t_{1}%
^{2}-2t_{1}q_{1}-V_{1}^{(4)})+a_{4}(t_{1}-q_{1})+a_{3}]x^{3}\\
&  \ \ \ +2[a_{-1}{\mathrm{e}^{2t_{{2}}}}V_{1}^{(-1)}+a_{-2}({\mathrm{e}%
^{2t_{{2}}}}(t_{3}{\mathrm{e}^{t_{{2}}}}+1)V_{1}^{(-1)}+{\mathrm{e}^{3t_{{2}}%
}}V_{1}^{(-2)})+a_{-3}{\mathrm{e}^{4t_{{2}}}}(t_{3}^{2}V_{1}^{(-1)}%
+2t_{3}V_{1}^{(-2)}+V_{1}^{(-3)})\\
&  \ \ \ +p_{0}p_{1}+\tfrac{1}{2}q_{1}p_{1}^{2}-\tfrac{1}{2}q_{3}p_{2}%
^{2}]x^{2}+2[a_{-2}{\mathrm{e}^{3t_{{2}}}}V_{1}^{(-1)}+a_{-3}{\mathrm{e}%
^{4t_{{2}}}}(t_{3}V_{1}^{(-1)}+V_{1}^{(-2)})-\tfrac{1}{2}q_{2}p_{1}^{2}%
-q_{3}p_{1}p_{2}]x\\
&  \ \ \ +2a_{-3}{\mathrm{e}^{4t_{{2}}}}V_{1}^{(-1)}-q_{3}p_{1}^{2},
\end{align*}
while
\[
U_{1}=\left(
\begin{array}
[c]{cc}%
0 & \frac{1}{2}\\%
\begin{array}
{c}%
\\
a_{5}x^{2}+[2a_{5}(t_{1}-q_{1})+a_{4}]x+a_{5}(t_{1}^{2}-4t_{1}q_{1}%
+q_{2}-3V_{1}^{(4)})+a_{4}(t_{1}-2q_{1})+a_{3}%
\end{array}
&
\begin{array}
{c}%
\\
0
\end{array}
\end{array}
\right)  ,
\]%
\[
\]%
\[
U_{2}=\left(
\begin{array}
[c]{cc}%
-\frac{1}{2}(q_{2}p_{1}+q_{3}p_{2})x^{-1}-\frac{1}{2}q_{3}p_{1}x^{-2} &
-\frac{1}{2}q_{2}x^{-1}-\frac{1}{2}q_{3}x^{-2}\\%
\begin{array}
{c}%
\\
(U_{2})_{21}%
\end{array}
&
\begin{array}
{c}%
\\
\frac{1}{2}(q_{2}p_{1}+q_{3}p_{2})x^{-1}+\frac{1}{2}q_{3}p_{1}x^{-2}%
\end{array}
\end{array}
\right)  ,
\]%
\[
\]%
\begin{align*}
(U_{2})_{21} &  =-a_{5}q_{2}x-a_{5}(2t_{1}q_{2}-2q_{1}q_{2}+q_{3})-a_{4}%
q_{2}-[a_{-2}{\mathrm{e}^{3t_{{2}}}}V_{1}^{(-1)}+a_{-3}{\mathrm{e}^{4t_{{2}}}%
}(2t_{3}V_{1}^{(-1)}+V_{1}^{(-2)})\\
&  -\tfrac{1}{2}q_{2}p_{1}^{2}-q_{3}p_{1}p_{2}]x^{-1}-(a_{-3}{\mathrm{e}%
^{4t_{{2}}}}V_{1}^{(-1)}-\tfrac{1}{2}q_{3}p_{1}^{2})x^{-2},
\end{align*}%
\[
\]%
\[
U_{3}=\left(
\begin{array}
[c]{cc}%
-\frac{1}{2}q_{3}p_{1}x^{-1} & -\frac{1}{2}q_{3}x^{-1}\\%
\begin{array}
{c}%
\\
-a_{5}q_{3}x-[2a_{5}(t_{1}-q_{1})+a_{4}]q_{3}-(a_{-3}{\mathrm{e}^{4t_{{2}}}%
}V_{1}^{(-1)}-\frac{1}{2}q_{3}p_{1}^{2})x^{-1}%
\end{array}
&
\begin{array}
{c}%
\\
\frac{1}{2}q_{3}p_{1}x^{-1}%
\end{array}
\end{array}
\right)  .
\]
The dynamical part of the map (\ref{mapka}) is
\[
a_{5}=\tfrac{1}{2}b_{4}^{2},\quad a_{4}=b_{3}b_{4},\quad a_{3}=\tfrac{1}%
{2}b_{3}^{2}+b_{2}b_{4}+\overline{b}-b_{4},\quad a_{-1}=\tfrac{1}{2}b_{1}%
^{2}-b_{0}b_{1}+b_{0}b_{2},
\]%
\[
a_{-2}=b_{0}b_{1},\quad a_{-3}=\tfrac{1}{2}b_{0}^{2}.
\]

Finally, for $m=4$ we have $\mathfrak{a}$ $=\mathfrak{g}$ so $H_{r}^{A}%
=h_{r}^{A}$ and $\overline{U}_{r}=U_{r}$ for all $r$ and our procedure yields
\begin{align*}
h_{r}^{A} &  =\mathcal{E}_{r}+a_{{4}}{\mathrm{e}^{2t_{{1}}}}V_{r}%
^{(4)}+(a_{{3}}{\mathrm{e}^{t_{{1}}}}+a_{4}{\mathrm{e}^{t_{{1}}}})V_{r}%
^{(3)}+[a_{-1}{+2a}_{-2}t_{2}+a_{-3}(t_{3}+3t_{2}^{2})+4a_{-4}(t_{2}%
t_{3}+t_{2}^{3}){]V}_{r}^{(-1)}\\
&  \quad{}+[a_{-2}+3a_{-3}t_{2}+2a_{-4}(t_{3}+3t_{2}^{2})]{V}_{r}%
^{(-2)}+(a_{-3}+4a_{-4}t_{2}){V}_{r}^{(-3)}+a_{-4}V_{r}^{(-4)},
\end{align*}
where geodesic quasi-St\"{a}ckel Hamiltonians $\mathcal{E}_{r}$ are given by
(\ref{geo9}) and basic ordinary potentials $V_{r}^{(\alpha)}$ are given by
(\ref{5d}) and (\ref{6a}). The matrix $L$ is given by
\[
L=\left(
\begin{array}
[c]{cc}%
(q_{1}p_{0}+q_{2}p_{1}+q_{3}p_{2})x^{2}+(q_{2}p_{0}+q_{3}p_{1})x+q_{3}p_{0} &
x^{3}+q_{1}x^{2}+q_{2}x+q_{3}\\%
\begin{array}
{c}%
\\
L_{21}%
\end{array}
&
\begin{array}
{c}%
\\
-L_{11}%
\end{array}
\end{array}
\right)  ,
\]
with%
\begin{align*}
L_{21} &  =2a_{4}{\mathrm{e}^{2t_{{1}}}}x^{5}+2[a_{4}({\mathrm{e}^{t_{{1}}%
}-\mathrm{e}^{2t_{{1}}}q}_{1})+a_{3}{\mathrm{e}^{t_{{1}}}}]x^{4}+2[a_{-1}%
V_{1}^{(-1)}+a_{-2}(2t_{2}V_{1}^{(-1)}+V_{1}^{(-2)})\\
&  \ \ \ +a_{-3}((3t_{2}^{2}+t_{3})V_{1}^{(-1)}+3t_{2}V_{1}^{(-2)}%
+V_{1}^{(-3)})+a_{-4}(4(t_{2}^{3}+t_{2}t_{3})V_{1}^{(-1)}+(6t_{2}^{2}%
+2t_{3})V_{1}^{(-2)}+4t_{2}V_{1}^{(-3)}\\
&  \ \ \ +V_{1}^{(-4)})+\tfrac{1}{2}p_{0}^{2}-\tfrac{1}{2}q_{2}p_{1}^{2}%
-q_{3}p_{1}p_{2}]x^{3}+2[a_{-2}V_{1}^{(-1)}+a_{-3}(3t_{2}V_{1}^{(-1)}%
+V_{1}^{(-2)})+a_{-4}((6t_{2}^{2}+2t_{3})V_{1}^{(-1)}\\
&  \ \ \ +4t_{2}V_{1}^{(-2)}+V_{1}^{(-3)})-\tfrac{1}{2}q_{1}p_{0}^{2}%
-q_{2}p_{0}p_{1}-q_{3}p_{0}p_{2}-\tfrac{1}{2}q_{3}p_{1}^{2}]x^{2}%
+2[a_{-3}V_{1}^{(-1)}+a_{-4}(4t_{2}V_{1}^{(-1)}+V_{1}^{(-2)})\\
&  \ \ \ -\tfrac{1}{2}q_{2}p_{0}^{2}-q_{3}p_{0}p_{1}]x+2a_{-4}V_{1}%
^{(-1)}-q_{3}p_{0}^{2}%
\end{align*}
while
\[
U_{1}=\left(
\begin{array}
[c]{cc}%
0 & \frac{1}{2}\\%
\begin{array}
{c}%
\\
(U_{1})_{21}%
\end{array}
&
\begin{array}
{c}%
\\
0
\end{array}
\end{array}
\right)  ,
\]%
\begin{align*}
(U_{1})_{21} &  =a_{4}{\mathrm{e}^{2t_{{1}}}x}^{2}+[a_{4}({\mathrm{e}^{t_{{1}%
}}-2\mathrm{e}^{2t_{{1}}}q}_{1})+a_{3}{\mathrm{e}^{t_{{1}}}}]x-a_{4}%
{\mathrm{e}^{t_{{1}}}(2\mathrm{e}^{t_{{1}}}V}_{1}^{(4)}+{\mathrm{e}^{t_{{1}}%
}q}_{2}-q_{1})-a_{3}{\mathrm{e}^{t_{{1}}}q}_{1}\\
&  \ \ \ +a_{-1}V_{1}^{(-1)}+a_{-2}(2t_{2}V_{1}^{(-1)}+V_{1}^{(-2)}%
)+a_{-3}[(3t_{2}^{2}+t_{3})V_{1}^{(-1)}+3t_{2}V_{1}^{(-2)}+V_{1}^{(-3)}]\\
&  \ \ \ +a_{-4}[4(t_{2}^{3}+t_{2}t_{3})V_{1}^{(-1)}+(6t_{2}^{2}+2t_{3}%
)V_{1}^{(-2)}+4t_{2}V_{1}^{(-3)}+V_{1}^{(-4)}]+\tfrac{1}{2}p_{0}^{2}-\tfrac
{1}{2}q_{2}p_{1}^{2}-q_{3}p_{1}p_{2},
\end{align*}%
\[
\]%
\[
U_{2}=\left(
\begin{array}
[c]{cc}%
-\frac{1}{2}(q_{2}p_{0}+q_{3}p_{1})x^{-1}+\frac{1}{2}q_{3}p_{0}x^{-2} &
-\frac{1}{2}q_{2}x^{-1}-\frac{1}{2}q_{3}x^{-2}\\%
\begin{array}
{c}%
\\
(U_{2})_{21}%
\end{array}
&
\begin{array}
{c}%
\\
\frac{1}{2}(q_{2}p_{0}+q_{3}p_{1})x^{-1}-\frac{1}{2}q_{3}p_{0}x^{-2}%
\end{array}
\end{array}
\right)  ,
\]%
\begin{align*}
(U_{2})_{21} &  =-a_{4}{\mathrm{e}^{2t_{{1}}}q}_{2}x-a_{4}{\mathrm{e}^{t_{{1}%
}}(2\mathrm{e}^{t_{{1}}}q}_{1}q_{2}-{\mathrm{e}^{t_{{1}}}q}_{3}-q_{2}%
)-a_{3}{\mathrm{e}^{t_{{1}}}q}_{2}+[-a_{-3}V_{1}^{(-1)}-a_{-4}(4t_{2}%
V_{1}^{(-1)}+V_{1}^{(-2)})\\
&  \ \ \ +\frac{1}{2}q_{2}p_{0}^{2}+q_{3}p_{0}p_{1}]x^{-1}+(-a_{4}V_{1}%
^{(-1)}+\frac{1}{2}q_{3}p_{0}^{2})x^{-2},
\end{align*}%
\[
\]%
\[
U_{3}=\left(
\begin{array}
[c]{cc}%
-\frac{1}{2}q_{3}p_{0}x^{-1} & -\frac{1}{2}q_{3}x^{-1}\\%
\begin{array}
{c}%
\\
-a_{4}{\mathrm{e}^{2t_{{1}}}q}_{3}x+a_{4}{\mathrm{e}^{t_{{1}}}q}%
_{3}(2{\mathrm{e}^{t_{{1}}}q}_{1}-1)-a_{3}{\mathrm{e}^{t_{{1}}}q}_{3}%
+(-a_{-4}V_{1}^{(-1)}+\frac{1}{2}q_{3}p_{0}^{2})x^{-1}%
\end{array}
&
\begin{array}
{c}%
\\
\frac{1}{2}q_{3}p_{0}x^{-1}%
\end{array}
\end{array}
\right)  .
\]
The dynamical part of the map (\ref{mapka}) reads
\[
a_{4}=\tfrac{1}{2}b_{4}^{2},\quad a_{3}=-\tfrac{1}{2}b_{4}^{2}+b_{3}%
b_{4}+\overline{b}-b_{4},\quad a_{-1}=b_{0}b_{3}+b_{1}b_{2},\quad
a_{-2}=\tfrac{1}{2}b_{1}^{2}+b_{0}b_{2},
\]%
\[
a_{-3}=b_{0}b_{1},\quad a_{-4}=\tfrac{1}{2}b_{0}^{2}.
\]

\subsection{Painlev\'{e} $P_{I}-P_{IV}$ hierarchies.}

The above list gives us a possibility of constructing the complete
Painlev\'{e} $P_{I}-P_{IV}$ hierarchies (up to some rescaling, see Part I,
Section 9), in the following way. Fixing $m=0$, choosing $a_{2n+1}=-1$ and
remaining $a_{i}=0$ and letting $n$ vary we obtain the $P_{I}$ hierarchy with
the first three members given by Hamiltonians%
\[%
\begin{array}
[c]{ll}%
n=1: & H=\tfrac{1}{2}p^{2}-q^{3}-tq\\
& \\
n=2: & H_{1}=p_{1}p_{2}+\tfrac{1}{2}q_{1}p_{2}^{2}+q_{1}^{4}-3q_{1}^{2}%
q_{2}+q_{2}^{2}-3t_{2}(q_{2}-q_{1}^{2})-t_{1}q_{1}\\
& \\
& H_{2}=\tfrac{1}{2}\,{p_{{1}}^{2}+}q_{{1}}p_{{1}}p_{{2}}+\tfrac{1}{2}\left(
{q_{{1}}^{2}}-q_{{2}}\right)  {p_{{2}}^{2}+p}_{2}-2q_{1}q_{2}^{2}+q_{1}%
^{3}q_{2}+3t_{2}q_{2}q_{1}-t_{1}q_{2}\\
& \\
n=3: & H_{r}=\mathcal{E}_{r}-V_{r}^{(7)}-5\,t_{{3}}V_{r}^{(5)}-3\,t_{{2}}%
V_{r}^{(4)}-(t_{{1}}+\tfrac{15}{2}\,{t}_{3}^{2})V_{r}^{(3)}\text{, \ }r=1,2,3
\end{array}
\]
where $\mathcal{E}_{r}$ are given by (\ref{geo5a}).

Further, fixing $m=0$ but choosing $a_{2n+2}=\frac{1}{4}$, $a_{2n-1}=-\alpha$
and remaining $a_{i}=0$ and letting $n$ vary we obtain the $P_{II}$ hierarchy
with the first three members given by%
\[%
\begin{array}
[c]{ll}%
n=1: & H=\tfrac{1}{2}p^{2}-\frac{1}{4}q^{4}-\frac{1}{2}tq^{2}-\alpha q\\
& \\
n=2: & H_{1}=p_{1}p_{2}+\tfrac{1}{2}q_{1}p_{2}^{2}+\frac{1}{4}(q_{1}%
^{5}-4q_{1}^{3}q_{2}+3q_{1}q_{2}^{2})+t_{2}(q_{1}^{3}-2q_{1}q_{2})+(\frac
{1}{2}t_{1}-\alpha)(q_{2}-q_{1}^{2})+t_{2}^{2}q_{1}\\
& \\
& H_{2}=\tfrac{1}{2}\,{p_{{1}}^{2}+}q_{{1}}p_{{1}}p_{{2}}+\tfrac{1}{2}\left(
{q_{{1}}^{2}}-q_{{2}}\right)  {p_{{2}}^{2}+p}_{2}+\frac{1}{4}(q_{1}^{4}%
q_{2}-3q_{1}^{2}q_{2}^{2}+q_{2}^{3})+t_{2}(q_{1}^{2}q_{2}-q_{2}^{2})-(\frac
{1}{2}t_{1}-\alpha)q_{2}q_{1}+t_{2}^{2}q_{2}\\
& \\
n=3: & H_{r}=\mathcal{E}_{r}+\frac{1}{4}V_{r}^{(8)}+\frac{3}{2}t_{{3}}%
V_{r}^{(6)}+\left(  t_{{2}}-\alpha\right)  V_{r}^{(5)}+\left(  3\,{t}_{3}%
^{2}+\frac{1}{2}t_{{1}}\right)  V_{r}^{(4)}+\left(  3\,t_{{2}}t_{{3}%
}-3\,\alpha t_{{3}}\right)  V_{r}^{(3)}\text{, \ }r=1,2,3
\end{array}
\]

with the same $\mathcal{E}_{r}$ given by (\ref{geo5a}).

For $m=1$ we fix $a_{2n+1}=1$, we let $a_{-1}=\alpha$ and $a_{2n-1}=\beta$ be
free, the remaining $a_{i}=0$ and we let $n$ vary. This way we obtain the
$P_{IV}$ hierarchy with the first three members given by%
\[%
\begin{array}
[c]{ll}%
n=1: & H=-\tfrac{1}{2}qp^{2}+q^{3}-2tq^{2}+(t^{2}+\beta)q+\alpha q^{-1}\\
& \\
n=2: & H_{1}=\tfrac{1}{2}\,{p}_{1}^{2}-\tfrac{1}{2}\,q_{{2}}{p}_{2}^{2}%
-(q_{1}^{4}-3q_{1}^{2}q_{2}+q_{2}^{2})+4t_{2}(q_{1}^{3}-2q_{1}q_{2}%
)+[2(3t_{2}^{2}+t_{1})+\beta](q_{2}-q_{1}^{2})\\
& \\
& \ \ \ \ \ \ \ \ +[4(t_{2}^{3}+t_{1}t_{2})+2\beta t_{2}]q_{1}+\alpha
q_{2}^{-1}\\
& \\
& H_{2}=-q_{{2}}p_{{1}}p_{{2}}-\tfrac{1}{2}\,q_{{1}}q_{{2}}{p}_{2}^{2}%
+p_{1}-(q_{1}^{3}q_{2}-2q_{1}q_{2}^{2})+4t_{2}(q_{1}^{2}q_{2}-q_{2}^{2})\\
& \\
& \ \ \ \ \ \ \ \ \ -\left[  2(3t_{2}^{2}+t_{1})+\beta\right]  q_{1}%
q_{2}+\left[  4(t_{2}^{3}+t_{1}t_{2})+2\beta t_{2}\right]  q_{2}+\alpha
q_{1}q_{2}^{-1}\\
& \\
n=3: & h_{r}=\mathcal{E}_{r}+V_{r}^{(7)}+6t_{3}V_{r}^{(6)}+\left(
\beta+\,4t_{{2}}+15t_{3}^{2}\right)  V_{r}^{(5)}+[4\,\beta t_{{3}}%
+2(t_{1}+9t_{2}t_{3}+10t_{3}^{3})]V_{r}^{(4)}\\
& \\
& \ \ \ \ \ \ \ +[2\,\beta(t_{{2}}+3t_{3}^{2})+(4t_{2}^{2}+6t_{1}t_{3}%
+30t_{2}t_{3}^{2}+15t_{3}^{4})]V_{r}^{(3)}+\alpha V_{r}^{(-1)}\text{,
\ }r=1,2,3
\end{array}
\]

and $H_{r}=h_{r},$ $,\ r=1,2,$ $H_{3}=h_{3}+t_{2}h_{1}$ and where
$\mathcal{E}_{r}$ are given by (\ref{geo6a}).

Finally, for $m=2$ we let $a_{2n}=\alpha$, $a_{2n-1}=\beta$, $a_{-1}=\gamma$
and $a_{-2}=\delta$ be free, the remaining $a_{i}=0$ and we let $n$ vary. This
leads to the $P_{III}$ hierarchy with its first three members given by:%
\[%
\begin{array}
[c]{ll}%
n=1: & H=\tfrac{1}{2}q^{2}p^{2}-\alpha{\mathrm{e}^{2t}}q^{2}+\beta
\,{\mathrm{e}^{t}}q+\gamma q^{-1}-\delta q^{-2}\\
& \\
n=2: & H_{1}=-\tfrac{1}{2}q_{{1}}\,{p}_{1}^{2}-q_{{2}}p_{{1}}p_{{2}}%
+\alpha(q_{1}^{3}-2q_{1}q_{2})+(2\alpha t_{1}+\beta)(q_{2}-q_{1}^{2})+(\alpha
t_{1}^{2}+\beta t_{1})q_{1}\\
& \\
& \ \ \ \ \ \ \ \ \ +\gamma{\mathrm{e}^{t_{{2}}}}q_{2}^{-1}-\delta
{\mathrm{e}^{2t_{{2}}}}q_{1}q_{2}^{-2}\ \\
& \\
& H_{2}=-\tfrac{1}{2}\,q_{{2}}{p}_{1}^{2}+\tfrac{1}{2}{q}_{2}^{2}\,{p}_{2}%
^{2}+q_{2}p_{2}+\alpha(q_{1}^{2}q_{2}-q_{2}^{2})-(2\alpha t_{1}+\beta
)q_{1}q_{2}+(\alpha t_{1}^{2}+\beta t_{1})q_{2}\\
& \\
& \ \ \ \ \ \ \ \ \ +\gamma{\mathrm{e}^{t_{{2}}}}q_{1}q_{2}^{-1}%
-\delta{\mathrm{e}^{2t_{{2}}}}(q_{1}^{2}-q_{2})q_{2}^{-2}\\
& \\
n=3: & H_{r}=\mathcal{E}_{r}+\alpha V_{r}^{(6)}+(\beta+4\alpha t_{2}%
)V_{r}^{(5)}+[3\beta t_{2}+\,2\alpha(t_{{1}}+3t_{2}^{2})]V_{r}^{(4)}\\
& \\
& \ \ \ \ \ \ \ +[\,\beta(t_{{1}}+3t_{2}^{2})+4\alpha(t_{1}t_{2}+t_{2}%
^{3})]V_{r}^{(3)}+\gamma{\mathrm{e}^{t_{{3}}}V}_{r}^{(-1)}+\delta
{\mathrm{e}^{2t_{{3}}}V}_{r}^{(-2)}\text{, \ }r=1,2,3
\end{array}
\]
and where $\mathcal{E}_{r}$ are given by (\ref{geo7}). The basic separable
potentials $V_{r}^{(\alpha)}$ in the formulas above are given by
(\ref{5d})-(\ref{6a}).

According with Remark 3 in Part I we can rescale the time $t_{n+2-m}$ in
$P_{III}$ above through $t_{n+2-m}^{\prime}=\exp(t_{n+2-m})$ which turns the
$P_{III}-$systems to%
\[%
\begin{array}
[c]{ll}%
n=1: & H=\frac{1}{t^{\prime}}\left(  \tfrac{1}{2}q^{2}p^{2}-\alpha t^{\prime
2}q^{2}+\beta t^{\prime}q+\gamma q^{-1}-\delta q^{-2}\right) \\
& \\
n=2: & H_{1}=-\tfrac{1}{2}q_{{1}}\,{p}_{1}^{2}-q_{{2}}p_{{1}}p_{{2}}%
+\alpha(q_{1}^{3}-2q_{1}q_{2})+(2\alpha t_{1}+\beta)(q_{2}-q_{1}^{2})+(\alpha
t_{1}^{2}+\beta t_{1})q_{1}\\
& \\
& \ \text{\ \ \ \ \ \ \ \ }+\gamma t_{2}^{\prime}q_{2}^{-1}-\delta
t_{2}^{\prime2}q_{1}q_{2}^{-2}\\
& \\
& H_{2}=\frac{1}{t_{2}^{\prime}}\left(  -\tfrac{1}{2}\,q_{{2}}{p}_{1}%
^{2}+\tfrac{1}{2}{q}_{2}^{2}\,{p}_{2}^{2}+q_{2}p_{2}+\alpha(q_{1}^{2}%
q_{2}-q_{2}^{2})-(2\alpha t_{1}+\beta)q_{1}q_{2}+(\alpha t_{1}^{2}+\beta
t_{1})q_{2}\right. \\
& \\
& \text{ \ \ \ \ \ \ }\ \left.  +\gamma t_{2}^{\prime}q_{1}q_{2}^{-1}-\delta
t_{2}^{\prime2}(q_{1}^{2}-q_{2})q_{2}^{-2}\right) \\
& \\
n=3: & H_{r}=\mathcal{E}_{r}+\alpha V_{r}^{(6)}+(\beta+4\alpha t_{2}%
)V_{r}^{(5)}+[3\beta t_{2}+\,2\alpha(t_{{1}}+3t_{2}^{2})]V_{r}^{(4)}\\
& \\
& \text{ \ \ \ \ \ \ }+[\,\beta(t_{{1}}+3t_{2}^{2})+4\alpha(t_{1}t_{2}%
+t_{2}^{3})]V_{r}^{(3)}+\gamma t_{3}{V}_{r}^{(-1)}+2\delta t_{3}{V}_{r}%
^{(-2)}\text{, \ }r=1,2\\
& \\
& H_{3}=\frac{1}{t_{3}^{\prime}}\left(  \mathcal{E}_{3}+\alpha V_{3}%
^{(6)}+(\beta+4\alpha t_{2})V_{3}^{(5)}+[3\beta t_{2}+\,2\alpha(t_{{1}}%
+3t_{2}^{2})]V_{3}^{(4)}\right. \\
& \\
& \text{ \ \ \ \ \ \ \ }\left.  \text{\ }+[\,\beta(t_{{1}}+3t_{2}^{2}%
)+4\alpha(t_{1}t_{2}+t_{2}^{3})]V_{3}^{(3)}+\gamma t_{3}^{\prime}{V}%
_{3}^{(-1)}+\delta t_{3}^{\prime2}{V}_{3}^{(-2)}\right)
\end{array}
\]

Note that while the Frobenius condition (\ref{4}) does not change after the
transformation $t_{n+2-m}^{\prime}=\exp(t_{n+2-m})$,\ in the Lax formulation
we have to replace $\frac{d}{dt_{n+2-m}}$ in (\ref{czasA}) with $t_{n+2-m}%
^{\prime}\frac{d}{dt_{n+2-m}^{\prime}}$.\ For higher $m$ the related
Painlev\'{e} hierarchies start from $n$ higher than one. Besides, the
hierarchies $P_{II}-P_{IV}$ can also be written in the magnetic representation
using the multi-time canonical transformation (\ref{7.3})-(\ref{7.5partI}).
This can't be done for $P_{I}$ which has no magnetic representation, due to
the fact that the map (\ref{mapka}) is not bijective.

\section{Conclusions}

In this article we constructed the isomonodromic Lax representations for all
Frobenius integrable systems constructed in Part I, thus proving that they are
of Painlev\'{e}-type. We also proposed, based on our construction, complete
(in the sense explained in Introduction) Painlev\'{e} $P_{I}-P_{IV}$
hierarchies. An interesting question, which will be a subject of separate
research, is to what extent these hierarchies are related to Painlev\'{e}
hierarchies that can be found in literature. We may expect that the
hierarchies existing in literature are sub-hierarchies within our scheme,
which is thus more general, as our hierarchies contain - for each fixed number
$n$ of degrees of freedom - $n$ different systems satisfying the Frobenius
integrability condition.

\appendix

\setcounter{equation}{0} \renewcommand{\theequation}{A.\arabic{equation}}

\section*{Appendix}

In order to prove Theorem \ref{main} we need a couple of lemmas. In what
follows all formulas are calculated in Vi\`{e}te coordinates (\ref{Viete}).
Below, to shorten the notations, we write simply $L_{ij}(x)$ although in
reality \thinspace$L_{11}$, $L_{21}$ and $L_{22}$ depend in general on $x$,
$t$, $q$ and $p$ while $L_{12}$ depends on $x$ and $q$.

\begin{lemma}
\label{1L}Denote $q_{0}=1$ and $p_{0}=-\sum_{i=1}^{n}q_{i}p_{i}$ (this follows
formally from (\ref{Viete}) for $i=0$). The entries of the Lax matrix $L(x)$
in (\ref{Lspec}) are given explicitly as follows:
\begin{equation}
L_{11}(x)=v(x)-\varphi(x)=-\sum_{k=0}^{n-1}M_{n-k}^{(m)}x^{k}-\sum_{\gamma
=0}^{n+1}d_{\gamma}(t_{1},\dotsc,t_{n})x^{\gamma}, \label{18}%
\end{equation}
where $L_{12}(x)$ is given by
\begin{equation}
L_{12}(x)=u(x)=\sum_{k=0}^{n}q_{n-k}x^{k} \label{12}%
\end{equation}
and finally%
\begin{align}
L_{21}(x)  &  =-x^{m}\left[  \frac{v^{2}(x)x^{-m}}{u(x)}\right]  _{+}%
+2x^{m}\left[  \frac{v(x)\varphi(x)x^{-m}}{u(x)}\right]  _{+}+2x^{m}\left[
\frac{e(t)x^{n}}{u(x)}\right]  _{+}\nonumber\\
&  =-w_{1}(x)+2w_{2}(x)+2w_{3}(x), \label{24}%
\end{align}
where for \ $m=0,\dotsc,n+1$
\begin{align}
w_{1}(x)  &  =\sum_{s=0}^{m-1}\left(  \sum_{k=n-s}^{n}q_{k}\sum_{j=n-m+1}%
^{k+s-m+1}p_{j}p_{n+k+s-2m+2-j}\right)  x^{s}\label{38}\\
&  +\sum_{s=m}^{n-2}\left(  \sum_{k=0}^{n-s-2}q_{k}\sum_{j=k+s-m+2}^{n-m}%
p_{j}p_{n+k+s-2m+2-j}\right)  x^{s},\nonumber
\end{align}%
\begin{equation}
w_{2}(x)=\sum_{s=0}^{m-1}\left(  \sum_{k=0}^{s}d_{k}~p_{n+1-m+s-k}\right)
x^{s}-\sum_{s=m}^{n}\left(  \sum_{k=s+1}^{n+1}d_{k}~p_{n+1-m+s-k}\right)
x^{s}, \label{25}%
\end{equation}%
\begin{equation}
w_{3}(x)=e(t)x^{m},\quad m=0,\dotsc,n+1. \label{w3}%
\end{equation}

\end{lemma}

Note that the above formulas mean that $w_{3}$ depends on times $t_{r}$ only,
$w_{2}$ depends both explicitly on times $t_{r}$ and on coordinates $(q,p)$ on
manifold $\mathcal{M}$, while $w_{1}$ depends only on coordinates $(q,p)$ on
manifold (of course, they all also depend on the spectral parameter $x$).

\begin{proof}
Formula (\ref{w3}) is obvious. We will thus only prove (\ref{25}) as the proof
of (\ref{38}) is similar. From (\ref{L2}), (\ref{L3}) and (\ref{ephi}) we have
that
\[
v(x)=\sum_{i=1}^{n}\frac{u(x)}{x-\lambda_{i}}\frac{\lambda_{i}^{m}\mu_{i}%
}{\Delta_{i}},\quad\varphi(x)=\sum_{k=0}^{n+1}d_{k}x^{k}.
\]
Using these formulas we can write $w_{2}(x)$ in the following form
\[
w_{2}(x)=x^{m}\left[  \frac{v(x)\varphi(x)x^{-m}}{u(x)}\right]  _{+}=x^{m}%
\sum_{k=0}^{n+1}\sum_{i=1}^{n}d_{k}\left[  \frac{x^{k-m}}{x-\lambda_{i}%
}\right]  _{+}\frac{\lambda_{i}^{m}\mu_{i}}{\Delta_{i}}.
\]
Further%
\begin{align*}
x^{m}\left[  \frac{x^{k-m}}{x-\lambda_{i}}\right]  _{+}  &  =x^{m}%
\frac{x^{k-m}-\lambda_{i}^{k-m}}{x-\lambda_{i}}=\frac{x^{k}}{x-\lambda_{i}%
}-\frac{x^{m}}{x-\lambda_{i}}\lambda_{i}^{k-m}\\
&  =\sum_{s=0}^{k-1}\lambda_{i}^{k-s-1}x^{s}+\frac{\lambda_{i}^{k}}%
{x-\lambda_{i}}-\left(  \sum_{s=0}^{m-1}\lambda_{i}^{m-s-1}x^{s}+\frac
{\lambda_{i}^{m}}{x-\lambda_{i}}\right)  \lambda_{i}^{k-m}\\
&  =\sum_{s=0}^{k-1}\lambda_{i}^{k-s-1}x^{s}-\sum_{s=0}^{m-1}\lambda
_{i}^{k-s-1}x^{s}=%
\begin{cases}
-\sum\limits_{s=k}^{m-1}\lambda_{i}^{k-s-1}x^{s} & \text{for $k<m$}\\
\phantom{-}\sum\limits_{s=m}^{k-1}\lambda_{i}^{k-s-1}x^{s} & \text{for $k>m$%
}\\
\phantom{-}0 & \text{for $k=m$}%
\end{cases}
.
\end{align*}
From this result we obtain
\begin{align*}
w_{2}(x)  &  =-\sum_{k=0}^{m-1}\sum_{i=1}^{n}\sum_{s=k}^{m-1}d_{k}\lambda
_{i}^{k-s-1}x^{s}\frac{\lambda_{i}^{m}\mu_{i}}{\Delta_{i}}+\sum_{k=m+1}%
^{n+1}\sum_{i=1}^{n}\sum_{s=m}^{k-1}d_{k}\lambda_{i}^{k-s-1}x^{s}\frac
{\lambda_{i}^{m}\mu_{i}}{\Delta_{i}}\\
&  =-\sum_{k=0}^{m-1}\sum_{s=k}^{m-1}d_{k}\sum_{i=1}^{n}\frac{\lambda
_{i}^{m+k-s-1}\mu_{i}}{\Delta_{i}}x^{s}+\sum_{k=m+1}^{n+1}\sum_{s=m}%
^{k-1}d_{k}\sum_{i=1}^{n}\frac{\lambda_{i}^{m+k-s-1}\mu_{i}}{\Delta_{i}}x^{s}.
\end{align*}
After passing to Vi\`{e}te coordinates (\ref{Viete}) we receive
\[
w_{2}(x)=\sum_{k=0}^{m-1}\sum_{s=k}^{m-1}d_{k}p_{n-m-k+s+1}x^{s}-\sum
_{k=m+1}^{n+1}\sum_{s=m}^{k-1}d_{k}p_{n-m-k+s+1}x^{s}%
\]
and interchanging the order of summation in the above equation we finally
obtain (\ref{25}).
\end{proof}

We will now calculate the entries of $U_{r}$ in (\ref{8}) and (\ref{11}). For
$r\in I_{1}^{m}$, directly from (\ref{8}) we obtain
\begin{align}
(U_{r})_{ij}(x;t)  &  =\frac{1}{2}\left[  \frac{\left[  \frac{u(x)}{x^{n-r+1}%
}\right]  _{+}L_{ij}(x;t)}{u(x)}\right]  _{+}=\frac{1}{2}\left[
\frac{x^{n-r+1}\left[  \frac{u(x)}{x^{n-r+1}}\right]  _{+}L_{ij}%
(x;t)}{x^{n-r+1}u(x)}\right]  _{+}\nonumber\\
&  =\frac{1}{2}\left[  \frac{x^{n-r+1}\left(  \frac{u(x)}{x^{n-r+1}}%
-\frac{r(x)}{x^{n-r+1}}\right)  L_{ij}(x;t)}{x^{n-r+1}u(x)}\right]  _{+}%
=\frac{1}{2}\left[  \frac{u(x)L_{ij}(x;t)-r(x)L_{ij}(x;t)}{x^{n-r+1}%
u(x)}\right]  _{+}\nonumber\\
&  =\frac{1}{2}\left[  \frac{u(x)L_{ij}(x;t)}{x^{n-r+1}u(x)}\right]
_{+}-\frac{1}{2}\left[  \frac{r(x)L_{ij}(x;t)}{x^{n-r+1}u(x)}\right]
_{+}=\frac{1}{2}\left[  \frac{L_{ij}(x;t)}{x^{n-r+1}}\right]  _{+}-\frac{1}%
{2}\left[  \frac{r(x)L_{ij}(x;t)}{x^{n-r+1}u(x)}\right]  _{+}, \label{11a}%
\end{align}
where
\begin{equation}
r(x)=u(x)\bmod x^{n-r+1}=\rho_{r}x^{n-r}+\dotsb+\rho_{n} \label{11c}%
\end{equation}
is a polynomial of degree $n-r$. Thus, if $L_{ij}(x;t)$ is a polynomial of
degree less than $n+1$ then
\begin{equation}
(U_{r})_{ij}(x;t)=\frac{1}{2}\left[  \frac{L_{ij}(x;t)}{x^{n-r+1}}\right]
_{+}. \label{11aa}%
\end{equation}
For $r\in I_{2}^{m}$, directly from (\ref{11}) we obtain
\begin{align}
(U_{r})_{ij}(x;t)  &  =-\frac{1}{2}\left[  \frac{\left[  \frac{u(x)}%
{x^{n-r+1}}\right]  _{-}L_{ij}(x;t)}{u(x)}\right]  _{+}=-\frac{1}{2}\left[
\frac{\frac{u(x)}{x^{n-r+1}}L_{ij}(x;t)-\left[  \frac{u(x)}{x^{n-r+1}}\right]
_{+}L_{ij}(x;t)}{u(x)}\right]  _{+}\nonumber\\
&  =-\frac{1}{2}\left[  \frac{u(x)x^{-n+r-1}L_{ij}(x;t)}{u(x)}\right]
_{+}+\frac{1}{2}\left[  \frac{\left[  \frac{u(x)}{x^{n-r+1}}\right]
_{+}L_{ij}(x;t)}{u(x)}\right]  _{+}\nonumber\\
&  =-\frac{1}{2}\frac{L_{ij}(x;t)}{x^{n-r+1}}+\frac{1}{2}\left[  \frac
{L_{ij}(x;t)}{x^{n-r+1}}\right]  _{+}-\frac{1}{2}\left[  \frac{r(x)L_{ij}%
(x;t)}{x^{n-r+1}u(x)}\right]  _{+}\nonumber\\
&  =-\frac{1}{2}\left[  \frac{L_{ij}(x;t)}{x^{n-r+1}}\right]  _{-}-\frac{1}%
{2}\left[  \frac{r(x)L_{ij}(x;t)}{x^{n-r+1}u(x)}\right]  _{+} \label{11b}%
\end{align}
and again, if $L_{ij}(x;t)$ is a polynomial of degree less than $n+1$ then
\begin{equation}
(U_{r})_{ij}(x;t)=-\frac{1}{2}\left[  \frac{L_{ij}(x;t)}{x^{n-r+1}}\right]
_{-}. \label{11bb}%
\end{equation}
The above results make it possible to calculate the entries of $U_{r}$ in
Vi\`{e}te coordinates.

\begin{lemma}
\label{2L}The entries of $U_{r}$ in (\ref{8}) and (\ref{11}) are as follows.%
\begin{equation}
(U_{r})_{12}(x)=\frac{1}{2}\left[  \frac{u(x)}{x^{n+1-r}}\right]  _{+}%
=\frac{1}{2}\sum_{k=0}^{r-1}q_{r-k-1}x^{k},\qquad r\in\{1\}\cup I_{1}^{m},
\label{13}%
\end{equation}%
\begin{equation}
(U_{r})_{12}(x)=-\frac{1}{2}\left[  \frac{u(x)}{x^{n+1-r}}\right]  _{-}%
=-\frac{1}{2}\sum_{k=1}^{n+1-r}q_{r+k-1}x^{-k},\qquad r\in I_{2}^{m}.
\label{14}%
\end{equation}
For $r\in\{1\}\cup I_{1}^{m}$,%
\begin{equation}
(U_{r})_{11}(x)=\frac{1}{2}\left[  \frac{v(x)}{x^{n+1-r}}\right]  _{+}%
-\frac{1}{2}\left[  \frac{\varphi(x)}{x^{n-r+1}}\right]  _{+}+\frac{1}%
{2}\left[  \frac{r(x)\varphi(x)}{x^{n-r+1}u(x)}\right]  _{+}=\frac{1}{2}%
a_{r}(x)-\frac{1}{2}b_{r}(x)+\frac{1}{2}c_{r}, \label{19}%
\end{equation}
where%
\begin{equation}
a_{r}(x)=-\sum_{k=0}^{r-2}M_{r-k-1}^{(m)}x^{k},\qquad b_{r}(x)=\sum_{k=0}%
^{r}d_{n-r+k+1}x^{k},\qquad c_{r}=q_{r}d_{n+1} \label{20a}%
\end{equation}
while for $r\in I_{2}^{m}$
\begin{equation}
(U_{r})_{11}(x)=-\frac{1}{2}\left[  \frac{v(x)}{x^{n-r+1}}\right]  _{-}%
+\frac{1}{2}\left[  \frac{\varphi(x)}{x^{n-r+1}}\right]  _{-}+\frac{1}%
{2}\left[  \frac{r(x)\varphi(x)}{x^{n-r+1}u(x)}\right]  _{+}=-\frac{1}{2}%
a_{r}(x)+\frac{1}{2}b_{r}(x)+\frac{1}{2}c_{r}, \label{21}%
\end{equation}
where%
\begin{equation}
a_{r}(x)=\sum_{k=1}^{n+1-r}M_{r+k-1}^{(m)}x^{-k},\qquad b_{r}(x)=\sum
_{k=1}^{n+1-r}d_{n-r-k+1}x^{-k},\qquad c_{r}=q_{r}d_{n+1}. \label{22a}%
\end{equation}
Further, denoting (in accordance with (\ref{24}), (\ref{11a}) and
(\ref{11b}))
\begin{equation}
(U_{r})_{21}=-(U_{r})_{21}(w_{1}(x))+(U_{r})_{21}(2w_{2}(x))+(U_{r}%
)_{21}(2w_{3}(x)), \label{Ur21}%
\end{equation}
(where $(U_{r})_{21}(w_{1}(x))$ denotes the part of $(U_{r})_{21}$ generated
by the first term in (\ref{24}) and so on) we receive for $r\in\{1\}\cup
I_{1}^{m}$,
\begin{align}
(U_{r})_{21}(w_{1}(x))  &  =\frac{1}{2}\left[  \frac{w_{1}(x)}{x^{n-r+1}%
}\right]  _{+}\nonumber\\
&  =\frac{1}{2}\sum_{s=n-r+1}^{n-2}\left(  \sum_{k=0}^{n-s-2}q_{k}%
\sum_{j=k+s-m+2}^{n-m}p_{j}p_{n+k+s-2m+2-j}\right)  x^{s-n+r-1}, \label{40}%
\end{align}
and for $r\in I_{2}^{m}$
\begin{align}
(U_{r})_{21}(w_{1}(x))  &  =-\frac{1}{2}\left[  \frac{w_{1}(x)}{x^{n-r+1}%
}\right]  _{-}\nonumber\\
&  =-\frac{1}{2}\sum_{s=0}^{m-1}\left(  \sum_{k=n-s}^{n}q_{k}\sum
_{j=n-m+1}^{k+s-m+1}p_{j}p_{n+k+s-2m+2-j}\right)  x^{s-n+r-1}\nonumber\\
&  \quad{}-\frac{1}{2}\sum_{s=m}^{n-r}\left(  \sum_{k=0}^{n-s-2}q_{k}%
\sum_{j=k+s-m+2}^{n-m}p_{j}p_{n+k+s-2m+2-j}\right)  x^{s-n+r-1}. \label{41}%
\end{align}
For $r\in\{1\}\cup I_{1}^{m}$
\begin{align}
(U_{r})_{21}(2w_{2}(x))  &  =\left[  \frac{w_{2}(x)}{x^{n+1-r}}\right]
_{+}\nonumber\\
&  =-\sum_{s=0}^{r-1}\left(  \sum_{k=1}^{r-s}d_{n+1-r+s+k}~p_{n+1-m-k}\right)
x^{s}, \label{28}%
\end{align}
and for $r\in I_{2}^{m}$
\begin{align}
(U_{r})_{21}(2w_{2}(x))  &  =-\left[  \frac{w_{2}(x)}{x^{n-r+1}}\right]
_{-}\nonumber\\
&  =-\sum_{s=1}^{n+1-r}\left(  \sum_{k=0}^{n+1-r-s}d_{n+1-r-s-k}%
\ p_{n+1-m+k}\right)  x^{-s}. \label{28a}%
\end{align}
Finally%
\begin{align}
(U_{r})_{21}(2w_{3}(x))  &  =e(t)\delta_{r,n-m+1}\text{ \ \ for }%
m=0,\ldots,n\label{Ur21w3}\\
(U_{r})_{21}(2w_{3}(x))  &  =-e(t)q_{r}+e(t)x\delta_{1,r}\text{ \ for
}m=n+1.\nonumber
\end{align}

\end{lemma}

The formulas (\ref{13}) and (\ref{14}) follow directly from (\ref{11aa}) and
(\ref{11bb}), respectively. The formulas (\ref{19})-(\ref{22a}) are easy
consequences of (\ref{18}). Finally, the formulas (\ref{40})-(\ref{Ur21w3})
follow from (\ref{38})-(\ref{w3}).

\begin{lemma}
\label{3L}Poisson brackets $\left\{  L_{ij},W_{r}\right\}  $ in Vi\`{e}te
coordinates become%
\begin{equation}
\left\{  L_{11},W_{r}\right\}  =\left\{  v(x),W_{r}\right\}  =\left\{
\begin{array}
{c}%
-\sum\limits_{k=1}^{r-2}kM_{r-k-1}^{(m)}x^{m+k-1}\text{ \ for }r\in I_{1}%
^{m}\\
\\
\sum\limits_{k=1}^{n+1-r}kM_{r+k-1}^{(m)}x^{m-k-1}\text{ \ for }r\in I_{2}^{m}%
\end{array}
\right.  \label{L11W}%
\end{equation}%
\begin{equation}
\left\{  L_{12},W_{r}\right\}  =\left\{  u(x),W_{r}\right\}  =\left\{
\begin{array}
{c}%
\sum\limits_{k=1}^{r-1}kq_{r-k-1}x^{m+k-1}\text{ \ for }r\in I_{1}^{m}\\
\\
\sum\limits_{k=1}^{n+1-r}kq_{r+k-1}x^{m-k-1}\text{ \ for }r\in I_{2}^{m}%
\end{array}
\right.  \label{L12W}%
\end{equation}
and, due to (\ref{24}),%
\[
\left\{  L_{21},W_{r}\right\}  =-\left\{  w_{1}(x),W_{r}\right\}  +2\left\{
w_{2}(x),W_{r}\right\}  +2\left\{  w_{3}(x),W_{r}\right\}
\]
with%
\begin{equation}
\left\{  w_{1}(x),W_{r}\right\}  =\left\{
\begin{array}
[c]{l}%
\sum\limits_{s=n-r+2}^{n-2}\sum\limits_{k=0}^{n-2-s}\sum\limits_{j=k+s-m+2}%
^{n-m}(s-n+r-1)q_{k}p_{j}p_{n+k+s-2m+2-j}x^{m+s-n+r-2}\text{ \ \ for }r\in
I_{1}^{m}\\
\\%
\begin{array}
[c]{l}%
-\sum\limits_{s=0}^{m-1}\sum\limits_{k=n-s}^{n}\sum\limits_{j=n-m+1}%
^{k+s-m+1}(s-n+r-1)q_{k}p_{j}p_{n+k+s-2m+2-j}x^{m+s-n+r-2}\\
-\sum\limits_{s=m}^{n-r}\sum\limits_{k=0}^{n-s-2}\sum\limits_{j=k+s-m+2}%
^{n-m}(s-n+r-1)q_{k}p_{j}p_{n+k+s-2m+2-j}x^{m+s-n+r-2}%
\end{array}
\ \ \text{for }r\in I_{2}^{m}%
\end{array}
\right.  \label{L21Wa}%
\end{equation}%
\begin{equation}
\left\{  w_{2}(x),W_{r}\right\}  =\left\{
\begin{array}
{c}%
x^{m}\sum\limits_{s=1}^{r-1}\left(  \sum\limits_{k=1}^{r-s}k\ d_{n+1-r+s+k}%
\ p_{n+1-m-k}\right)  x^{s-1}\text{\ for }r\in I_{1}^{m}\\
\\
-x^{m}\sum\limits_{s=0}^{n+1-r}\left(  \sum\limits_{k=0}^{n+1-r-s}%
k~d_{n+1-r-s-k}~p_{n+1-m+k}\right)  x^{-s-1}\text{ for }r\in I_{2}^{m}%
\end{array}
\right.  \label{L21Wb}%
\end{equation}
while $\left\{  w_{3}(x),W_{r}\right\}  =0$.
\end{lemma}

All the equations in Lemma \ref{3L} follow from Lemma \ref{1L} and from the
explicit form of perturbations terms $W_{r}$ in Vi\`{e}te coordinates
(\ref{kil5}). For example, in order to prove (\ref{L12W}) we note that for
$r\in I_{1}^{m}$ and due to (\ref{kil5})%

\begin{align*}
\left\{  u(x),W_{r}\right\}   &  =\sum\limits_{s=n-m-r+2}^{n-m}x^{n-s}\left\{
q_{s},W_{r}\right\}  =\sum\limits_{s=n-m-r+2}^{n-m}(n+1-m-s)x^{n-s}%
q_{m+r-n-2+s}=\\
&  =\sum\limits_{k=1}^{r-1}kq_{r-k-1}x^{m+k-1},
\end{align*}
where the last equality is obtained through the substitution $k=n+1-m-s$. The
case $r\in I_{2}^{m}$ is treated similarly.

\begin{proof}
(of Theorem \ref{main}) The equation (\ref{6}) reads componentwise as%
\begin{equation}
\frac{\partial L_{ij}}{\partial t_{r}}+\left\{  L_{ij},H_{r}^{B}\right\}
=[\overline{U}_{r},L]_{ij}+2x^{m}\frac{\partial}{\partial x}\left(
\overline{U}_{r}\right)  _{ij},\quad r=1,\dotsc,n, \label{sklad}%
\end{equation}
with $H_{r}^{B}$ given by (\ref{7b}) and (\ref{7c}) (we remind also that
$h_{r}^{B}=h_{r}+W_{r}$ for all $r$). Thus, for $r=1,\dotsc,\kappa_{1}$ and
for $r=n-\kappa_{2}+1,\dotsc,n$ the equation (\ref{sklad}) reads%
\[
\frac{\partial L_{ij}}{\partial t_{r}}+\left\{  L_{ij},h_{r}\right\}
+\left\{  L_{ij},W_{r}\right\}  =[U_{r},L]_{ij}+2x^{m}\frac{\partial}{\partial
x}\left(  U_{r}\right)  _{ij},\quad r=1,\dotsc,n,
\]
and since due to (\ref{eq:3}) $\left\{  L_{ij},h_{r}\right\}  =[U_{r},L]_{ij}$
(this follows from the fact that $h_{r}$ are St\"{a}ckel Hamiltonians when the
remaining times $t_{1},...,t_{r-1},t_{r+1},\ldots,t_{n}$ are considered as
parameters), we obtain that for $r=1,\dotsc,\kappa_{1}$ and for $r=n-\kappa
_{2}+1,\dotsc,n$ the equation (\ref{sklad}) becomes
\begin{equation}
\frac{\partial L_{ij}}{\partial t_{r}}+\left\{  L_{ij},W_{r}\right\}
=2x^{m}\frac{\partial}{\partial x}\left(  U_{r}\right)  _{ij},\quad
r=1,\dotsc,n\text{.} \label{ab}%
\end{equation}
Similarly, and again due to the fact that $h_{r}$ are St\"{a}ckel
Hamiltonians, the equation (\ref{sklad}) for $r=\kappa_{1}+1,\ldots,n-m+1$
reads%
\begin{equation}
\frac{\partial L_{ij}}{\partial t_{r}}+\sum\limits_{s=1}^{r}\zeta
_{r,s}\left\{  L_{ij},W_{s}\right\}  =2x^{m}\sum\limits_{s=1}^{r}\zeta
_{r,s}\frac{\partial}{\partial x}\left(  U_{s}\right)  _{ij},\quad
r=1,\dotsc,n\text{,} \label{nab1}%
\end{equation}
while for $r=n-m+2,\dotsc,n-\kappa_{2}$ it becomes%
\begin{equation}
\frac{\partial L_{ij}}{\partial t_{r}}+\sum\limits_{s=0}^{n-r}\zeta
_{r,r+s}\left\{  L_{ij},W_{s}\right\}  =2x^{m}\sum\limits_{s=0}^{n-r}%
\zeta_{r,r+s}\frac{\partial}{\partial x}\left(  U_{r+s}\right)  _{ij},\quad
r=1,\dotsc,n\text{.} \label{nab2}%
\end{equation}
Thus, our task is to show that (\ref{ab})-(\ref{nab2}) do follow from the
assumptions of our theorem. We will prove it componentwise. For the component
$(1,2)$ we have $\frac{\partial L_{12}}{\partial t_{r}}=0$, so by
(\ref{L12W})\ and (\ref{13})-(\ref{14}) the equations (\ref{ab})-(\ref{nab2})
are identically satisfied. Note that in case when $L_{ij}$ does not depend on
$t_{r}$ all the equations (\ref{ab})-(\ref{nab2}) coincide. Next, since
$\frac{\partial v(x)}{\partial t_{r}}=0$ while $\varphi(x)$ does not depend on
coordinates on manifold $\mathcal{M}$, the $(1,1)$-component in equations
(\ref{ab})-(\ref{nab2}) splits into two parts. For the $v(x)$-part it reduces
to the equation%
\[
\left\{  v(x),W_{r}\right\}  =2x^{m}\frac{\partial}{\partial x}\left(
U_{r}\right)  _{11},\quad r=1,\dotsc,n
\]
and by comparing (\ref{L11W}) with (\ref{19})-(\ref{22a}) we see that it is
identically satisfied. Further, for $r\in\{1,\dotsc,\kappa_{1}\}$, the
$\varphi(x)$-part becomes
\[
\frac{\partial\varphi(x)}{\partial t_{r}}=x^{m}\frac{\partial b_{r}%
(x)}{\partial x},
\]
which reads explicitly as the polynomial in $x$ equation
\[
\sum_{\gamma=0}^{n+1}\frac{\partial d_{\gamma}}{\partial t_{r}}x^{\gamma}%
=\sum_{j=1}^{r}jd_{n-r+j+1}x^{j+m-1}=\sum_{\gamma=m}^{m+r-1}(\gamma
-m+1)d_{n-m-r+\gamma+2}x^{\gamma}%
\]
and comparing the coefficients at $x^{\gamma}$ we obtain:
\begin{align}
\frac{\partial d_{\gamma}}{\partial t_{r}}  &  =0,\quad\gamma\neq
m,\dotsc,m+r-1,\nonumber\\
\frac{\partial d_{\gamma}}{\partial t_{r}}  &  =(\gamma-m+1)d_{n-m+2+\gamma
-r},\quad\gamma=m,\dotsc,m+r-1, \label{7.9b}%
\end{align}
i.e. exactly the PDE's (\ref{4o})-(\ref{4a}) that are satisfied by the
assumptions of the theorem. Similarly, for $r\in\{n+1-\kappa_{2},\dotsc,n\}$,
the $\varphi(x)$-part becomes
\[
\frac{\partial\varphi(x)}{\partial t_{r}}=-x^{m}\frac{\partial b_{r}%
(x)}{\partial x}%
\]
which is equivalent to
\[
\sum_{\gamma=0}^{n+1}\frac{\partial d_{\gamma}}{\partial t_{r}}x^{\gamma}%
=\sum_{j=1}^{n+1-r}jd_{n-r-j+1}x^{m-j-1}=\sum_{\gamma=r-(n-m+2)}%
^{m-2}(m-\gamma-1)d_{n-m-r+\gamma+2}x^{\gamma}%
\]
which in turn is equivalent to
\begin{align*}
\frac{\partial d_{\gamma}}{\partial t_{r}}  &  =0,\quad\gamma\neq
r-(n-m+2),\ldots,m-2,\\
\ \frac{\partial d_{\gamma}}{\partial t_{r}}  &  =-(\gamma-m+1)d_{n-m+2+\gamma
-r},\quad\gamma=r-(n-m+2),\ldots,m-2,
\end{align*}
which recover exactly the PDE's (\ref{4r})-(\ref{4d}) and are thus satisfied
by the assumptions of the theorem. Further, for $r\in\{\kappa_{1}%
,\dotsc,n-m+1\}$ the $\varphi(x)$-part becomes
\[
\frac{\partial\varphi(x)}{\partial t_{r}}=x^{m}\sum_{j=1}^{r}\zeta_{r,j}%
\frac{\partial b_{j}(x)}{\partial x},
\]
which is
\begin{align*}
\sum_{\gamma=0}^{n+1}\frac{\partial d_{\gamma}}{\partial t_{r}}x^{\gamma}  &
=\sum_{j=1}^{r}\zeta_{r,j}\sum_{k=1}^{j}kd_{n-j+k+1}x^{k+m-1}\\
&  =\sum_{j=1}^{r}\zeta_{r,j}\sum_{\gamma=m}^{m+j-1}(\gamma
-m+1)d_{n-m-j+\gamma+2}x^{\gamma}\\
&  =\sum_{\gamma=m}^{m+r-1}(\gamma-m+1)\left(  \sum_{j=\gamma-m+1}^{r}%
\zeta_{r,j}(t_{1},\dotsc,t_{r-1})d_{n-m+2+\gamma-j}\right)  x^{\gamma}%
\end{align*}
which is equivalent to
\begin{align*}
\frac{\partial d_{\gamma}}{\partial t_{r}}  &  =0,\quad\gamma\neq
m,\dotsc,m+r-1\\
\frac{\partial d_{\gamma}}{\partial t_{r}}  &  =(\gamma-m+1)\sum
_{j=\gamma-m+1}^{r}\zeta_{r,j}(t_{1},\dotsc,t_{r-1})d_{n-m+2+\gamma-j}%
,\quad\gamma=m,\dotsc,m+r-1.
\end{align*}
The above equations are exactly the PDE's (\ref{4n})-(\ref{4b}) and as such
are satisfied by the assumptions of the theorem. For $r\in\{n-m+2,\dotsc
,n-\kappa_{2}\}$ the $\varphi(x)$-part becomes
\[
\frac{\partial\varphi(x)}{\partial t_{r}}=-x^{m}\sum_{j=0}^{n-r}\zeta
_{r,r+j}\frac{\partial b_{r+j}(x)}{\partial x},
\]
that is
\begin{align*}
\sum_{\gamma=0}^{n+1}\frac{\partial d_{\gamma}}{\partial t_{r}}x^{\gamma}  &
=\sum_{j=0}^{n-r}\zeta_{r,r+j}\sum_{k=1}^{n+1-r-j}kd_{n+1-r-k-j}x^{m-k-1}\\
&  =\sum_{j=0}^{n-r}\zeta_{r,r+j}\sum_{\gamma=r-(n-m+2-j)}^{m-2}%
(m-\gamma-1)d_{n-m-r-j+\gamma+2}x^{\gamma}\\
&  =\sum_{\gamma=r-(n-m+2)}^{m-2}(m-\gamma-1)\left(  \sum_{j=0}%
^{n-m+2-r+\gamma}\zeta_{r,r+j}d_{n-m+2-r-j+\gamma}\right)  x^{\gamma}%
\end{align*}
which is equivalent to
\begin{align*}
\frac{\partial d_{\gamma}}{\partial t_{r}}  &  =0,\quad\gamma\neq
r-(n-m+2),\ldots,m-2,\\
\ \frac{\partial d_{\gamma}}{\partial t_{r}}  &  =-(\gamma-m+1)\sum
_{j=0}^{n-m+2+\gamma-r}\zeta_{r,r+j}(t_{r+1},\dotsc,t_{n})d_{n-m+2+\gamma
-r-j},\quad\gamma=r-(n-m+2),\ldots,m-2.
\end{align*}
The above equations are exactly the PDE's (\ref{ziuta})-(\ref{4c}) and thus
are satisfied by the assumptions of the theorem. Thus, the $(1,1)$-component
of the equations (\ref{ab})-(\ref{nab2}) is satisfied. Let us finally turn to
the $(2,1)$-component of (\ref{ab})-(\ref{nab2}). Since by (\ref{w3}) $w_{3}$
does not depend on coordinates on $\mathcal{M}$, the $w_{3}$-part of the
$(2,1)$-component of equations (\ref{ab})-(\ref{nab2}) read%
\begin{equation}
\frac{\partial w_{3}(x)}{\partial t_{r}}=x^{m}\frac{\partial}{\partial
x}(U_{r})_{21}(2w_{3}(x)) \label{27c}%
\end{equation}
which is identically satisfied due to (\ref{Ur21w3}) and (\ref{4e}). Next we
prove the $w_{2}$-part of the $(2,1)$-component of equations (\ref{ab}%
)-(\ref{nab2}). Since $w_{2}$ depends both on times $t_{r}$ and on the
coordinates on $\mathcal{M}$, we have to consider four separate cases.
Consider first the case $r=1,\dotsc,\kappa_{1}$. Then the $w_{2}$-part of the
$(2,1)$-component of (\ref{ab}) is%
\begin{equation}
\frac{\partial w_{2}(x)}{\partial t_{r}}+\{w_{2}(x),W_{r}\}=x^{m}%
\frac{\partial}{\partial x}(U_{r})_{21}(2w_{2}(x)) \label{pokaz1}%
\end{equation}
Differentiating (\ref{25}) w.r.t $t_{r}$ with the help of PDE's (\ref{4o}%
)-(\ref{4a}) yields that for $m=0,\ldots,n+1$%
\[
\frac{\partial w_{2}(x)}{\partial t_{r}}=-\sum_{s=m}^{m+r-2}\left(
\sum_{k=s+1}^{m+r-1}(k-m+1)d_{n+2-m-r+k}\ p_{n+1-m+s-k}\right)  x^{s}%
\]
which after the reparametrization of indices $s\rightarrow s+m-1$ and then
$k\rightarrow k+s+m-1$ reads
\begin{equation}
\frac{\partial w_{2}(x)}{\partial t_{r}}=-x^{m}\sum_{s=1}^{r-1}\left(
\sum_{k=1}^{r-s}(k+s)d_{n+1-r+s+k}\ p_{n+1-m-k}\right)  x^{s-1}, \label{31}%
\end{equation}
Combining this result with (\ref{28}) and (\ref{L21Wb}) yields (\ref{pokaz1}).
Consider now the case $r=\kappa_{1}+1,\ldots,n-m+1$. Then the $w_{2}$-part of
the $(2,1)$-component of (\ref{nab1}) is%
\begin{equation}
\frac{\partial w_{2}(x)}{\partial t_{r}}+\sum\limits_{s=1}^{r}\zeta
_{r,s}\{w_{2}(x),W_{s}\}=x^{m}\sum\limits_{s=1}^{r}\zeta_{r,s}\frac{\partial
}{\partial x}(U_{s})_{21}(2w_{2}(x)) \label{pokaz2}%
\end{equation}
Differentiating (\ref{25}) w.r.t $t_{r}$ with the help of PDE's (\ref{4n}%
)-(\ref{4b}) yields that for $m=0,\ldots,n$%
\[
\frac{\partial w_{2}(x)}{\partial t_{r}}=-\sum_{s=m}^{m+r-2}\sum
_{k=s+1}^{m+r-1}\left(  \sum_{j=k-m+1}^{r}(k-m+1)\zeta_{r,j}d_{n+2-m-j+k}%
\right)  \ p_{n+1-m+s-k}x^{s}%
\]
Substituting this result, (\ref{28}) and (\ref{L21Wb}) to (\ref{pokaz2})
yields%
\begin{align*}
&  \sum_{s=m}^{m+r-2}\sum_{k=s+1}^{m+r-1}\sum_{j=k-m+1}^{r}\zeta
_{r,j}(k-m+1)d_{n+2-m-j+k}\ p_{n+1-m+s-k}x^{s}\\
&  -\sum\limits_{s=1}^{r}\zeta_{r,s}x^{m}\sum\limits_{l=1}^{s-1}\left(
\sum\limits_{k=1}^{s-l}k\ d_{n+1-s+l+k}\ p_{n+1-m-k}\right)  x^{l-1}\\
&  =x^{m}\sum\limits_{s=1}^{r}\zeta_{r,s}\sum_{l=0}^{s-1}l\left(  \sum
_{k=1}^{s-l}d_{n+1-s+l+k}~p_{n+1-m-k}\right)  x^{l-1}%
\end{align*}
that can be shown, through suitable changes of summation indices and careful
changes of the summation order, to be identically satisfied. The proof of the
two remaining cases when $r\in I_{2}^{m}$ is similar. Let us now turn into the
$w_{1}$-part of the $(2,1)$-component of (\ref{ab})-(\ref{nab2}). Since
$w_{1}$ does not depend on times $t_{r}$, this part reads
\begin{equation}
\{w_{1}(x),W_{r}\}=2x^{m}\frac{\partial}{\partial x}(U_{r})_{21}(w_{1}(x)).
\label{44}%
\end{equation}
Comparing (\ref{40}) and (\ref{L21Wa}) we get that (\ref{44}) is satisfied.
Similarly we prove the case $r\in I_{2}^{m}$. This concludes the proof.
\end{proof}

\end{document}